\documentclass[12pt,a4paper]{amsart}
\usepackage{eurosym}
\usepackage{geometry}
\usepackage{amssymb}
\usepackage{amsmath,bm}
\usepackage{amsfonts}
\usepackage{graphicx,epsfig,lscape}
\usepackage{color}
\usepackage[dvipsnames]{xcolor}
\usepackage{setspace}
\usepackage{url}
\usepackage{upgreek}
\usepackage{commath}
\usepackage{enumitem}
\usepackage{natbib}
\usepackage{hyperref}
\hypersetup{
  colorlinks,
  citecolor=Blue,
  linkcolor=Maroon,
  urlcolor=Blue}

\usepackage{indentfirst}
\usepackage{cleveref}
\usepackage{caption}
\usepackage{bbm}
\usepackage{amsthm}
\usepackage{comment}
\usepackage{subcaption}
\usepackage{setspace}
\usepackage[multiple]{footmisc}

\DeclareMathOperator*{\argmax}{argmax}

\newtheorem{theorem}{Theorem}

\newtheorem{assumption}{Assumption}

\newtheorem{claim}{Claim}

\newtheorem{conjecture}{Conjecture}
\newtheorem{corollary}{Corollary}

\newtheorem{definition}{Definition}
\newtheorem{example}{Example}

\newtheorem{lemma}{Lemma}

\newtheorem{proposition}{Proposition}
\newtheorem{remark}{Remark}

\newtheorem*{thm2s}{Theorem $\mathbf{2^*}$}
\newcommand{\rv}[1]{\noindent \normalfont \textcolor{violet}{[revision: #1]}}

\makeatletter
\renewenvironment{proof}[1][\proofname]{%
  \par\pushQED{\qed}\normalfont%
  \topsep6\p@\@plus6\p@\relax
  \trivlist\item[\hskip\labelsep\bfseries#1\@addpunct{.}]%
  \ignorespaces
}{%
  \popQED\endtrivlist\@endpefalse
}
\makeatother

\newcommand{\supp}{\text{supp}}

\newcommand{\sprg}{\mathcal{O}}
\newcommand{\group}{G}
\newcommand{\piv}{w_{\texttt{piv}}}
\newcommand{\mkvstg}{\alpha^{\mathrm{m}}}
\newcommand{\scale}{\varepsilon}
\newcommand{\cost}{c}
\newcommand{\quadvar}{\text{qv}}
\newcommand{\cont}{C}
\newcommand{\stpbf}{S}
\newcommand{\lbp}{\text{L}}
\newcommand{\ubp}{\text{U}}
\newcommand{\lbpil}{{\underline{\lbp}{}_{-i}}}
\newcommand{\lbpih}{{\overline{\lbp}{}_{-i}}}
\newcommand{\ubpih}{{\overline{\ubp}{}_{-i}}}
\newcommand{\ubpil}{{\underline{\ubp}{}_{-i}}}

\makeatletter
\def\paragraph{\@startsection{paragraph}{4}%
  \z@\z@{-\fontdimen2\font}%
  {\normalfont\itshape}}
\makeatother

\title{Collective Sampling: An Ex Ante Perspective
}
\author{Yangfan Zhou}
\thanks{%
Department of Economics, Columbia University. Email: \href{mailto:yz3905@columbia.edu}%
{yz3905@columbia.edu}.\\ I am deeply indebted to Navin Kartik and Laura Doval for their continuing guidance and support throughout the development of this paper. I am also grateful to C\'esar Barilla, Yeon-Koo Che, Elliot Lipnowski, Qingmin Liu, Tianhao Liu, Yu Fu Wong and other participants at Columbia University Micro Theory Colloquium for their valuable comments and suggestions. Any errors are my own.}

\begin{document}
\maketitle
\begin{abstract}
\thispagestyle{empty}
I study collective dynamic information acquisition. Players decide when to stop sequential sampling via a collective stopping rule, which specifies decisive coalitions that can terminate information acquisition upon agreement. 
I develop a methodology to characterize equilibria using an ex ante perspective. 
Instead of stopping strategies, players choose distributions over posterior beliefs subject to majorization constraints.  
Equilibrium sampling regions are characterized via a fixed-point argument based on concavification. 
Collective sampling generates learning inefficiencies 
and having more decisive coalitions 
typically reduces learning. I apply the model to committee search and competition in persuasion.
\end{abstract}

\pagenumbering{Alph} 

\pagenumbering{arabic}%
\setcounter{page}{0}%
\newpage

\section{Introduction}

Committees, teams, and boards making important decisions often face a dilemma: when should they stop gathering information and act? That collective decision itself is riddled with conflicts of interest. Whether a group member wants to continue gathering information depends on what has already been learned, 
their preferences over group actions, their expectations about how the group will act with different kinds of further information, 
and crucially, the collective decision rule for information collection. While dynamic information acquisition has been well studied for individuals \citep*[e.g.,][]{wald1947foundations,arrow1949bayes,morris2019wald,zhong2022optimal}, its group counterpart remains underexplored.

Consider, for example, a company's board members deciding on an acquisition. Before making decisions, they can gather relevant information through market research. The board members collectively determine how detailed the research should be or whether additional data is needed. 
How do conflicts of interest and the collective decision rule interact to shape information acquisition and the final decision? For instance, imagine a scenario where stopping information collection requires only a simple majority vote, but the chairperson holds veto power. 
How does the chairperson's preference influence information acquisition? And how do others adjust their strategies in anticipation of a veto?


To address these questions, this paper proposes a model of collective dynamic information acquisition with a focus on stopping decisions. I extend \citeauthor{wald1947foundations}'s (\citeyear{wald1947foundations}) sequential sampling model to a strategic situation. Players collectively decide when to stop acquiring costly signals about a binary state of the world. When information acquisition ends, players get terminal payoffs that depend on their common posterior belief about the state.\footnote{The terminal payoffs can come from equilibrium payoffs in a follow-up game where players take individual or collective actions and the state is payoff-relevant.} To focus on stopping decisions, the model features public and exogenous signals so that players can only choose when to stop acquiring information but not what information to acquire. 

Collective bargaining over stopping is modeled by a reduced-form protocol, represented as a collective stopping rule dubbed \emph{$\mathcal{\group}$-collective stopping}. Formally, $\mathcal{\group}$-collective stopping specifies all the \emph{decisive coalitions} of players $\group\in\mathcal{\group}$ that can terminate collective information acquisition when they reach an agreement (one player can belong to different coalitions): 
as long as every player within any decisive coalition stops, information acquisition ends. Examples include unanimity, (super-)majority, and dictatorship rules.
I focus on pure-strategy Markov perfect equilibria where players' stopping decisions only depend on their current beliefs. 

To convey the main ideas and analytical tools transparently, I begin with two simple rules in \Cref{sect-twoplayer}: \emph{unilateral stopping} where information acquisition can be unilaterally terminated by either player, and \emph{unanimous stopping} where the termination requires unanimity of all players. To ease the exposition, the analysis is focused on two players. I defer the full analysis with many players and general collective stopping rules to \Cref{discussion}. 

For unilateral and unanimous stopping, \Cref{MPE} establishes that 
players' constrained stopping problems can be reformulated into static problems of 
information acquisition where they choose distributions over posterior beliefs subject to majorization constraints. This reformulation is based on the idea of \citet{morris2019wald} in the single-agent case that 
choosing a stopping strategy is equivalent to choosing a Bayes plausible belief distribution. 
In the strategic situation, unilateral stopping only allows players to choose among mean-preserving contractions (MPC) of the distribution induced by the other player's stopping strategy (\Cref{MPC}). Under unanimous stopping, they can only choose mean-preserving spreads (MPS) with support in the other player's stopping region (\Cref{MPS}). 

The ex ante reformulation enables an characterization of equilibrium sampling regions (henceforth equilibria) in \Cref{concave} using the concavification method \citep*[see][]{aumann1995repeated,kamenica2011bayesian}. 
Relevant concave closures 
are identified to figure out players' best responses under different rules. Equilibria are fixed-points of the best responses, or graphically, sampling regions by using which the relevant concave closures of both players are attained simultaneously. I show in \Cref{lattice} that equilibria under both rules form a semi-lattice with the set inclusion order: Their union of two equilibria is also an equilibrium. 

The semi-lattice structure facilitates comparative statics. \Cref{pareto} shows that collective information acquisition is in general Pareto inefficient: Players learn too little under unilateral stopping and possibly too much under unanimous stopping. 
\Cref{misalign-max} further shows that players' preference misalignment amplifies these learning inefficiencies. 

However, players do not necessarily learn too much under unanimous stopping. In fact, any equilibrium with strictly more learning than both players' individual optima has to rely on miscoordination (\Cref{control-sharing}). This is due to a control-sharing effect: Because players have to share the control over stopping, which weakly lowers their option values of waiting, they tend to stop earlier than if they were to make the decision alone.\footnote{This phenomenon is also identified 
by \citet{strulovici2010learning}, \citet{albrecht2010search}, and \citet{chan2018deliberating} in the contexts of collective experimentation, search, and deliberation.}

\Cref{discussion} extends the analysis to many players and $\mathcal{\group}$-collective stopping. In general, players face both majorization constraints in terms of MPC and MPS. Equilibria can still be characterized via concavification (\Cref{thm-g-collective}). A particularly interesting question is how learning is affected by different collective stopping rules. \Cref{cs-general} shows that adding new decisive coalitions (or reducing players in existing coalitions) typically reduces learning. 
This result also speaks to the comparison between different quota rules and the effect of giving a player more votes or veto power. 

To operationalize the model, I consider two economic applications: committee search and dynamic compeitition in persuasion.\footnote{For interested readers, the previous version of this paper contains another application on dynamic persuasion between a sender and a receiver; see \citet{zhou2023collective}.}
First, I study a committee search problem in which committee members facing two alternatives dynamically acquire relevant information under a $\mathcal{\group}$-collective stopping rule. 
The binary choice enables a significant simplification of the problem: 
$\mathcal{\group}$-collective stopping within the committee can be reduced to the interaction between only two pivotal players, one for each alternative (\Cref{prop:strong}). 
As an implication, in symmetric committees, quota rules generate inefficiencies that are amplified by committee diversity.

Then, I consider another application on competition in persuasion where multiple senders compete to persuade a receiver under unanimous stopping. As a corollary of the results in the general model, increased competition leads to more information revelation. However, compared to the static benchmark in \citet{gentzkow2017competition}, dynamic competition in general features less information revelation due to the control-sharing effect.


The paper is organized as follows. 
\Cref{framework} introduces the baseline model. \Cref{sect-twoplayer} focuses on unilateral and unanimous stopping. 
\Cref{discussion} generalizes the analysis to many players and general collective stopping rules. \Cref{application} applies the framework to committee search and dynamic competition in persuasion. \Cref{conclusion} concludes. Proofs for the baseline model are relegated to \Cref{proofs}, while those for the applications can be found in Supplementary \Cref{proofs-app}.

\subsection{Related Literature}
\label{lit}
This paper contributes to the literature on collective information acquisition, strategic sequential sampling, as well as Bayesian persuasion.

\vspace{1mm}
\paragraph{Collective Information Acquisition} 
This paper is related to the literature on collective decisions on dynamic information acquisition. \citet{strulovici2010learning} studies dynamic collective experimentation 
in a two-armed bandit framework, showing that the control-sharing effect creates externality 
and reduces equilibrium experimentation below the socially efficient level. Other works adopt a collective search approach where a committee votes on when to stop searching for alternatives 
\citep*[e.g.,][]{albrecht2010search,compte2010bargaining,moldovanu2013specialization,titova2019collaborative}. The closest related are \citet*{chan2018deliberating} and \citet{anesi2023deciding} which use the sequential sampling model \citep{wald1947foundations} to study 
committee deliberation between two alternatives. 

Albeit using different approaches, these studies explore how collective dynamic decisions on information acquisition diverge from individual ones, investigating the control-sharing effect and the impact of voting rules and committee conflicts. 
However, they often restrict attention to committee settings with (super-)majority rules.\footnote{The only exception is \citet{anesi2023deciding} who also allow for general voting procedures, 
but their focus is on 
the stability of different procedures.
} This paper expands the scope to more general interactions and voting procedures, providing a unified, yet simple, analytical method applicable beyond committee decision-making.




\vspace{1mm}

\paragraph{Sequential Sampling with Strategic Concerns.}

This paper relates to other papers 
that incorporate strategic concerns into sequential sampling. \citet{brocas2007influence}, \citet{henry2019research}, \citet*{che2023keeping}, and \citet{hebert2022engagement} explore dynamic persuasion between a persuader and a decision-maker. 
\citet*{brocas2012information} and \citet{gul2012war} study competition between conflicting parties that dynamically provide information to a receiver. 
This paper unifies these problems within a common framework 
and shows that players' behaviors across these scenarios are shaped by shared economic forces, particularly externalities in stopping due to preference misalignment.

\vspace{1mm}
\paragraph{Optimal Stopping and Concavification.}
Methodologically, this paper draws on \citet{morris2019wald}, who use results from the Skorokhod embedding literature \citep[see][for a survey]{obloj2004skorokhod} to show that sequential sampling problems can be reformulated in terms of choosing posterior belief distributions. Building on this, I extend their insights to strategic settings, identifying feasible belief distributions for each player given others' choices and incorporating fixed-point arguments into the concavification approach. This methodology based on Skorokhod embedding is broadly applicable and has been used in other optimal stopping problems, such as \citet{georgiadis2020optimal} in monitoring design.

\vspace{1mm}

\paragraph{Static Information Provision.}
My approach reformulates dynamic information acquisition as a static game and thus connects to the work on static information provision, 
particularly in Bayesian persuasion with information costs 
\citep*[e.g.,][]{gentzkow2014costly,lipnowski2020attention,lipnowski2022optimal,wei2021persuasion,matyskova2023bayesian}. Differently, my model assumes no commitment power, necessitating equilibrium characterization through fixed-point arguments rather than merely concavification. 

\citet{gentzkow2017bayesian,gentzkow2017competition} study simultaneous Bayesian persuasion with multiple senders. I show that their insights on the impact of competition extend to my dynamic setting, but dynamic competition reduces information revelation due to the control-sharing effect.



\section{A Model of Collective Sampling}
\label{framework}

\subsection{Model Setup}

Players $i\in N$, with $|N|\geq2$, collectively acquire information about an unknown binary state of the world, $\theta\in\Theta=\{0,1\}$. Let $i$ denote one generic player and $-i$ the others. Players share a common prior belief about the state given by $p_0=\mathbb{P}(\theta=1)\in(0,1)$. The collective information acquisition problem is modeled as a stopping game in continuous time with infinite horizon and time indexed by $t\in[0,\infty)$. Before information acquisition ends, players sequentially sample and observe public signals generated by an exogenous experimentation technology to learn about the state. Since signals are public and players share a common prior, their beliefs can be characterized by a stochastic \textbf{public belief process} $(p_t)_{t\geq0}$ derived from Bayes' rule, where $p_t$ denotes the public posterior belief about $\theta=1$ at time $t$. I use the belief process $(p_t)_{t\geq0}$ throughout the analysis without formalizing the underlying data-generating process. Assumptions on $(p_t)_{t\geq0}$ are introduced later.

At every instant during information acquisition, player $i$ privately\footnote{Unobservability is assumed for the convenience of defining strategies in continuous time, which can be replaced by a restriction to Markov strategies with observable actions. 
} decides whether to continue sampling ($a_{i,t}=1$) or stop ($a_{i,t}=0$); hence, individual stopping is reversible.\footnote{Irreversible individual stopping is discussed in \Cref{conclusion}.} The termination of information acquisition is determined by the profile of players' stopping choices. I consider a general class of collective stopping rules, 
dubbed $\mathcal{\group}$-collective stopping rules:
Let $\mathcal{\group}\subset 2^{N}\setminus\{\emptyset\}$ denote a collection of decisive coalitions, in which each $\group\in\mathcal{\group}$ is a coalition of players who can collectively terminate information acquisition upon agreement. Following the collective choice literature \citep[cf.][]{austen2000positive}, $\mathcal{\group}$ is assumed to be non-empty and monotone, i.e., $\group\in\mathcal{\group}$ and $\group\subset\group'$ imply $\group'\in\mathcal{\group}$.
\begin{definition}
    Under the \textbf{$\bm{\mathcal{\group}}$-collective stopping rule}, information acquisition ends the first time all players from a decisive coalition stop, that is, at the smallest $t$ such that $\{i\in N|a_{i,t}=0\}\in\mathcal{\group}$.
\end{definition}

When information acquisition ends at time $t$, the game also ends. If the realized belief path is $(p_s)_{s\in[0,t]}$, player $i$ obtains a payoff of $u_i(p_t)-\int_0^t \cost_i(p_s)\mathrm{d}s$. The first component $u_i(p)$ is a terminal payoff that depends on the public posterior belief $p$ about $\theta=1$ upon stopping; in applications, $u_i$ comes from continuation equilibrium play in a follow-up game where $\theta$ is payoff-relevant. The second component is player $i$'s total sampling cost: player $i$ bears a flow cost $\cost_i(p_s)\geq0$ at time $s$ that may depend on the posterior belief $p_s$;\footnote{
Such dependence can capture flow payoffs associated with other ``Markov'' strategic interactions, e.g., effort provision in experimentation; see the discussions in \Cref{conclusion}.} 
$\int_0^t \cost_i(p_s)\mathrm{d}s$ is the total cost accumulated up to time $t$.\footnote{Players will still bear sampling costs even when they stop at that moment but collective information acquisition has not ended.} Assume that both $u_i$ and $\cost_i$ are Lebesgue measurable and bounded. If information acquisition never ends, player $i$ obtains $-\int_0^\infty \cost_i(p_s)\mathrm{d}s$ with ``terminal'' payoffs normalized to zero.

\vspace{2mm}
\paragraph{Discussion on $\mathcal{\group}$-collective stopping} Examples of $\mathcal{\group}$-collective stopping rules include unanimity,
(super-)majority, and dictatorship rules. 

When $\mathcal{\group}=2^{N}\setminus\{\emptyset\}$, any player can terminate learning unilaterally, so this case is called \textbf{unilateral stopping}. When $\mathcal{\group}=\{N\}$, it is called \textbf{unanimous stopping} where learning only ends when all players stop at the same time. 
Another salient example is $\mathcal{\group}=\mathcal{\group}^q:=\{\group\subset N:|\group|\geq q\}$ for $q\in\{1,\dots,|N|\}$, called $q$-collective stopping, where information acquisition is terminated whenever $q$ players stop. This is a typical quota rule in the collective choice literature and is studied by \citet{chan2018deliberating} as supermajority rules (with $q\geq|N|/2$) in committee deliberation. 


$\mathcal{\group}$-collective stopping rules can also capture rules without anonymity like weighted voting. As another example, in some contexts, there is a ``chairperson'' such that information acquisition ends when a majority stops so long as the majority includes the chairperson. 
This can be modeled by $\mathcal{\group}=\mathcal{\group}^q\cap\mathcal{N}_i$ where $i$ is the chairperson and $\mathcal{N}_i:=\{\group\subset N:i\in \group\}$.

\subsection{Preliminaries: Beliefs and Strategies}
In this section, 
I discuss some preliminaries of this model, including the public belief process, strategy spaces, and solution concepts.

\subsubsection{Public Belief Process} Working with the belief process $(p_t)_{t\geq0}$ allows me to model information acquisition in a general way. Since beliefs are derived from Bayes' rule, $(p_t)_{t\geq0}$ must be a martingale bounded between 0 and 1. Additional assumptions on $(p_t)_{t\geq0}$ are stated in Assumption \ref{assumption}. 

Let $(\langle p\rangle_t)_{t\geq0}$ denote the quadratic variation process of $(p_t)_{t\geq0}$.\footnote{
$\langle p\rangle_t:=\lim_{||P||\to0}\sum_{k=1}^n(p_{t_{k}}-p_{t_{k-1}})^2$, where $P=\{0=t_0<t_1<\cdots<t_n=t\}$ stands for a partition of the interval $[0,t]$ and $||P||=\max_{k\in\{1,\dots,n\}}(t_k-t_{k-1})$.}

\begin{assumption}
\label{assumption}
	Given a probability space $(\Omega,\mathcal{F},\mathbb{P})$ and a filtration $\mathbb{F}=\{\mathcal{F}_t\}_{t\geq0}$ satisfying the usual assumptions, $(p_t)_{t\geq0}$ is a $[0,1]$-valued martingale such that:
	\begin{enumerate}
		\item[(i)] $(p_t)_{t\geq0}$ is Markov;
		\item[(ii)] $p_\infty:=\lim_{t\to\infty}p_t\in\{0,1\}$ almost surely;
		\item[(iii)] $(p_t)_{t\geq0}$ is continuous in $t$ almost surely;
		\item[(iv)] $\mathrm{d}\langle p\rangle_t$ is absolutely continuous with respect to $\mathrm{d}t$ (i.e., the Lebesgue measure) almost surely, 
        and $\mathbb{E}[\langle p\rangle_t]<\infty,\forall t\geq0$.\footnote{In fact, Condition (iii) is implied by (iv). 
        }
	\end{enumerate}
\end{assumption}

These assumptions require that (i) learning is memoryless, (ii) perfect learning can be achieved in the limit, and (iii) and (iv) learning is gradual. 
Condition (i) holds when information at each instant comes independently. 
Condition (ii) is only for expositional simplicity. 
Conditions (iii) and (iv) are critical for my analysis and exclude the possibility that information arrives in discrete lumps, but I show in \Cref{Poisson} that the analysis can also be adapted to Poisson learning with conclusive evidence. 

The widely used drift-diffusion process satisfies all the assumptions.

\begin{example}
\label{diffusion}
	Suppose that when players sample, they observe signals generated by a drift-diffusion process $(Z_t)_{t\geq0}$ with the following dynamics:
	\[
	\mathrm{d}Z_t=(2\theta-1)\mathrm{d}t+\sigma \mathrm{d}W_t, 
	\]
	where $\sigma>0$ and $(W_t)_{t\geq0}$ is a standard Brownian motion. According to \citet[Theorem 9.1]{liptser2013statistics}, the belief process satisfies
	\[
	\mathrm{d}p_t=\frac{2}{\sigma}\cdot p_t(1-p_t)\mathrm{d}B_t,\text{ where } B_t=\frac{1}{\sigma}\Big[Z_t-\int_0^t(2p_s-1)\mathrm{d}s\Big].
	\]
	Notice that $(B_t)_{t\geq0}$ is also a Brownian motion. Hence, $(p_t)_{t\geq0}$ is Markov and continuous with $\mathbb{P}\left(p_\infty\in\{0,1\}\right)=1$ and $\mathrm{d}\langle p\rangle_t=(4/\sigma^2)\cdot p_t^2(1-p_t)^2\mathrm{d}t$.
\end{example}

\subsubsection{Strategies and Solution Concepts}
\label{sect-strategy}

Player $i$ can condition her stopping decision on the histories she observes (including public beliefs and her previous actions). Hence, player $i$'s pure strategy can be any progressively measurable $\{0,1\}$-valued stochastic process on $(\Omega,\mathcal{F},\mathbb{F},\mathbb{P})$, denoted by $\alpha_i=(\alpha_{i,t})_{t\geq0}$, where $\alpha_{i,t}=1$ refers to sampling and $\alpha_{i,t}=0$ refers to stopping at time $t$. Let $\bm{\alpha}=(\alpha_j)_{j\in N}$ be the profile of players' stopping strategies and $\bm{\alpha}_{-i}=(\alpha_j)_{j\in N\setminus\{i\}}$ strategies by players other than $i$.

Since the belief process is Markov and only posterior beliefs matter for payoffs, it is natural to use the public belief $p_t$ as a state variable and consider Markov strategies in this game. For player $i$, a \textbf{pure Markov strategy} is a measurable function 
$\mkvstg_i:[0,1]\to\{0,1\}$ where $\mkvstg_i(p)=1$ refers to sampling and $\mkvstg_i(p)=0$ refers to stopping (where ``m'' stands for ``Markov''). Then $\sprg_i:=\{p\in[0,1]:\mkvstg_i(p)=1\}$ is the \textbf{sampling 
region} induced by $\mkvstg_i$: player $i$ only samples within this set and stops immediately when the public belief escapes from this set.
Accordingly, $[0,1]\setminus\sprg_i$ is the stopping region. There is a one-to-one correspondence between pure Markov strategies and sampling regions where $\mkvstg_i(p)=\mathbbm{1}_{p\in\sprg_i}$, so I often use them interchangeably. For technical reasons, I only consider pure Markov strategies that induce an \emph{open} sampling region.\footnote{Technically, 
the actual belief process is a truncated version of $(p_t)_{t\geq0}$ up to the time information acquisition ends. 
In order for the stochastic differential equation of the belief process to admit a solution, the stopping region must be closed.
} When there is no risk of confusion, I abuse notation and use $\alpha$ 
to denote pure Markov strategies.

Fix a prior $p_0$. With only one player, any stopping strategy $\alpha$ induces a random stopping time given by $\tau(\alpha):=\inf\left\{t\geq0:\alpha_t=0\right\}$, which is the first time this strategy specifies stopping. With many players, instead, the stopping time is jointly determined by the strategy profile $\bm{\alpha}$ and the collective stopping rule $\mathcal{\group}$, denoted by $\tau^{\mathcal{\group}}(\bm{\alpha})$:
\[
\tau^{\mathcal{\group}}(\bm{\alpha}):=\inf\left\{t\geq0:\exists\group\in\mathcal{\group},\text{ s.t. }\alpha_{i,t}=0,\forall i\in\group\right\}.
\]
When $\alpha_i$'s are all pure Markov strategies, 
\[
\tau^{\mathcal{\group}}(\bm{\alpha})=\inf\Big\{t\geq0:\min_{\group\in\mathcal{\group}}\max_{i\in \group}\alpha_i(p_t)=0\Big\}.
\]

At each moment $t$, given other players' strategies $\bm{\alpha}_{-i}$ and the history up to time $t$, each player chooses a continuation strategy $(\alpha_{i,s})_{s\geq t}$ to maximize her expected payoff as follows:
\[
\label{eq:constrained}
\mathbb{E}\left[u_i(p_{\tau^{\mathcal{\group}}(\alpha_i,\bm{\alpha}_{-i})})\mathbbm{1}_{\tau^{\mathcal{\group}}(\alpha_i,\bm{\alpha}_{-i})<\infty}-\int_t^{\tau^{\mathcal{\group}}(\alpha_i,\bm{\alpha}_{-i})} \cost_i(p_{s})\mathrm{d}s\bigg|\mathcal{F}_t\right]\tag{P}
\]
The indicator functions in the objectives denote that terminal payoffs $u_i$ are earned if and only if the game ends (in finite time). 

Notice that players are facing constrained stopping problems: player $i$ can only choose a stopping time from the set $\{\tau:\tau=\tau^{\mathcal{\group}}(\alpha_i,\bm{\alpha}_{-i})\text{ for some }\alpha_i\}$.

\vspace{2mm}
\paragraph{Solution Concepts}
I consider pure-strategy Markov perfect equilibrium (PSMPE).\footnote{\label{fn-mixed-strategy}
The focus on pure strategies is without loss of generality. Since $p_t$ is a continuous Feller process, as a corollary of Blumenthal's 0-1 law \citep[see][Theorem 2.15 and Proposition 2.19]{revuz2013continuous}, 
if players stop at some belief $p$ with \emph{any} strictly positive probability, then the process will stop at $p$ \emph{almost surely}. 
} 
Note that 
non-Markov deviations are allowed.

\vspace{2mm}
\paragraph{Sampling Regions v.s. Intervals}
At first glance, it might be compelling to claim that we only need to focus on sampling \emph{intervals}, since after all starting from a prior $p_0$, players will either stop at some belief $\underline{p}$ below $p_0$ or at some belief $\overline{p}$ above, leading to an interval $[\underline{p},\overline{p}]$. 
However, this argument does not take into account ``off-path'' continuation play and the sequential rationality therein. The continuation play  affects players' incentives for stopping, therefore we have to also specify what will happen outside $[\underline{p},\overline{p}]$ and thus need the whole sampling region $\sprg_i$. 
Indeed, there can be non-interval equilibria even under a collective stopping rule as simple as unanimous stopping; see examples in \Cref{sect-app-examples}. 

\section{Unilateral Stopping and Unanimous Stopping}
\label{sect-twoplayer}

To convey the main ideas and analytical methods transparently, I begin with two simple collective stopping rules: unilateral stopping and unanimous stopping. To ease the exposition, I focus on the two-player case, $N=\{1,2\}$, which entails no loss under these rules. The full analysis for many players and general collective stopping rules is deferred to \Cref{discussion}.

Recall that under unilateral stopping, either player can unilaterally terminate information acquisition; under unanimous stopping, information acquisition only ends when all players stop. Therefore, the corresponding stopping times of information acquisition simplify to
\[
\hat{\tau}(\alpha_i,\alpha_{-i}):=\inf\left\{t\geq0:\alpha_{1,t}=0\text{ or }\alpha_{2,t}=0\right\}
\]
for unilateral stopping, and
\[
\check{\tau}(\alpha_i,\alpha_{-i}):=\inf\left\{t\geq0:\alpha_{1,t}=\alpha_{2,t}=0\right\}
\]
for unanimous stopping. When players use pure Markov strategies $(\mkvstg_i,\mkvstg_{-i})$,
\begin{equation}
\label{eq:hattau}
    \hat\tau(\mkvstg_i,\mkvstg_{-i})=\inf\left\{t\geq0:\min\{\mkvstg_i(p_t),\mkvstg_{-i}(p_t)\}=0\right\}.
\end{equation}
\begin{equation}
\label{eq:checktau}
    \check\tau(\mkvstg_i,\mkvstg_{-i})=\inf\left\{t\geq0:\max\{\mkvstg_1(p_t),\mkvstg_2(p_t)\}=0\right\}.
\end{equation}

By definition, the stopping times satisfy 
\[
\hat{\tau}(\alpha_i,\alpha_{-i})\leq\tau(\alpha_{-i})\leq\check{\tau}(\alpha_i,\alpha_{-i}), 
\]
where $\tau(\alpha_{-i})=\inf\{t\geq0:\alpha_{-i,t}=0\}$ is the first time player $-i$ stops. Therefore, under unilateral stopping, player $i$ can only choose some stopping time 
earlier than $\tau(\alpha_{-i})$, while under unanimous stopping, player $i$ can only choose 
stopping times later than $\tau(\alpha_{-i})$. Players' constrained stopping problems (\ref{eq:constrained}) are thus relatively simple and transparent under unilateral and unanimous stopping.

\vspace{2mm}
\paragraph{Symmetric Equilibrium} For any profile of pure Markov strategies $(\alpha_1,\alpha_2)$, by \Cref{eq:hattau,eq:checktau}, $\hat{\tau}(\alpha_1,\alpha_2)=\hat{\tau}(\min\{\alpha_1,\alpha_2\},\min\{\alpha_1,\alpha_2\})$ and $\check{\tau}(\alpha_1,\alpha_2)=\check{\tau}(\max\{\alpha_1,\alpha_2\},\max\{\alpha_1,\alpha_2\})$. Hence, 
if $(\alpha_1,\alpha_2)$ is a PSMPE under unilateral stopping, $(\min\{\alpha_1,\alpha_2\},\min\{\alpha_1,\alpha_2\})$ is also a PSMPE and these two are outcome equivalent; similarly for unanimous stopping with respect to $(\max\{\alpha_1,\alpha_2\},\max\{\alpha_1,\alpha_2\})$. Without loss of generality, I focus on symmetric PSMPE in the following analysis.

\subsection{Ex Ante Reformulation of the Game}
\label{sect-reformulation}
In this section, I present an approach to transforming players' constrained stopping problems (\ref{eq:constrained}) into ex ante problems of semi-flexible information acquisition in which players choose distributions of posterior beliefs subject to majorization constraints. This approach bypasses the difficulty in the traditional dynamic programming approach in directly dealing with sophisticated Markov strategies (beyond simple cutoff strategies) and non-Markov deviations.

In the single-player case, every stopping strategy $\alpha$ induces a distribution of posterior beliefs $p_{\tau(\alpha)}$, denoted by $F_{\tau(\alpha)}(p):=\mathbb{P}(p_{\tau(\alpha)}\leq p),\forall p\in[0,1]$. 
Since $(p_t)_{t\geq0}$ is a (uniformly integrable) martingale, any induced distributions over posterior beliefs is \textbf{Bayes plausible}: $\mathbb{E}_{F_{\tau(\alpha)}}[p]=\mathbb{E}[p_{\tau(\alpha)}]=p_0$.\footnote{By the Optional Stopping Theorem; see Theorem 3.2 in \citet{revuz2013continuous}.} 

Under unilateral stopping, given the other player's strategy $\alpha_{-i}$, player $i$ chooses a stopping strategy $\alpha_i$ and thus implements a stopping time $\hat\tau(\alpha_i,\alpha_{-i})\leq\hat\tau(\alpha_{-i},\alpha_{-i})=\tau(\alpha_{-i})$. It is easy to verify that $F_{\hat\tau(\alpha_i,\alpha_{-i})}$ is a \textbf{mean-preserving contraction (MPC)} of $F_{\hat\tau(\alpha_{-i},\alpha_{-i})}$: players should have learned more than $F_{\hat\tau(\alpha_i,\alpha_{-i})}$ if they continued learning from $\hat\tau(\alpha_i,\alpha_{-i})$ to $\hat\tau(\alpha_{-i},\alpha_{-i})$. Under unanimous stopping, player $i$ can only implement some $\check\tau(\alpha_i,\alpha_{-i})\geq\check\tau(\alpha_{-i},\alpha_{-i})=\tau(\alpha_{-i})$ so that $F_{\check\tau(\alpha_i,\alpha_{-i})}$ is a \textbf{mean-preserving spread (MPS)} of $F_{\check\tau(\alpha_{-i},\alpha_{-i})}$. Moreover, since player $i$ can never end information acquisition where the other is sampling, $F_{\check\tau(\alpha_i,\alpha_{-i})}$ must have support in $[0,1]\setminus\sprg_{-i}$ when $\alpha_{-i}$ can be represented by a sampling region $\sprg_{-i}$. These claims hold for general stopping strategies, and the converse is also true when $\alpha_{-i}$ is pure Markov.

\begin{lemma}[UNI-MPC]
\label{MPC}
    Fix a prior $p_0$, a player $i$, and a pure Markov strategy $\mkvstg_{-i}$. A stopping strategy $\alpha_i$ exists such that $F_{\hat{\tau}(\alpha_i,\mkvstg_{-i})}=F$ if and only if $F$ is an MPC of $F_{\hat{\tau}(\mkvstg_{-i},\mkvstg_{-i})}$.
\end{lemma}

\begin{lemma}[UNA-MPS]
\label{MPS}
    Fix a prior $p_0$, a player $i$, and a pure Markov strategy $\mkvstg_{-i}$. A stopping strategy $\alpha_i$ exists such that $F_{\check{\tau}(\alpha_i,\mkvstg_{-i})}=F$ if and only if $F$ is an MPS of $F_{\check{\tau}(\mkvstg_{-i},\mkvstg_{-i})}$ and $\supp(F)\subset[0,1]\setminus\{p:\mkvstg_{-i}(p)=1\}$.
\end{lemma}

The proofs are based on the Skorokhod embedding theory and in particular an embedding method by \citet{chacon1976one}.


\subsubsection{Semi-flexible Information Acquisition with Costs}
Given the previous observations, we can study collective information acquisition from the ex ante perspective in which players choose distributions of posterior beliefs instead of stopping strategies. Accordingly, we should consider ex ante sampling costs.

Let $\quadvar(p):=\mathrm{d}\langle p\rangle_t/\mathrm{d}t$ at $p_t=p$.\footnote{$\quadvar$ is well defined because $(p_t)_{t\geq0}$ is Markov; see the proof of \Cref{cost} for details.} Define
\begin{equation}
	\label{static cost}
	\phi_i(p):=\int_{z}^p\int_{z}^x \frac{2\cost_i(y)}{\quadvar(y)}\mathrm{d}y\mathrm{d}x\text{ }\text{ for some }z\in(0,1), \quad\text{for }i=1,2.
	\end{equation}
Since $c_i\geq0$, $\phi_i$ is convex. \Cref{cost} below shows the ex ante total sampling cost can be rewritten using distributions over posterior beliefs. 
\begin{lemma}
\label{cost}
	For every a.s. finite stopping time $\tau$,
	\[
	\mathbb{E}\left[\int_0^\tau \cost_i(p_s)\mathrm{d}s\right]
	=\int_0^1\phi_i(p)\mathrm{d}F_\tau(p)-\phi_i(p_0),\quad\text{for }i=1,2,
	\]
where $\phi_i$ is defined in \Cref{static cost}.
\end{lemma}

With the above ex ante sampling costs, the semi-flexible information acquisition problems are formulated as follows:
\vspace{2mm}
\paragraph{Unilateral Stopping} A distribution of posterior belief $F$ is a solution to Problem (\ref{unilateral}) with $p_0$ if $\mathbb{E}_{F}[p]=p_0$ and for $i=1,2$,
\[
\label{unilateral}
	F\in\argmax_{H} \int_0^1[u_i(p)-\phi_i(p)]\mathrm{d}H(p)\text{ s.t. }H \text{ is an MPC of }F.
\tag{S-Uni}
\]
\paragraph{Unanimous Stopping} A distribution of posterior belief $F$ is a solution to Problem (\ref{unanimous}) with $(p_0,\sprg)$ if $\mathbb{E}_{F}[p]=p_0$ and for $i=1,2$,
\[
\begin{aligned}
\label{unanimous}
	F\in&\argmax_H\int_0^1[u_i(p)-\phi_i(p)]\mathrm{d}H(p)\\
    &\text{ s.t. } H \text{ is an MPS of }F,\supp(H)\subset[0,1]\setminus\sprg.
\end{aligned}
\tag{S-Una}
\]
Here $\sprg$ captures the candidate (or the other player's) sampling region.
\vspace{3mm}

A solution $F$ to Problem (\ref{unilateral}) or (\ref{unanimous}) must simultaneously solve both players' information acquisition problems given that the other player chooses $F$. In other words, $F$ is a fixed point of players' best responses.

Notice that distributions over posterior beliefs induced by pure Markov strategies always have binary support: $p_{\tau(\alpha)}\in\{\underline{p}(p_0;\alpha),\overline{p}(p_0;\alpha)\}$ where $\underline{p}(p_0;\alpha):=\sup\{p\leq p_0:\alpha(p)=0\}$ and $\overline{p}(p_0;\alpha):=\inf\{p\geq p_0:\alpha(p)=0\}$.\footnote{Both $\{p\leq p_0:\alpha(p)=0\}$ and $\{p\geq p_0:\alpha(p)=0\}$ are non-empty since we restrict to strategies that have open sampling regions.} Equivalently, in terms of sampling regions $\sprg$, we have 
\begin{equation}
\label{ulbd}
\begin{aligned}
\underline{p}(p_0;\sprg):=&\sup\left\{p\leq p_0:p\in[0,1]\setminus\sprg\right\},\\
\
\overline{p}(p_0;\sprg):=&\inf\left\{p\geq p_0:p\in[0,1]\setminus\sprg\right\}.
\end{aligned}
\end{equation}
\begin{definition}
	A distribution of posterior beliefs $F$ is a \textbf{binary policy} given a prior $p_0$ if $F$ is Bayes plausible with respect to $p_0$ and 
    $\supp(F)=\{\underline{p},\overline{p}\}$ for some $\underline{p},\overline{p}\in[0,1]$.
\end{definition}

The following theorem states that solving for PSMPE is equivalent to finding binary policy solutions to Problem (\ref{unilateral}) or (\ref{unanimous}).
\begin{theorem}
	\label{MPE}
	A profile of pure Markov strategies $(\alpha,\alpha)$ is sequentially rational at $p_t=p$ under unilateral stopping if and only if $F_{\hat{\tau}(\alpha,\alpha)}$ is a binary policy solution to Problem (\ref{unilateral}) with $p_0=p$. Instead, $(\alpha,\alpha)$ is sequentially rational at $p_t=p$ under unanimous stopping if and only if $F_{\check{\tau}(\alpha,\alpha)}$ is a binary policy solution to Problem (\ref{unanimous}) with $(p_0,\sprg)$ where $p_0=p$ and $\sprg=\{p':\alpha(p')=1\}$.
\end{theorem}
By definition, $(\alpha,\alpha)$ is a PSMPE if and only if it is sequentially rational at any $p_t$. Equilibrium existence is implied by the fact that the degenerate distribution at prior (no information) is a binary policy solution to Problem (\ref{unilateral}) for any $p_0$, thus $\alpha(p)=\mathbbm{1}_{p\in\emptyset}$ is a PSMPE under unilateral stopping, and the binary policy with support $0$ and $1$ (full information) a solution to Problem (\ref{unanimous}) for any $p_0$ and $\sprg=(0,1)$, so $\alpha(p)=\mathbbm{1}_{p\in(0,1)}$ is a PSMPE under unanimous stopping.

\Cref{MPE} naturally extends to many players where the induced distribution over posterior beliefs has to solve every player's ex ante problem.

\subsection{Equilibrium Characterization}
\label{sect-eqm-characterization}
Henceforth I refer to the sampling region in a PSMPE as an ``equilibrium''. This section characterizes what sampling regions can be equilibria. According to \Cref{MPE}, the problem becomes characterizing binary policy solutions to the ex ante problems (\ref{unilateral}) and (\ref{unanimous}), which can be done by concavification \citep[see][]{aumann1995repeated,kamenica2011bayesian}.

Different from standard Bayesian persuasion, players face majorization constraints in addition to Bayes plausibility. With binary policies, we are able to embed these majorization constraints into players' objectives, so the concavification technique becomes applicable. The optimal choices of distributions of posterior beliefs determined by concavification are now best responses to the other player' choice of sampling regions. To find equilibria, we need to characterize the fixed points in players' best responses.

For any sampling region $\sprg$ and $i=1,2$, define the following concave closures: For unilateral stopping, construct the \textbf{``$\sprg$-inward'' concave closure} of $u_i-\phi_i$:
\begin{align*}
	{V}_i^{IN}\left(p;\sprg\right):=&\sup\left\{z:(p,z)\in co\left((u_i-\phi_i)|_{[\underline{p}(p;\sprg),\overline{p}(p;\sprg)]}\right)\right\},\text{ for }p\in[0,1],
\end{align*}
where $\underline{p}(p;\sprg)$ and $\overline{p}(p;\sprg)$ are defined in \Cref{ulbd} 
and $co(f)$ denotes the convex hull of the graph of $f$. Notice that ${V}_i^{IN}(p;\sprg)=u_i(p)-\phi_i(p)$ for $p\in[0,1]\setminus\sprg$. See \Cref{fig-unilateral} for graphical illustrations with $\sprg=(b,B)$, in which blue curves are $u_i-\phi_i$ and orange curves are ``$\sprg$-inward'' concave closures $V_i^{IN}$. 
When $p_t=p\in(b,B)$, $V_i^{IN}(p;(b,B))$ captures the highest expected continuation payoff player $i$ can obtain by choosing an MPC of the binary policy with support $\{b,B\}$, or say, by curtailing sampling within $(b,B)$ under unilateral stopping.

\begin{figure}[!t]
	\centering
	\begin{subcaptionblock}{0.48\textwidth}
		\centering
		\includegraphics[scale=0.35]{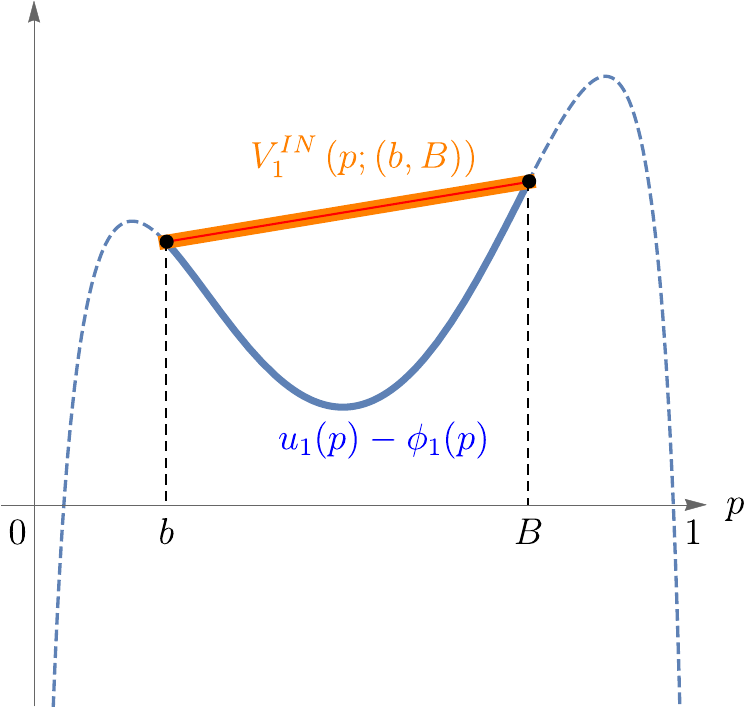}
            \caption{P1's incentive under $\sprg=(b,B)$.}
            \label{fig-unilateral-p1-b}
	\end{subcaptionblock}
        \begin{subcaptionblock}{0.48\textwidth}
		\centering
		\includegraphics[scale=0.35]{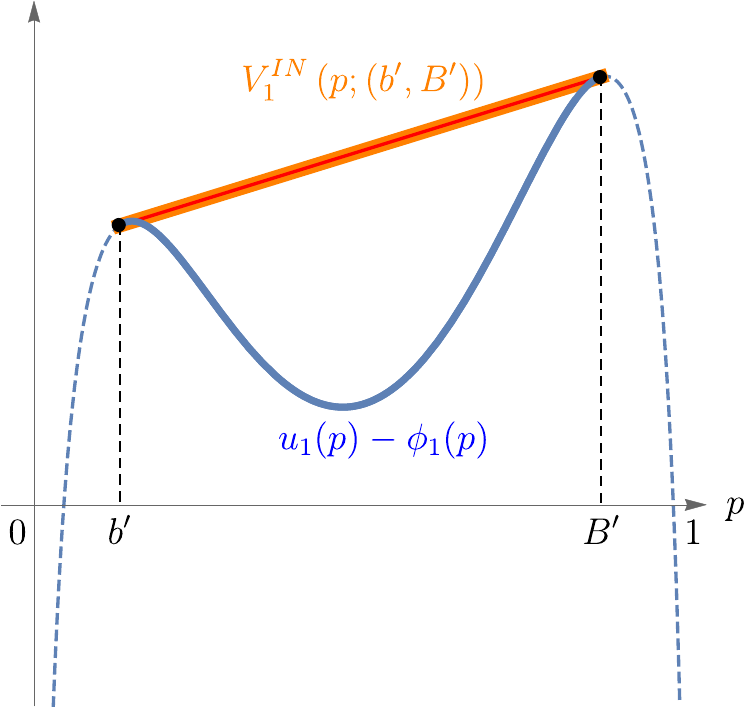}
            \caption{P1's incentive under $\sprg=(b',B')$.}
	\end{subcaptionblock}
	\begin{subcaptionblock}{0.48\textwidth}
		\centering
		\includegraphics[scale=0.35]{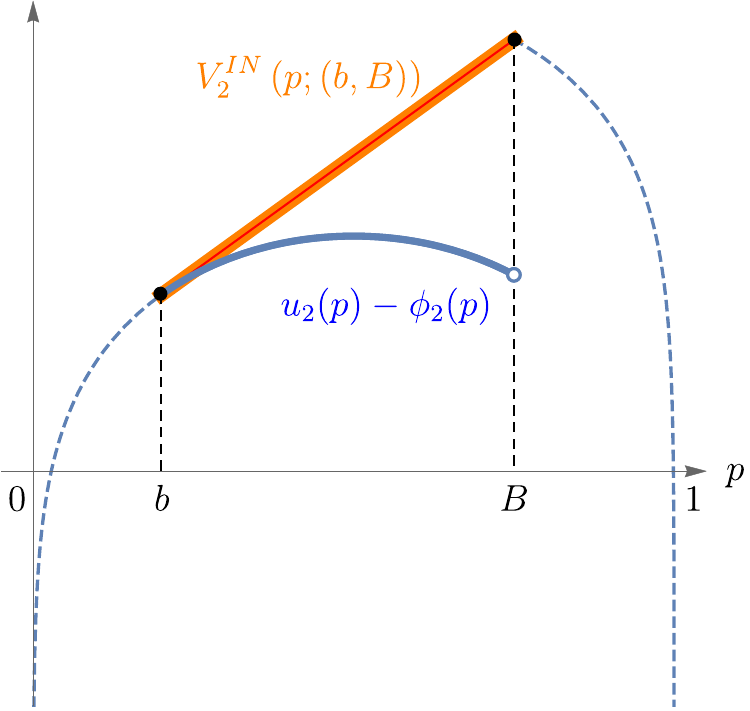}
            \caption{P2's incentive under $\sprg=(b,B)$.}
	\end{subcaptionblock}
        \begin{subcaptionblock}{0.48\textwidth}
		\centering
		\includegraphics[scale=0.35]{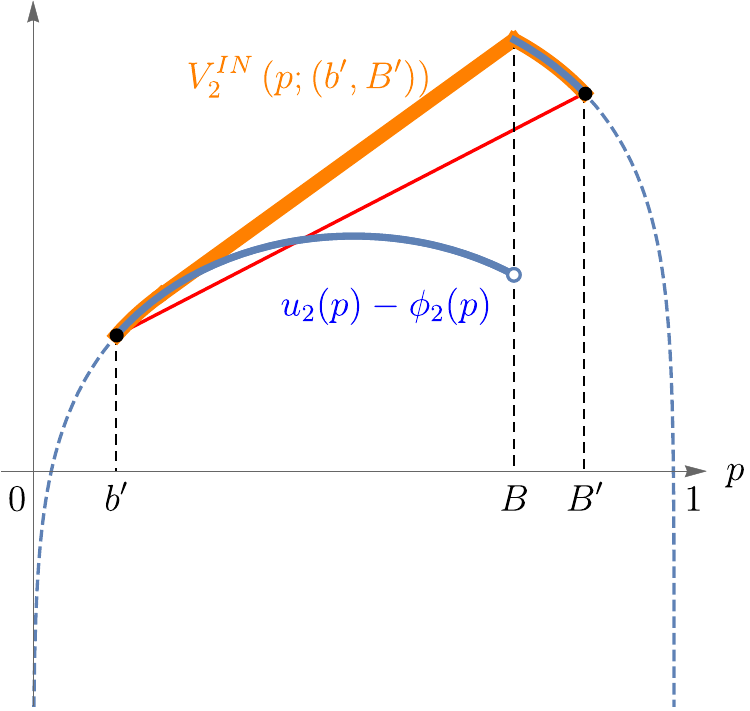}
            \caption{P2's incentive under $\sprg=(b',B')$.}
            \label{fig-unilateral-p2-b'}
	\end{subcaptionblock}
	\caption{Unilateral Stopping.}
    \label{fig-unilateral}
\end{figure}

For unanimous stopping, first define an auxiliary payoff $v_i^{OUT}$ that creates ``holes'' within the sampling region $\sprg$ and otherwise equals $u_i-\phi_i$. Then construct the concave closure of $v_i^{OUT}$, which is the \textbf{``$\sprg$-outward'' concave closure} of $u_i-\phi_i$:
\begin{align*}
	v^{OUT}_i\left(p;\sprg\right):=&\left\{
	\begin{array}{ll}
		-\infty, &\text{for }p\in\sprg,\\
		u_i(p)-\phi_i(p),&\text{for }p\in[0,1]\setminus\sprg,\\
	\end{array}
	\right.\\
	{V}_i^{OUT}\left(p;\sprg\right):=&\sup\left\{z:(p,z)\in co(v_i^{OUT})\right\},\text{ for }p\in[0,1].
\end{align*}
See \Cref{fig-unanimous} for graphical illustrations with $\sprg=(b,B)$, where blue curves are ``modified'' payoffs $v_i^{OUT}$ and orange curves are ``$\sprg$-outward'' concave closures $V_i^{OUT}$. The ``hole'' within the graph of $v^{OUT}_i$ is constructed to prohibit stopping within $(b,B)$ since it is impossible to terminate information acquisition there under unanimous stopping when the other player uses $\sprg=(b,B)$. When $p_t=p$, $V_i^{OUT}(p;(b,B))$ captures the highest expected continuation payoff player $i$ can obtain by choosing an MPS with support in $[0,1]\setminus(b,B)$ of the binary policy with support $\{b,B\}$, or say, by prolonging the sampling process beyond $(b,B)$.

\begin{figure}[!t]
	\centering
	\begin{subcaptionblock}{0.48\textwidth}
		\centering
		\includegraphics[scale=0.35]{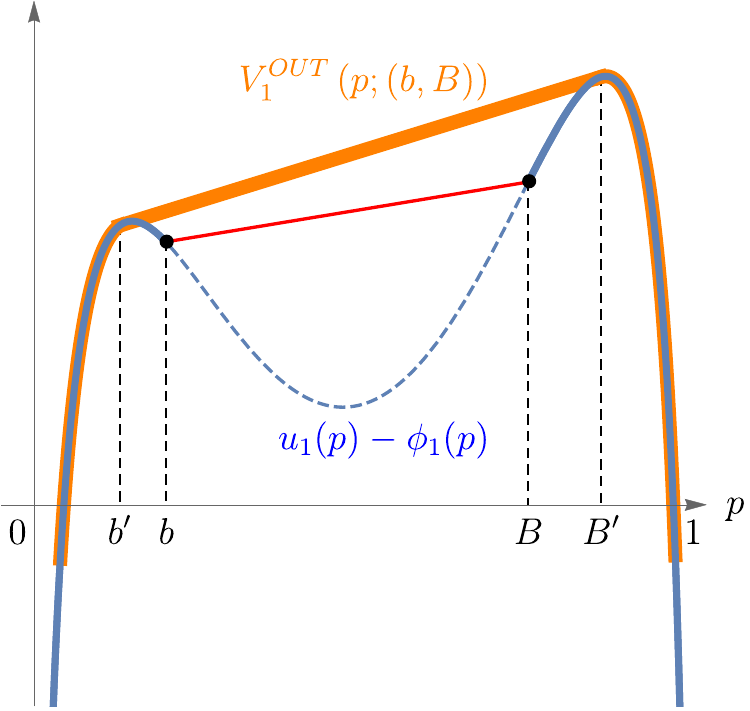}
            \caption{P1's incentive under $\sprg=(b,B)$.}
            \label{fig-unanimous-p1-b}
	\end{subcaptionblock}
        \begin{subcaptionblock}{0.48\textwidth}
		\centering
		\includegraphics[scale=0.35]{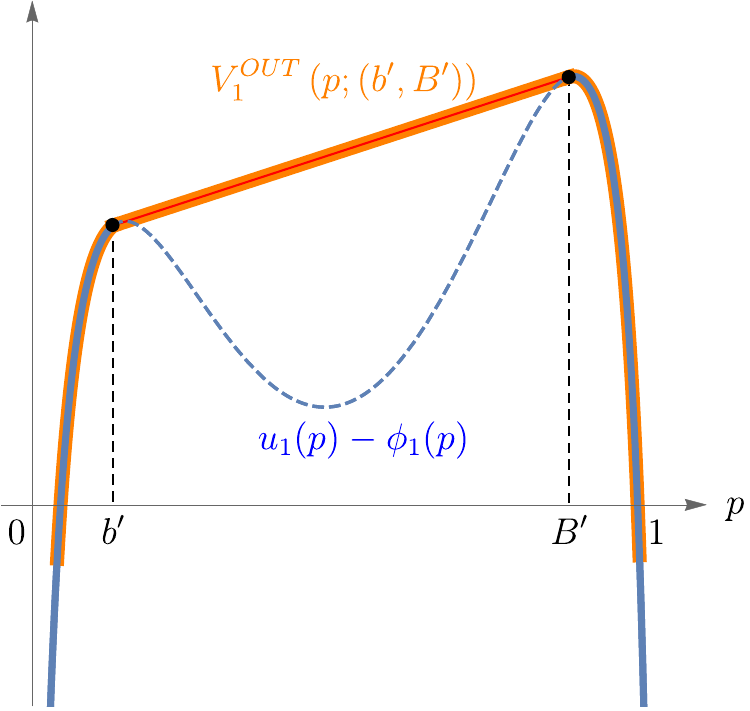}
            \caption{P1's incentive under $\sprg=(b',B')$.}
	\end{subcaptionblock}
	\begin{subcaptionblock}{0.48\textwidth}
		\centering
		\includegraphics[scale=0.35]{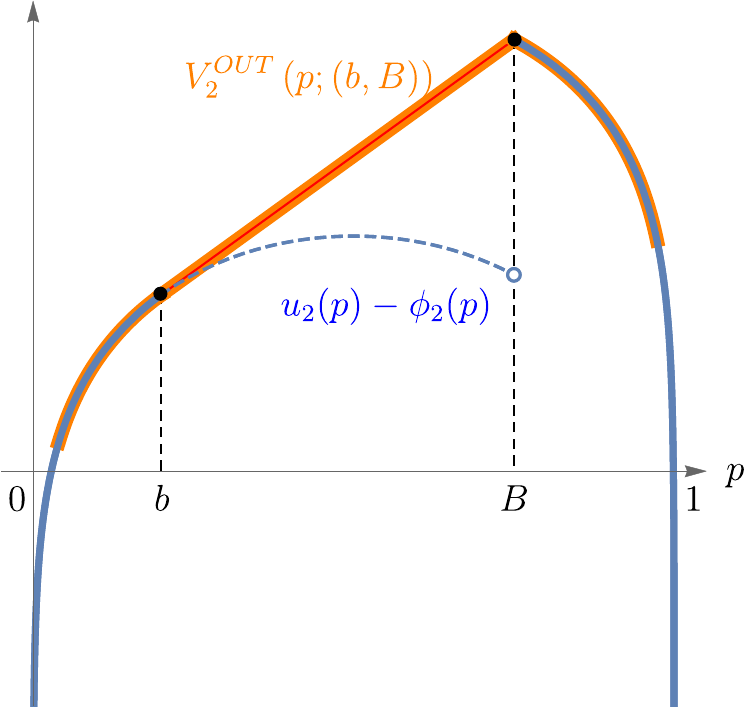}
            \caption{P2's incentive under $\sprg=(b,B)$.}
	\end{subcaptionblock}
        \begin{subcaptionblock}{0.48\textwidth}
		\centering
		\includegraphics[scale=0.35]{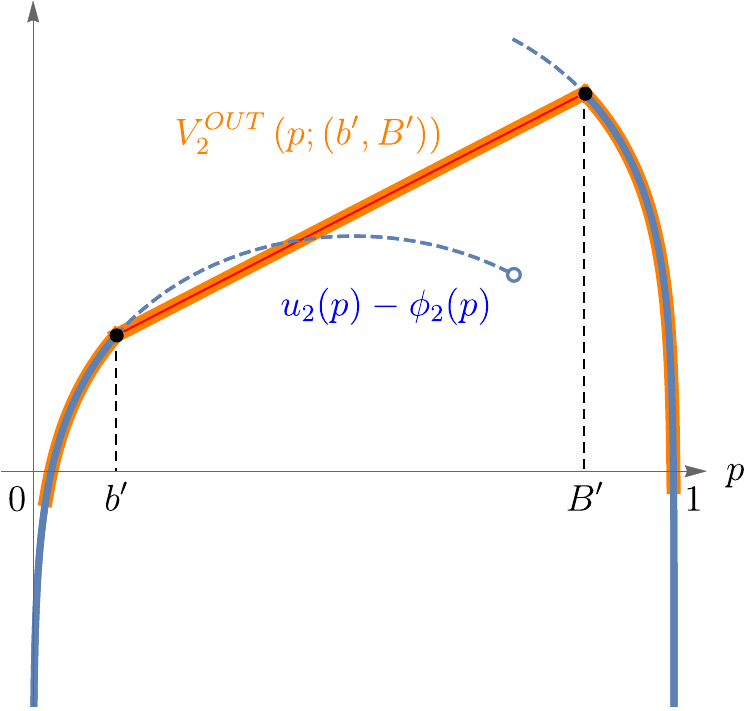}
            \caption{P2's incentive under $\sprg=(b',B')$.}
            \label{fig-unanimous-p2-b'}
	\end{subcaptionblock}
	\caption{Unanimous Stopping.}
 \label{fig-unanimous}
\end{figure}

\begin{theorem}[Concavification]
\label{concave}
	A sampling region $\sprg$ is an equilibrium under unilateral stopping if and only if $\forall i=1,2$,
\begin{itemize}
    \item[(i)] For any $p\in\sprg$, with $\underline{p}:=\underline{p}(p;\sprg)$ and $\overline{p}:=\overline{p}(p;\sprg)$ as defined in \Cref{ulbd},
\begin{align*}
\frac{\overline{p}-p}{\overline{p}-\underline{p}}[u_i(\underline{p})-\phi_i(\underline{p})]+\frac{p-\underline{p}}{\overline{p}-\underline{p}}[u_i(\overline{p})-\phi_i(\overline{p})]
=V_i^{IN}\left(p;\sprg\right).
\end{align*}
    \item[(ii)] For any $p\in[0,1]\setminus\sprg$, $u_i(p)-\phi_i(p)=V_i^{IN}\left(p;\sprg\right).$ 
\end{itemize}
The same result holds for unanimous stopping with $V_i^{IN}$ replaced by ${V}_i^{OUT}$.
\end{theorem}
For any $\underline{p}\leq\overline{p}$, define
\begin{equation}
    \label{ex-ante-payoff}
U_i\big(\underline{p},\overline{p};p\big):=\left\{
\begin{array}{ll}
   \frac{\overline{p}-p}{\overline{p}-\underline{p}}[u_i(\underline{p})-\phi_i(\underline{p})]+\frac{p-\underline{p}}{\overline{p}-\underline{p}}[u_i(\overline{p})-\phi_i(\overline{p})]  & \text{if }\underline{p}<\overline{p}\text{ and }p\in[\underline{p},\overline{p}], \\
   u_i(p)-\phi_i(p)  & \text{otherwise}.
\end{array}
\right.
\end{equation}
With some abuse of the notation, let $U_i(\sprg;p):=U_i(\underline{p}(p;\sprg),\overline{p}(p;\sprg);p)$ for any $\sprg$ and $p\in[0,1]$. Then conditions (i) and (ii) in \Cref{concave} can be summarized by:
\begin{equation}
\label{eqm-concave-overall}
U_i\left(\sprg;p\right)=V_i^{IN}\left(p;\sprg\right) \left(\text{or } {V}_i^{OUT}\left(p;\sprg\right)\right), \text{ for all }p\in[0,1].
\end{equation}



Suppose that player $-i$ uses the sampling region $\sprg$, i.e., stops if and only if $p_t$ escapes $\sprg$. Then with a belief $p$, player $i$'s expected continuation payoff from using the same sampling region is exactly $U_i(\sprg;p)$: when $p\in\sprg$, it is the linear combination of payoffs at the boundaries of $(\underline{p}(p;\sprg),\overline{p}(p;\sprg))$; otherwise, it is $u_i(p)-\phi_i(p)$ since information acquisition ends immediately at $p\not\in\sprg$. 
Hence, \Cref{concave} simply presents the conditions whereby using the equilibrium sampling region is optimal for both players given that the other player uses it. 
This result naturally extends to many players where \Cref{eqm-concave-overall} should hold for all $i\in N$. 

See \Cref{fig-unilateral,fig-unanimous} for graphical illustrations of equilibrium conditions. The figures are drawn based on the same payoffs $u_1-\phi_1$ and $u_2-\phi_2$. 
Let us consider two candidate sampling regions: $(b,B)$ and $(b',B')$, where $(b,B)$ is player 2's individually optimal sampling region and $(b',B')$ is player 1's. We only need to check whether player 1 has incentives to deviate from $(b,B)$ and similarly for $(b',B')$ but with respect to player 2. 

Focus on unilateral stopping. The $(b,B)$-inward concave closure of player 1's payoff, $V_1^{IN}(p;(b,B))$, is exactly given by the line connecting points $(b,u_1(b)-\phi_1(b))$ and $(B,u_1(B)-\phi_1(B))$, as shown in \Cref{fig-unilateral-p1-b}. Therefore, player 1 always likes more information within $(b,B)$ and has no incentive to stop earlier. As a result, $(b,B)$ is an equilibrium. However, given 
$(b',B')$, player 2 wants to stop immediately when $p_t$ hits or goes above $B$, as illustrated by the $(b',B')$-outward concave closure $V_2^{IN}(p;(b',B'))$ in \Cref{fig-unilateral-p2-b'}. Hence, $(b',B')$ is not an equilibrium. 

In contrast, under unanimous stopping, players should have no incentive to stop later and outside the sampling region in equilibrium. According to \Cref{fig-unanimous}, $(b',B')$ is an equilibrium but $(b,B)$ is not.


It is worth noting that equilibria are not necessarily intervals, especially under unanimous stopping. See \Cref{sect-app-examples} for examples.

\vspace{2mm}
\paragraph{Maximum and Minimal Equilibria}
Multiple sampling regions can satisfy the conditions in \Cref{concave} and be equilibria. A sampling region $\sprg$ is \textbf{larger than} another $\sprg'$ if $\sprg'\subset \sprg$.
Based on \Cref{concave}, I show that equilibria form a semi-lattice under this partial order. 

\begin{lemma}
\label{lattice}
    If $\sprg$ and $\sprg'$ are equilibria, $\sprg\cup\sprg'$ is also an equilibrium.
\end{lemma}

As an implication of \Cref{lattice}, a \textbf{maximum equilibrium} exists. 
Under unilateral stopping, for any two comparable equilibria, the larger one is preferred by both players, since they both have the option to deviate to the smaller equilibrium 
but neither has such incentive. Hence, the maximum equilibrium  
is Pareto dominant among all equilibria.

Although a minimum equilibrium may not exist
, we can define \textbf{minimal equilibria} that are \emph{no strictly larger than} any other equilibria. 
Under unanimous stopping, 
only minimal equilibria are not Pareto dominated.

Thus for comparative statics, I often focus on the maximum equilibrium under unilateral stopping and the minimal ones under unanimous stopping.

\subsection{Inefficiencies and Conflicts of Interest}
\label{sect-comparative}

\paragraph{Pareto Inefficiencies \& the Control-sharing Effect}
A natural question is how collective information acquisition compares to the single-agent case and whether it is efficient or not. 
For any $\lambda=(\lambda_1,\lambda_2)\in\mathbb{R}_+^2$ such that $\lambda\ne\bm{0}=(0,0)$, define
$
W(\cdot;\lambda):=\lambda_1(u_1-\phi_1)+\lambda_2(u_2-\phi_2).
$
A sampling region $\sprg$ is $\lambda$-efficient if it is optimal for $W(\cdot;\lambda)$:

\begin{definition}
    A  sampling region $\sprg$ is \textbf{$\bm{\lambda}$-efficient} if 
    \[
    \sum_{i=1,2}\lambda_iU_i\left(\sprg;p\right)=\hat{W}(p;\lambda),\quad\forall p\in[0,1].
    \]
    where $\hat{W}(p;\lambda)$ is the concave closure of $W(p;\lambda)$.
\end{definition}

$\lambda$-efficient sampling regions capture all Pareto efficient sampling regions, including player $i$'s individual optimal choice, denoted by $\sprg_i^*$.

Suppose that $u_i$ are upper semi-continuous. Then a $\lambda$-efficient sampling region exists: $\sprg_\lambda^*:=\{p:\hat{W}(p;\lambda)>W(p;\lambda)\}$ is (open and) $\lambda$-efficient. For expositional simplicity, assume that $\sprg_\lambda^*$ is the unique $\lambda$-efficient sampling region for any $\lambda$.\footnote{In general, we can consider maximal $\lambda$-efficient sampling regions for unilateral stopping and minimal ones for unanimous stopping.} By comparing equilibria with $\lambda$-efficient sampling regions, I show that collective information acquisition is (weakly) Pareto inefficient: there is too little learning under unilateral stopping; and under unanimous stopping, learning is never too little but possibly too much. 

\begin{proposition}[Pareto Inefficiency]
\label{pareto}
    Suppose that $\sprg_\lambda^*$ is the unique $\lambda$-efficient sampling region for $\lambda\in\mathbb{R}_+^2,\lambda\ne\bm{0}$. 
    \begin{enumerate}
        \item[(a)] Unilateral stopping: An equilibrium is smaller than $\sprg_\lambda^*$.
        \item[(b)] Unanimous stopping: An equilibrium is no strictly smaller than $\sprg_\lambda^*$.
    \end{enumerate}
\end{proposition}

Due to the nature of unilateral stopping, equilibria can never be strictly larger than the individual optimum; similarly, equilibria under unanimous stopping can never be strictly smaller. 
\Cref{pareto} claims a stronger result for unilateral stopping that equilibria must be smaller. This is driven by the control-sharing effect: The possibility that the other player may shut down learning before it achieves player $i$'s optimum always lowers player $i$'s continuation value of learning and thus urges her to stop earlier.

Formally, the result for unilateral stopping is due to the semi-lattice structure of equilibria. Fix a weight $\lambda\in\mathbb{R}_+^2\setminus\{\bm{0}\}$ and consider an auxiliary game between two virtual players with the same preference $W(\cdot;\lambda)=\lambda_1(u_1-\phi_1)+\lambda_2(u_2-\phi_2)$. If $\sprg$ is an equilibrium, 
it should also be an equilibrium in the auxiliary game; otherwise, either player 1 or 2 has strict incentives to deviate in the original game. The $\lambda$-efficient $\sprg_\lambda^*$ is another equilibrium in the auxiliary game. Then by \Cref{lattice}, $\sprg_\lambda^*\cup \sprg$ is also an equilibrium, from which players have no incentive to deviate to $\sprg_\lambda^*$. Since $\sprg_\lambda^*$ is uniquely optimal for $W$, it must be $\sprg_\lambda^*\cup \sprg=\sprg_\lambda^*$ and thus $\sprg\subset \sprg_\lambda^*$. 

In contrast, equilibria under unanimous stopping are not necessarily larger than the individual optimum (see \Cref{sect-app-examples} for examples). The reason behind this is also the control-sharing effect: 
That the other player may prolong learning after it achieves player $i$'s optimum also urges player $i$ to stop earlier to avoid wasteful sampling in the future. 
In fact, any equilibrium with strictly more learning than both players' individual optima must rely on miscoordination:
\begin{proposition}[Control-sharing Effect]
\label{control-sharing}
    Under unanimous stopping, a minimal equilibrium is no strictly larger than $\sprg_1^*\cup\sprg_2^*$.
\end{proposition}

\paragraph{Impact of Conflicts of Interest}
Given \Cref{pareto}, 
we expect learning inefficiencies to be amplified when players' preferences become less aligned. This is formalized as follows. Suppose that there are 
$f$ and $g$ such that players' terminal payoffs take the following form:
\[
u_1(p)=f(p)+bg(p)\quad\text{and}\quad u_2(p)=f(p)-bg(p),
\]
where $b\geq0$ captures the degree of preference misalignment.\footnote{The same form of preference misalignment is studied by \citet{gentzkow2017bayesian,gentzkow2017competition} in competition in persuasion.} Let $T^b$ denote the set of all equilibria when misalignment is $b$. 

\begin{lemma}
\label{misalign}
    Suppose that $\cost_1(\cdot)=\cost_2(\cdot)$. Then $T^b\subset T^{b'}$ for any $b>b'\geq0$.
\end{lemma}

\Cref{lattice,misalign} imply the following result. 

\begin{proposition}[Misalignment]
\label{misalign-max}
    Suppose that $\cost_1(\cdot)=\cost_2(\cdot)$ and $b>b'\geq0$. 
    \begin{enumerate}
        \item[(a)] Unilateral stopping: The maximum equilibrium under misalignment $b$ is smaller than the maximum one under misalignment $b'$.
        \item[(b)] Unanimous stopping: A minimal equilibrium under misalignment $b$ is no strictly smaller than a minimal one under misalignment $b'$.
    \end{enumerate} 
\end{proposition}


\section{Many Players and Coalitional Stopping Rules}
\label{discussion}

The analysis of unilateral and unanimous stopping rules in the two-player case extends naturally to more players and $\mathcal{\group}$-collective stopping rules.

Recall that under $\mathcal{\group}$-collective stopping, information acquisition ends when players in at least one coalition $\group\in\mathcal{\group}$ stop. 
When players use pure Markov strategies, with $\sprg_j=\{p:\alpha_j(p)=1\}$, the stopping time is
\[
\tau^{\mathcal{\group}}(\bm{\alpha})=\inf\Big\{t\geq0:\min_{\group\in\mathcal{\group}}\max_{j\in \group}\alpha_j(p_t)=0\Big\}=\inf\left\{t\geq0:p_t\notin\cap_{\group\in\mathcal{\group}}\cup_{j\in \group}\sprg_j\right\}.
\]
That is, the collective sampling region is given by $\cap_{\group\in\mathcal{\group}}\cup_{j\in \group}\sprg_j$.

For any player $i\in N$, consider two extremes: $i$ always stops ($\alpha_i\equiv0$) or never stops ($\alpha_i\equiv1$). With fixed 
$\bm{\alpha}_{-i}$, these two extreme strategies induce two extreme stopping times. 
By adjusting her strategy, player $i$ is only able to induce some stopping time between the two extreme times. A similar argument thus goes through as in \Cref{sect-reformulation}: In terms of distributions of posterior beliefs, player $i$ can only choose some distribution that is not only an MPS of some $\underline{F}{}_i$ but also an MPC of some $\overline{F}{}_i$. This reformulation again allows equilibrium characterization via concavification.

To keep the exposition concise, I directly present the equilibrium characterization and intuition in this section. Interested readers can consult \Cref{proof-ex-ante} for the detailed reformulation. 

\subsection{Equilibrium Characterization under $\mathcal{\group}$-collective Stopping} Consider an arbitrary player $i\in N$. Let $\mathcal{\group}_i:=\{\group\in\mathcal{\group}:i\in \group\}$, so $\mathcal{\group}\backslash\mathcal{\group}_i$ is the collection of groups in $\mathcal{\group}$ that do not contain player $i$. For any profile of other players' sampling regions $\bm{\sprg}_{-i}=(\sprg_j)_{j\in N\setminus\{i\}}$, define\footnote{I adopt the convention that the union over $\emptyset$ is $\emptyset$ and the intersection over $\emptyset$ is $(0,1)$. Hence, $\cont_i^{\mathcal{\group}}(\cdot)=(0,1)$ if $\mathcal{\group}\backslash\mathcal{\group}_i=\emptyset$ and $\stpbf_i^{\mathcal{\group}}(\cdot)=\emptyset$ if $\{i\}\in\mathcal{\group}$.}
\begin{equation}
\label{eq:supp}
\cont_i^{\mathcal{\group}}(\bm{\sprg}_{-i}):=\bigcap_{\group\in\mathcal{\group}\backslash\mathcal{\group}_i}\bigcup_{j\in \group}\sprg_j
\quad\text{and}\quad
\stpbf_i^{\mathcal{\group}}(\bm{\sprg}_{-i}):=\bigcap_{\group\in\mathcal{\group}}\bigcup_{j\in \group\setminus\{i\}}\sprg_j,
\end{equation}
\vspace{-5mm}
\begin{align*}
\text{with}\quad\underline{p}{}_i^{\mathcal{\group}}(p;\bm{\sprg}_{-i}):=&\sup\left\{p'\leq p:p\in[0,1]\setminus\cont_i^{\mathcal{\group}}(\bm{\sprg}_{-i})\right\},\\
\overline{p}_i^{\mathcal{\group}}(p;\bm{\sprg}_{-i}):=&\inf\left\{p'\geq p:p'\in[0,1]\setminus\cont_i^{\mathcal{\group}}(\bm{\sprg}_{-i})\right\}.
\end{align*}
The dependence on $\bm{\sprg}_{-i}$ will often be suppressed for simplicity.

On the one hand, $\cont_i^{\mathcal{\group}}(\bm{\sprg}_{-i})$ captures the collective sampling region if player $i$ never stops; starting from a belief $p$, $\underline{p}{}_i^{\mathcal{\group}}$ and $\overline{p}{}_i^{\mathcal{\group}}$ are the stopping beliefs given the sampling region $\cont_i^{\mathcal{\group}}$. On the other hand, $\stpbf_i^{\mathcal{\group}}(\bm{\sprg}_{-i})$ is the the sampling region if player $i$ always stops; no matter what player $i$ does, information acquisition can never stop within $\stpbf_i^{\mathcal{\group}}$. By definition, $\stpbf_i^{\mathcal{\group}}\subset\cont_i^{\mathcal{\group}}$. 

Intuitively, $\overline{\cont_i^{\mathcal{\group}}}\setminus\stpbf_i^{\mathcal{\group}}$ captures player $i$'s freedom in deviations\footnote{$\overline{E}$ denotes the closure of a set $E\subset[0,1]$.}---she can only stop at beliefs where all groups in $\mathcal{\group}\setminus\mathcal{\group}_i$ are still sampling (thus within $\cont_i^{\mathcal{\group}}$) and at the same time at least one group in $\mathcal{\group}_i$ exists such that all members other than $i$ stop so that player $i$ is pivotal (thus outside $\stpbf_i^{\mathcal{\group}}$). 
Formally, for any starting belief $p\in\cont_i^{\mathcal{\group}}$, player $i$ can either never stop so players will stop at the boundaries of $\cont_i^{\mathcal{\group}}$---$\underline{p}{}_i^{\mathcal{\group}}$ and $\overline{p}_i^{\mathcal{\group}}$---or she can stop earlier but only outside $\stpbf_i^{\mathcal{\group}}$.\footnote{In the ex ante reformulation, player $i$ can only choose $F$ with support in $\overline{\cont_i^{\mathcal{\group}}}\setminus\stpbf_i^{\mathcal{\group}}$. }

Following the constructions of $V_i^{IN}$ and $V_i^{OUT}$, we now define
\begin{align*}
v^{\mathcal{\group}}_i\left(p;\bm{\sprg}_{-i}\right):=&\left\{
	\begin{array}{ll}
        u_i(p)-\phi_i(p),&\text{for }p\in[0,1]\setminus\stpbf_i^{\mathcal{\group}}(\bm{\sprg}_{-i}),\\
		-\infty, &\text{for }p\in\stpbf_i^{\mathcal{\group}}(\bm{\sprg}_{-i}),\\
	\end{array}
	\right.\\
V_i^{\mathcal{\group}}(p;\bm{\sprg}_{-i}):=&\sup\left\{z:(p,z)\in co\left(v^{\mathcal{\group}}_i\big|_{\left[\underline{p}{}_i^{\mathcal{\group}}(p;\bm{\sprg}_{-i}),\overline{p}_i^{\mathcal{\group}}(p;\bm{\sprg}_{-i})\right]}\right)\right\},\text{ for }p\in[0,1].
\end{align*}

In the special case of unilateral stopping with two players, $\cont_i^{\mathcal{\group}}(\bm{\sprg}_{-i})=\sprg_{-i}$ and $\stpbf_i^{\mathcal{\group}}(\bm{\sprg}_{-i})=\emptyset$, that is player $i$ can stop anywhere within $\sprg_{-i}$ and thus $V_i^{\mathcal{\group}}$ reduces to $V_i^{IN}$. Instead, for unanimous stopping, $\cont_i^{\mathcal{\group}}(\bm{\sprg}_{-i})=(0,1)$ and $\stpbf_i^{\mathcal{\group}}(\bm{\sprg}_{-i})=\sprg_{-i}$, so player $i$ can stop anywhere outside $\sprg_{-i}$, in which case $V_i^{\mathcal{\group}}$ reduces to $V_i^{OUT}$. 

The equilibrium characterization reads similar to that in \Cref{concave}:
\begin{theorem}
    \label{thm-g-collective}
    A profile of sampling regions $\bm{\sprg}=(\sprg_i)_{i\in N}$, with the induced collective sampling region $\sprg:=\cap_{\group\in\mathcal{\group}}\cup_{j\in \group}\sprg_j$, is an equilibrium under $\mathcal{\group}$-collective stopping if and only if for any $i\in N$,
    \begin{itemize}
    \item[(i)] For any $p\in\sprg$, with $\underline{p}:=\underline{p}(p;\sprg)$ and $\overline{p}:=\overline{p}(p;\sprg)$,
\begin{align*}
\frac{\overline{p}-p}{\overline{p}-\underline{p}}[u_i(\underline{p})-\phi_i(\underline{p})]+\frac{p-\underline{p}}{\overline{p}-\underline{p}}[u_i(\overline{p})-\phi_i(\overline{p})]
=V_i^{\mathcal{\group}}(p;\bm{\sprg}_{-i}); \text{ and }
\end{align*}
    \item[(ii)] For any $p\in[0,1]\setminus\sprg$, $u_i(p)-\phi_i(p)=V_i^{\mathcal{\group}}(p;\bm{\sprg}_{-i}).$ 
\end{itemize}
\end{theorem}
\Cref{thm-g-collective} generalizes \Cref{concave} to many players and $\mathcal{\group}$-collective stopping. 
Again the left-hand side of the equilibrium condition refers to player $i$'s expected payoff in equilibrium, and the right-hand side is the best payoff she can ever achieve given $\bm{\sprg}_{-i}$, characterized by 
concave closure $V_i^{\mathcal{\group}}$. 

\subsection{Comparative Statics}
Under $\mathcal{\group}$-collective stopping, two classes of players are particularly salient: $N_{uni}(\mathcal{\group}):=\{i\in N:\{i\}\in\mathcal{\group}\}$ and $N_{una}(\mathcal{\group}):=\{i\in N:i\in\cap_{\group\in\mathcal{\group}}\group\}$. Players in $N_{uni}(\mathcal{\group})$ can always deviate to smaller sampling regions like under unilateral stopping, while players in $N_{una}(\mathcal{\group})$ can always deviate to larger ones like under unanimous stopping. In general, both sets can be empty.


The same intuition as in \Cref{sect-twoplayer} applies: Equilibria are too small for players in $N_{uni}(\mathcal{\group})$ and possibly too large for those in $N_{una}(\mathcal{\group})$;\footnote{It is impossible to have multiple players in both classes at the same time: $|N_{uni}|\geq1$ and $|N_{una}|\geq1$ implies $N_{uni}=N_{uni}=\{i\}$ for some $i\in N$, in which case equilibria are simply player $i$'s individual optima.} and greater preference misalignment within either class amplifies inefficiencies.

\begin{proposition}
\label{pareto-general}
    Any equilibrium is smaller than any Pareto efficient sampling region for players in $N_{uni}(\mathcal{\group})$ and no strictly smaller than any Pareto efficient one for players in $N_{una}(\mathcal{\group})$. 
    
\end{proposition}
\begin{proposition}
\label{misalign-general}
    Suppose that $\cost_i(\cdot)=\cost_j(\cdot)$ for players $i$ and $j$. When preference misalignment between $i$ and $j$ increases, the maximum equilibrium becomes smaller if $i,j\in N_{uni}(\mathcal{\group})$, and minimal equilibria never become strictly smaller if $i,j\in N_{una}(\mathcal{\group})$.
\end{proposition}

A more interesting exercise is with respect to collective stopping rules. Intuitively, on the one hand, since every coalition can unilaterally end information acquisition, the more coalitions are involved, the shorter learning is. On the other hand, since unanimous consent is required within coalitions, the more players one coalition consists of, the longer learning is. 

\begin{proposition}
\label{cs-general}
    For any $\mathcal{\group}\subset\mathcal{\group}'$, the maximum equilibrium under $\mathcal{\group}'$ cannot be strictly larger than that under $\mathcal{\group}$. The same for minimal equilibria.
\end{proposition}
Given the monotonicity of $\mathcal{\group}$, adding players into existing coalitions is equivalent to reducing decisive coalitions.

This result applies to many situations of interest. For example, since $\mathcal{\group}^q\subset\mathcal{\group}^{q'}$ for $q'<q$, a higher majority requirement tends to increase learning. 
Likewise, consider the weighted-voting variant of $q$-collective stopping where 
some players have more than one vote. This is as if adding coalitions into $\mathcal{\group}^q$ and thus typically reduces learning. In contrast, having a chairperson reduces coalitions and thus tends to increase learning.

However, a caveat exists: Equilibrium multiplicity is much more severe than under unilateral or unanimous stopping. The comparative statics in \Cref{cs-general} are only meaningful when $N_{uni}(\mathcal{\group})\ne\emptyset$ or $N_{una}(\mathcal{\group})\ne\emptyset$. 
Otherwise, all players using the same strategy is always an equilibrium regardless of what that strategy is. 
Instead of pursuing a universal solution, I turn to a specific situation with more structures on payoffs in \Cref{sect-committee-search} and investigate the implication of an appropriate refinement. We will see that the previous comparative statics still hold given the refinement.

\section{Applications}
\label{application}

\subsection{Committee Search between Two Alternatives}
\label{sect-committee-search}

In this section, I consider 
committee search between two alternatives. With binary choice, the problem with many players simplifies to one with two pivotal players.

A committee of players, $N=\{1,\dots,2m-1\}$ ($m\geq2$ is an integer), decides between two alternatives, $L$ and $R$. The payoff from each alternative depends on an underlying binary state $\theta\in\Theta=\{0,1\}$, about which players share a common prior, $p_0=\mathbb{P}(\theta=1)\in(0,1)$. In state $\theta=0$, player $i$'s payoff is 1 from $L$ and 0 from $R$; in state $\theta=1$, the payoff is 0 from $L$ and $v_i\geq0$ from $R$. Players are heterogeneous, with $v_1<\cdots<v_i<\cdots<v_{2m-1}$. Notice that players' preferences are aligned: they all prefer $L$ in state $\theta=0$ and $R$ in state $\theta=1$; however, the intensity differs: player $i$ prefers $R$ if and only if her belief about $\theta=1$ is high, i.e., $p\geq1/(1+v_i)=:w_i$, where $w_i$ is player $i$'s individual belief threshold. 

Before making a collective decision, the committee can search for public information and update their beliefs. This process is modeled as in the baseline model under \textbf{$\pmb{\mathcal{\group}}$-collective stopping}. 
At the end of information acquisition, players vote for alternatives by a $\texttt{piv}$-majority rule for some $\texttt{piv}\in\{m,\dots,2m-1\}$: when the public posterior belief is $p$, player $i$ obtains a payoff of $u_i(p)=1-p$ if $p<\piv$ and $u_i(p)=pv_i$ if $p\geq \piv$.\footnote{
Under the $\texttt{piv}$-majority rule, the committee selects $R$ if more than $\texttt{piv}$ players vote for it, and $L$ otherwise. In the unique equilibrium in undominated strategies, all players vote sincerely according to their individual thresholds.
}

\vspace{2mm}
\paragraph{Strong Equilibrium, Intervals, and One-sided Best Responses} 
Consider the following equilibrium refinement based on the notion of Strong Nash equilibrium \citep{aumann1959acceptable}: An equilibrium is strong if no group, taking the strategies of its complement as given, can cooperatively deviate in a way that benefits all of its member \emph{at all beliefs}.\footnote{Different from \citet{aumann1959acceptable}, my notion of strong equilibrium requires potential deviations to be sequentially rational, so deviating players must benefit at all beliefs.} Formally,
\begin{definition}
    An equilibrium $\bm{\sprg}=(\sprg_i)_{i\in N}$ is \textbf{strong} if for all $J\subset N$ and for all $(\sprg_j')_{j\in J}$, there exist a player $k\in J$ and a belief $p\in[0,1]$ such that $U_k(\sprg;p)> U_k(\sprg';p)$ where $\sprg=\cap_{\group\in\mathcal{\group}}\cup_{i\in\group}\sprg_i$ and $\sprg'=\cap_{\group\in\mathcal{\group}}\cup_{i\in\group}\sprg_i'$ with $\sprg_i'=\sprg_i$ for $i\in N\setminus J$.
\end{definition}

It is straightforward to see that in this committee search problem, a strong equilibrium must induce an interval collective sampling region $(\underline{p},\overline{p})$ with $\underline{p}\leq\piv\leq\overline{p}$. Otherwise, if $\sprg$ contains an interval either below or above $\piv$, then all players prefer to eliminate that interval (regardless of the current belief) because $u_i-\phi_i$ is piecewise concave. 

For any interval sampling region $(\underline{p},\overline{p})$, 
it will turn out to be useful to investigate the following one-sided best responses by players. Recall that player $i$'s expected payoff at belief $p$ under $(\underline{p},\overline{p})$ is given by $U_i(\underline{p},\overline{p};p)$ as defined in \Cref{ex-ante-payoff}. For $i\in N$, define
\[
\label{upper-undominated}
\overline{B}_i(\underline{p}):=\argmax_{\overline{p}\geq w_{\texttt{piv}}} U_i\big(\underline{p},\overline{p};w_{\texttt{piv}}\big),\quad\text{for }\underline{p}<w_{\texttt{piv}},
\tag{UB}
\]
with $\overline{B}_i(w_{\texttt{piv}}):=\lim_{\underline{p}\nearrow w_{\texttt{piv}}}\overline{B}_i(\underline{p})$, and
\[
\label{lower-undominated}
\underline{B}{}_i(\overline{p}{}):=\argmax_{\underline{p}< w_{\texttt{piv}}} U_i\big(\underline{p},\overline{p}{};w_{\texttt{piv}}\big), \quad\text{for }\overline{p}>w_{\texttt{piv}},\footnote{Let $\underline{B}_i(\overline{p}):=\{w_{\texttt{piv}}\}$ when the maximum does not exist, which happens if and only if the supremum is $\lim_{\underline{p}\nearrow w_{\texttt{piv}}}U_i(\underline{p},\overline{p}{};w_{\texttt{piv}})$. 
}
\tag{LB}
\]
with $\underline{B}{}_i(w_{\texttt{piv}}):=\lim_{\overline{p}{}\searrow w_{\texttt{piv}}}\underline{B}{}_i(\overline{p}{})$.\footnote{When either $\underline{p}=w_{\texttt{piv}}$ or $\overline{p}=w_{\texttt{piv}}$, $U_i$ is equal to $u_i(w_{\texttt{piv}})-\phi_i(w_{\texttt{piv}})$, so anything can be a best response. Instead, I consider the limit of best responses at this point.
} 
When $\cost_i(p)>0$ and thus $\phi_i$ is strictly convex, both $\overline{B}_i(\underline{p})$ and $\underline{B}{}_i(\overline{p}{})$ are singletons. 


For a sampling interval $(\underline{p},\overline{p})$ with $\underline{p}\leq\piv\leq\overline{p}$, the one-sided best responses tell us whether players think stopping at $\underline{p}$ and $\overline{p}$ is too early or too late for them. 
If some players agree that stopping and adopting $L$ at $\underline{p}$ is too early (i.e., $\underline{B}{}_i(\overline{p}{})<\underline{p}$), then they can all benefit from acquiring more information and stopping at a lower belief. Players can indeed implement this deviation if they can gather at least one player from each coalition $\group\in\mathcal{\group}$, which is plausible when $\max_{\group\in\mathcal{\group}}\min_{i\in\group}\underline{B}{}_i(\overline{p}{})<\underline{p}$. Instead, if $\max_{\group\in\mathcal{\group}}\min_{i\in\group}\underline{B}{}_i(\overline{p}{})>\underline{p}$, there is a decisive coalition whose members all agree that stopping at $\underline{p}$ is too late, thus they can benefit from stopping earlier. Hence, it must be  $\max_{\group\in\mathcal{\group}}\min_{i\in\group}\underline{B}{}_i(\overline{p}{})=\underline{p}$ in any strong equilibrium. Applying the same argument to $\overline{p}$ yields $\min_{\group\in\mathcal{\group}}\max_{i\in\group}\overline{B}{}_i(\underline{p}{})=\overline{p}$.

\vspace{2mm}
\paragraph{Homogeneous Sampling Costs}
To further simplify the analysis, consider the case in which players are homogeneous in sampling costs, i.e., $\cost_i(p)=\cost(p)>0$ for all $i$ and $p\in[0,1]$. 
Monotone comparative statics can then verify that both $\overline{B}_i(\underline{p})$ and $\underline{B}{}_i(\overline{p})$ are decreasing in $i$. 

Hence, we have $\overline{p}=\min_{\group\in\mathcal{\group}}\max_{i\in\group}\overline{B}{}_i(\underline{p}{})=\overline{B}{}_{\ubp(\mathcal{\group})}(\underline{p}{})$ where $\ubp(\mathcal{\group})=\max_{\group\in\mathcal{\group}}\min_{i\in\group}i$, therefore the upper bound of a strong equilibrium must be determined by player $\ubp(\mathcal{\group})$. Symmetrically, the lower bound 
is determined by player $\lbp(\mathcal{\group})=\min_{\group\in\mathcal{\group}}\max_{i\in\group}i$. It thus suffices to focus on the interaction between two pivotal players $\ubp(\mathcal{\group})$ and $\lbp(\mathcal{\group})$: in any strong equilibrium, $\overline{p}{}=\overline{B}_{\ubp(\mathcal{\group})}(\underline{p})$ and $\underline{p}=\underline{B}{}_{\lbp(\mathcal{\group})}(\overline{p}{})$. 
The converse is also true when $\lbp(\mathcal{\group})>\ubp(\mathcal{\group})$: Every fixed point in $\underline{B}{}_{\lbp(\mathcal{\group})}$ and $\overline{B}_{\ubp(\mathcal{\group})}$ 
is also a strong equilibrium. However, when $\lbp(\mathcal{\group})\leq\ubp(\mathcal{\group})$, only the maximum fixed point is guaranteed to be a strong equilibrium. 


\begin{proposition}
\label{prop:strong}
    Suppose that players are homogeneous in sampling costs. If $\sprg$ is a strong equilibrium, then $\sprg=(\underline{p},\overline{p})$ such that $\underline{p}=\underline{B}{}_{\lbp(\mathcal{\group})}(\overline{p})$ and $\overline{p}=\overline{B}_{\ubp(\mathcal{\group})}(\underline{p})$. When $\lbp(\mathcal{\group})>\ubp(\mathcal{\group})$, if $\underline{p}=\underline{B}{}_{\lbp(\mathcal{\group})}(\overline{p})$ and $\overline{p}=\overline{B}_{\ubp(\mathcal{\group})}(\underline{p})$,
    then $(\underline{p},\overline{p})$ is a strong equilibrium; when $\lbp(\mathcal{\group})\leq\ubp(\mathcal{\group})$, if $(\underline{p},\overline{p})$ is the maximum fixed point, it is a strong equilibrium.
\end{proposition}



This phenomenon that only two players are pivotal in collective decisions with binary choice and a one-dimensional state space is also present in \citet{compte2010bargaining} and \citet{chan2018deliberating}.


Adding more decisive coalitions typically reduces learning: For $\mathcal{\group}\subset\mathcal{\group}'$, $\lbp(\mathcal{\group})\geq\lbp(\mathcal{\group}')$ and $\ubp(\mathcal{\group})\leq\ubp(\mathcal{\group}')$, thus $\underline{B}{}_{\lbp(\mathcal{\group})}\leq\underline{B}{}_{\lbp(\mathcal{\group}')}$ and $\overline{B}_{\ubp(\mathcal{\group})}\geq\overline{B}_{\ubp(\mathcal{\group}')}$. Hence, strong equilibria under $\mathcal{\group}'$ are expected to be smaller than under $\mathcal{\group}$.
\begin{proposition}
\label{cs-committee}
    Suppose that players have homogeneous sampling costs. For any $\mathcal{\group}\subset\mathcal{\group}'$, when $\lbp(\mathcal{\group})>\ubp(\mathcal{\group})$, any strong equilibrium under $\mathcal{\group}'$ is not strictly larger than that under $\mathcal{\group}$; when $\lbp(\mathcal{\group})\leq\ubp(\mathcal{\group})$, any strong equilibrium under $\mathcal{\group}'$ is not strictly larger than the maximum one under $\mathcal{\group}$.
\end{proposition}
However, the effect of increased diversity/conflict in the committee is less clear and depends on the collective stopping rule $\mathcal{\group}$. First, only when $\lbp(\mathcal{\group})>m>\ubp(\mathcal{\group})$ or $\lbp(\mathcal{\group})\leq m\leq\ubp(\mathcal{\group})$ (i.e., both the leftists and rightists of the committee are represented in $\mathcal{\group}$), a more dispersed distribution of $v_i$'s for sure increases the conflict between the pivotal players. Second, even if so, when $\lbp(\mathcal{\group})>m>\ubp(\mathcal{\group})$ (unanimous stopping between two players), an increased conflict often increases learning; when $\lbp(\mathcal{\group})\leq m\leq\ubp(\mathcal{\group})$ (unilateral stopping), the opposite is true.

\vspace{2mm}
\paragraph{Quota Rules
} When $\mathcal{\group}=\mathcal{\group}^q=\{\group\subset N:|\group|\geq q\}$, $\lbp(\mathcal{\group}^q)=q$ and $\ubp(\mathcal{\group}^q)=2m-q$. \Cref{cs-committee} implies that an increase in $q$ typically increase committee learning. When $q>m$, $\lbp(\mathcal{\group}^q)>m>\ubp(\mathcal{\group}^q)$, an increase in diversity thus often increases learning; the opposite is true when $q\leq m$ and thus $\lbp(\mathcal{\group}^q)\leq m\leq\ubp(\mathcal{\group}^q)$.

\vspace{2mm}

\paragraph{Quota Rules with a Chairperson} Suppose that a chairperson $i^*\in N$ can veto stopping on top of $q$-collective stopping. Then $\mathcal{\group}=\mathcal{\group}^q\cap\mathcal{N}_{i^*}$ where $\mathcal{N}_{i^*}:=\{\group\subset N:i^*\in \group\}$. Without loss assume $1\leq i^*\leq m$.

By \Cref{cs-committee}, having a chairperson never leads to less learning. In particular, $\lbp(\mathcal{\group})=\max\{q,i^*\}\geq q=\lbp(\mathcal{\group}^q)$ and $\ubp(\mathcal{\group})=\min\{i^*,2m-q\}\leq2m-q=\ubp(\mathcal{\group}^q)$. A larger $\lbp$ means a lower $\underline{B}{}_{\lbp}$ and a smaller $\ubp$ means a higher $\underline{B}{}_{\ubp}$. Hence, when $q\leq i^*$, the committee learns more with higher accuracy upon selecting either alternative. When $i^*<q<2m-i^*$, the equilibrium sampling region shifts upward, so the committee only has higher accuracy upon selecting $\beta$, but also accepts $\alpha$ more easily. 

Not surprisingly, a higher majority requirement, i.e., a higher $q$, undermines the chairperson's power: When $q\leq i^*$, chairperson $i^*$ fully controls the learning process; when $q>2m-i^*$, she plays no role at all.




\vspace{2mm}
\paragraph{Connection to \citet{chan2018deliberating}} 
This application on committee search is similar to the committee deliberation model in \citet{chan2018deliberating}. Except for some modeling differences, 
\citeauthor{chan2018deliberating} only consider (super-)majority rules, while I study general deliberation rules and provide a general comparison between different rules. 
My analysis also reveals that their result with respect to diversity (Propositions 4(3)) relies on their focus on (super-)majority rules. 
I do not further explore properties of the one-sided best responses. By pursuing it, one should be able to derive sharper and richer results as in \citet{chan2018deliberating}.\footnote{Methodologically, \citet{chan2018deliberating} use dynamic programming, while my analysis relies on the ex ante approach developed in the baseline model. 
}  

\subsection{Competition in Persuasion}
\label{sect-competition-persuasion}
In this section, I discuss another application on competition in persuasion and show how dynamic incentive shapes information revelation differently from its static analogue.

Consider a game in which a group of senders, denoted by $J$, compete to persuade a receiver who cares about an unknown binary state $\theta\in\Theta=\{0,1\}$. 
The senders persuade the receiver by dynamically disclosing public information about the state, with the public belief process exogenously given by $(p_t)_{t\geq0}$ satisfying \Cref{assumption}. The competition in persuasion is modeled by \textbf{unanimous stopping} among the senders.\footnote{To focus on competition, I mute the receiver's incentive. 
See the previous version of this paper \citet{zhou2023collective} for an application on dynamic persuasion between a sender and a receiver where the receiver is free to stop and make decisions at any time.} When all senders stop, the game ends and sender $i\in J$ collects a payoff $u_i(p)$ when the (receiver's) public posterior belief is $p$; otherwise, everyone, including the receiver, continues to learn the state and senders bear flow costs $\cost_i(p_t)$. 

The following result ensues as a corollary of \Cref{pareto-general,misalign-general,cs-general}: Increasing competition by adding senders or reducing the alignment of senders' preferences leads to more information revelation.

\begin{corollary}
\label{GK17}
Regardless of preferences $u_i$ and $c_i$,
\begin{enumerate}
    \item[(a)] A collusive sampling region is no strictly larger than an equilibrium;
    \item[(b)] A minimal equilibrium when the set of senders is $J'$ is no strictly larger than a minimal equilibrium when the set of senders is $J\supset J'$;
    \item[(c)] A minimal equilibrium cannot become strictly smaller when preference misalignment between any two senders increases.
\end{enumerate}
\end{corollary}

\paragraph{Unanimous Stopping and Connectedness}
\citet{gentzkow2017bayesian,gentzkow2017competition} study static competition in persuasion. 
They introduce a ``Blackwell-connectedness'' condition on the information environment 
and establish results on the impact of competition under this condition. \Cref{GK17} is similar to their results despite derived in a dynamic setup. This is because my dynamic model has a similar ``connectedness'' property in terms of sampling regions due to unanimous stopping.\footnote{For a given prior, sampling regions can be transformed into belief distributions, but Blackwell-connectedness is not satisfied due to the support restriction in \Cref{MPS}. 
} Let $\langle\bm{\alpha}\rangle$ denote the sampling region induced by a profile of pure Markov stopping strategies $\bm{\alpha}$. For any sampling region $\sprg$, any player $i$, and any $\bm{\alpha}_{-i}$ such that $\sprg\supset\sprg_{-i}:=\langle\bm{\alpha}_{-i}\rangle$, a stopping strategy $\alpha_i$ exists such that $\sprg=\langle\alpha_i,\bm{\alpha}_{-i}\rangle$. 

\vspace{2mm}

\paragraph{Comparison to the Static Model}
Compared to \citeauthor{gentzkow2017competition}'s static model, information revelation is less in dynamic persuasion: The set of equilibrium outcomes in static persuasion is a subset of that in dynamic persuasion. Therefore, for any equilibrium distribution over posterior beliefs in dynamic persuasion, we can find a more informative equilibrium in static persuasion. These follow directly from the ex ante perspective in \Cref{MPE}, for senders can only deviate to MPS satisfying the support restriction in dynamic persuasion while they can choose any MPS in static persuasion. Hence, dynamic interactions limit the scope for competition and reduce information revelation due to the control-sharing effect. 

\vspace{2mm}

\paragraph{A Special Case}
\citet{gul2012war} study a similar, but more specific, dynamic model of war of information in which competition happens between two parties with opposite interests. I show in \Cref{GP12} that my model can also be applied to their setup and simplify the analysis.

\section{Concluding Remarks}
\label{conclusion}
I study collective dynamic information acquisition under general collective stopping rules and develop a method to characterize equilibria from an ex ante perspective via concavification. 
Given its generality and simplicity, I believe this model will prove useful for studying many other interesting problems than the applications I consider.



To conclude, I discuss some possible extensions of the current model.

\vspace{2mm}
\paragraph{Reversible vs. Irreversible Individual Stopping} In the baseline model, individual stopping is assumed to be reversible. 
What if players must stop irreversibly? For simplicity, consider unilateral and unanimous stopping. First, this changes nothing for unilateral stopping where individual stopping immediately triggers the termination of information acquisition. 

In contrast, under unanimous stopping, if players must stop irreversibly, 
they are able to choose any MPS of the distribution of posterior beliefs induced by others' stopping strategies, without the support restriction. As the set of deviations becomes larger, the set of PSMPE under irreversible stopping is smaller. Inefficiency is also more severe. This highlights the role of the control-sharing effect, which vanishes under irreversible stopping since players will get back full control after all other players stop. 


\vspace{2mm}
\paragraph{Beyond MPE} 
The ex ante perspective also has implications for Nash equilibrium. Every strategy profile induces a stopping time and a distribution over posteriors, and players still can only deviate to implement an MPC/MPS. However, for general strategy profiles, not every MPC/MPS is feasible, hence the ex ante perspective can only provide sufficient conditions for equilibria. For instance, every solution to Problem (\ref{unilateral}) or (\ref{unanimous}) with $p_0$ corresponds to a Nash equilibrium under unilateral or unanimous stopping with prior $p_0$; but the converse is not true. 

\vspace{2mm}
\paragraph{What Information to Acquire}
This model can be enriched to allow for players' information choice, as long as the belief process still satisfies Assumption \ref{assumption} and players use Markov strategies.\footnote{Since conclusive Poisson signals can be accommodated (see \Cref{Poisson}), we can also consider 
information choices of the kind studied by \citet{che2019optimal}.
} To illustrate, assume that $\mathrm{d}p_t=p_t(1-p_t)(2/\sigma_t)\mathrm{d}B_t$ as in \Cref{diffusion}, where 
$1/\sigma_t$ is the accuracy of signals. At each moment, players can invest in the level of experimentation 
\citep[cf.][]{moscarini2001optimal} via $1/\sigma_t\equiv R(e_{i,t},e_{-i,t})$ (e.g., by conducting more experiments) with flow costs $K_i(e_{i,t})$. If we focus on Markov strategies, then $e_{i,t}=e_i(p_t)$. 
Let $(e_i^*(\cdot))_{i\in N}$ denote the equilibrium effort profile.
Then the problem reverts back to a stopping problem in the baseline model with flow costs $\cost_i(p)+K_i(e_i^*(p))$. 

\vspace{2mm}
\paragraph{Many States}
The binary state assumption is critical for the results. 
In particular, when $|\Theta|>3$, not all distributions of posterior beliefs can be embedded by some stopping strategy (not mention Markov strategies), because 
continuous Markov martingales cannot hit every point in $\mathbb{R}^n$ almost surely when $n>2$. 
\citet{rost1971stopping} provides a necessary and sufficient condition for feasible distributions, but it is hard to work with.\footnote{See \citet[Proposition 7]{morris2019wald} or \citet[Section 3.5]{obloj2004skorokhod}. 
} 
Hence, concavification arguments analogous to \Cref{thm-g-collective} may only provide sufficient conditions for equilibria. An interesting question is whether the sets of feasible deviations are already rich enough to generate the relevant concave closures, so that those concavification conditions are also necessary.


\vspace{2mm}
\paragraph{Discounting} My model only allows for flow costs but not discounting. To incorporate discounting in this approach, one has to characterize feasible joint distributions over time and posterior beliefs that can be induced by some stopping strategy given the other player's strategy. A recent work by \citet{sannikov2024exploration} gives such a characterization in the single-agent case with flexible information acquisition. The challenging task to extend their results to a strategic setting is left for future research. 

{{\bibliographystyle{ecta}
\bibliography{sampling}}}

\begin{thebibliography}{38}
\newcommand{\enquote}[1]{``#1''}
\expandafter\ifx\csname natexlab\endcsname\relax\def\natexlab#1{#1}\fi

\bibitem[\protect\citeauthoryear{Albrecht, Anderson, and Vroman}{Albrecht et~al.}{2010}]{albrecht2010search}
\textsc{Albrecht, J., A.~Anderson, and S.~Vroman} (2010): \enquote{Search by committee,} \emph{Journal of Economic Theory}, 145, 1386--1407.

\bibitem[\protect\citeauthoryear{Anesi and Safronov}{Anesi and Safronov}{2023}]{anesi2023deciding}
\textsc{Anesi, V. and M.~Safronov} (2023): \enquote{Deciding When To Decide: Collective Deliberation and Obstruction,} \emph{International Economic Review}, 64, 757--781.

\bibitem[\protect\citeauthoryear{Arrow, Blackwell, and Girshick}{Arrow et~al.}{1949}]{arrow1949bayes}
\textsc{Arrow, K.~J., D.~Blackwell, and M.~A. Girshick} (1949): \enquote{Bayes and minimax solutions of sequential decision problems,} \emph{Econometrica, Journal of the Econometric Society}, 213--244.

\bibitem[\protect\citeauthoryear{Aumann}{Aumann}{1959}]{aumann1959acceptable}
\textsc{Aumann, R.~J.} (1959): \enquote{Acceptable points in general cooperative n-person games,} \emph{Contributions to the Theory of Games}, 4, 287--324.

\bibitem[\protect\citeauthoryear{Aumann, Maschler, and Stearns}{Aumann et~al.}{1995}]{aumann1995repeated}
\textsc{Aumann, R.~J., M.~Maschler, and R.~E. Stearns} (1995): \emph{Repeated games with incomplete information}, MIT press.

\bibitem[\protect\citeauthoryear{Austen-Smith and Banks}{Austen-Smith and Banks}{2000}]{austen2000positive}
\textsc{Austen-Smith, D. and J.~S. Banks} (2000): \emph{Positive political theory I: Collective preference}, vol.~1, University of Michigan Press.

\bibitem[\protect\citeauthoryear{Brocas and Carrillo}{Brocas and Carrillo}{2007}]{brocas2007influence}
\textsc{Brocas, I. and J.~D. Carrillo} (2007): \enquote{Influence through ignorance,} \emph{The RAND Journal of Economics}, 38, 931--947.

\bibitem[\protect\citeauthoryear{Brocas, Carrillo, and Palfrey}{Brocas et~al.}{2012}]{brocas2012information}
\textsc{Brocas, I., J.~D. Carrillo, and T.~R. Palfrey} (2012): \enquote{Information gatekeepers: Theory and experimental evidence,} \emph{Economic Theory}, 51, 649--676.

\bibitem[\protect\citeauthoryear{Chacon and Walsh}{Chacon and Walsh}{1976}]{chacon1976one}
\textsc{Chacon, R.~V. and J.~B. Walsh} (1976): \enquote{One-dimensional potential embedding,} \emph{S{\'e}minaire de probabilit{\'e}s de Strasbourg}, 10, 19--23.

\bibitem[\protect\citeauthoryear{Chan, Lizzeri, Suen, and Yariv}{Chan et~al.}{2018}]{chan2018deliberating}
\textsc{Chan, J., A.~Lizzeri, W.~Suen, and L.~Yariv} (2018): \enquote{Deliberating collective decisions,} \emph{The Review of Economic Studies}, 85, 929--963.

\bibitem[\protect\citeauthoryear{Che, Kim, and Mierendorff}{Che et~al.}{2023}]{che2023keeping}
\textsc{Che, Y.-K., K.~Kim, and K.~Mierendorff} (2023): \enquote{Keeping the listener engaged: a dynamic model of bayesian persuasion,} \emph{Journal of Political Economy}, 131, 1797--1844.

\bibitem[\protect\citeauthoryear{Che and Mierendorff}{Che and Mierendorff}{2019}]{che2019optimal}
\textsc{Che, Y.-K. and K.~Mierendorff} (2019): \enquote{Optimal dynamic allocation of attention,} \emph{American Economic Review}, 109, 2993--3029.

\bibitem[\protect\citeauthoryear{Compte and Jehiel}{Compte and Jehiel}{2010}]{compte2010bargaining}
\textsc{Compte, O. and P.~Jehiel} (2010): \enquote{Bargaining and majority rules: A collective search perspective,} \emph{Journal of Political Economy}, 118, 189--221.

\bibitem[\protect\citeauthoryear{Gentzkow and Kamenica}{Gentzkow and Kamenica}{2014}]{gentzkow2014costly}
\textsc{Gentzkow, M. and E.~Kamenica} (2014): \enquote{Costly persuasion,} \emph{American Economic Review}, 104, 457--62.

\bibitem[\protect\citeauthoryear{Gentzkow and Kamenica}{Gentzkow and Kamenica}{2017{\natexlab{a}}}]{gentzkow2017bayesian}
---\hspace{-.1pt}---\hspace{-.1pt}--- (2017{\natexlab{a}}): \enquote{Bayesian persuasion with multiple senders and rich signal spaces,} \emph{Games and Economic Behavior}, 104, 411--429.

\bibitem[\protect\citeauthoryear{Gentzkow and Kamenica}{Gentzkow and Kamenica}{2017{\natexlab{b}}}]{gentzkow2017competition}
---\hspace{-.1pt}---\hspace{-.1pt}--- (2017{\natexlab{b}}): \enquote{Competition in persuasion,} \emph{The Review of Economic Studies}, 84, 300--322.

\bibitem[\protect\citeauthoryear{Georgiadis and Szentes}{Georgiadis and Szentes}{2020}]{georgiadis2020optimal}
\textsc{Georgiadis, G. and B.~Szentes} (2020): \enquote{Optimal monitoring design,} \emph{Econometrica}, 88, 2075--2107.

\bibitem[\protect\citeauthoryear{Gul and Pesendorfer}{Gul and Pesendorfer}{2012}]{gul2012war}
\textsc{Gul, F. and W.~Pesendorfer} (2012): \enquote{The war of information,} \emph{The Review of Economic Studies}, 79, 707--734.

\bibitem[\protect\citeauthoryear{H{\'e}bert and Zhong}{H{\'e}bert and Zhong}{2022}]{hebert2022engagement}
\textsc{H{\'e}bert, B. and W.~Zhong} (2022): \enquote{Engagement Maximization,} \emph{arXiv e-prints}, arXiv--2207.

\bibitem[\protect\citeauthoryear{Henry and Ottaviani}{Henry and Ottaviani}{2019}]{henry2019research}
\textsc{Henry, E. and M.~Ottaviani} (2019): \enquote{Research and the approval process: The organization of persuasion,} \emph{American Economic Review}, 109, 911--55.

\bibitem[\protect\citeauthoryear{Kamenica and Gentzkow}{Kamenica and Gentzkow}{2011}]{kamenica2011bayesian}
\textsc{Kamenica, E. and M.~Gentzkow} (2011): \enquote{Bayesian persuasion,} \emph{American Economic Review}, 101, 2590--2615.

\bibitem[\protect\citeauthoryear{Lipnowski, Mathevet, and Wei}{Lipnowski et~al.}{2020}]{lipnowski2020attention}
\textsc{Lipnowski, E., L.~Mathevet, and D.~Wei} (2020): \enquote{Attention management,} \emph{American Economic Review: Insights}, 2, 17--32.

\bibitem[\protect\citeauthoryear{Lipnowski, Mathevet, and Wei}{Lipnowski et~al.}{2022}]{lipnowski2022optimal}
---\hspace{-.1pt}---\hspace{-.1pt}--- (2022): \enquote{Optimal attention management: A tractable framework,} \emph{Games and Economic Behavior}, 133, 170--180.

\bibitem[\protect\citeauthoryear{Liptser and Shiryaev}{Liptser and Shiryaev}{2013}]{liptser2013statistics}
\textsc{Liptser, R.~S. and A.~N. Shiryaev} (2013): \emph{Statistics of Random Processes: I. General Theory}, vol.~5, Springer Science \& Business Media.

\bibitem[\protect\citeauthoryear{Matyskova and Montes}{Matyskova and Montes}{2023}]{matyskova2023bayesian}
\textsc{Matyskova, L. and A.~Montes} (2023): \enquote{Bayesian persuasion with costly information acquisition,} \emph{Journal of Economic Theory}, 105678.

\bibitem[\protect\citeauthoryear{Moldovanu and Shi}{Moldovanu and Shi}{2013}]{moldovanu2013specialization}
\textsc{Moldovanu, B. and X.~Shi} (2013): \enquote{Specialization and partisanship in committee search,} \emph{Theoretical Economics}, 8, 751--774.

\bibitem[\protect\citeauthoryear{Morris and Strack}{Morris and Strack}{2019}]{morris2019wald}
\textsc{Morris, S. and P.~Strack} (2019): \enquote{The wald problem and the relation of sequential sampling and ex-ante information costs,} \emph{Available at SSRN 2991567}.

\bibitem[\protect\citeauthoryear{Moscarini and Smith}{Moscarini and Smith}{2001}]{moscarini2001optimal}
\textsc{Moscarini, G. and L.~Smith} (2001): \enquote{The optimal level of experimentation,} \emph{Econometrica}, 69, 1629--1644.

\bibitem[\protect\citeauthoryear{Ob{\l}{\'o}j}{Ob{\l}{\'o}j}{2004}]{obloj2004skorokhod}
\textsc{Ob{\l}{\'o}j, J.} (2004): \enquote{The Skorokhod embedding problem and its offspring,} \emph{Probability Surveys}, 1, 321--392.

\bibitem[\protect\citeauthoryear{Revuz and Yor}{Revuz and Yor}{2013}]{revuz2013continuous}
\textsc{Revuz, D. and M.~Yor} (2013): \emph{Continuous martingales and Brownian motion}, vol. 293, Springer Science \& Business Media.

\bibitem[\protect\citeauthoryear{Rost}{Rost}{1971}]{rost1971stopping}
\textsc{Rost, H.} (1971): \enquote{The stopping distributions of a Markov process,} \emph{Inventiones mathematicae}, 14, 1--16.

\bibitem[\protect\citeauthoryear{Sannikov and Zhong}{Sannikov and Zhong}{2024}]{sannikov2024exploration}
\textsc{Sannikov, Y. and W.~Zhong} (2024): \enquote{Exploration and Stopping,} .

\bibitem[\protect\citeauthoryear{Strulovici}{Strulovici}{2010}]{strulovici2010learning}
\textsc{Strulovici, B.} (2010): \enquote{Learning while voting: Determinants of collective experimentation,} \emph{Econometrica}, 78, 933--971.

\bibitem[\protect\citeauthoryear{Titova}{Titova}{2021}]{titova2019collaborative}
\textsc{Titova, M.} (2021): \enquote{Collaborative search for a public good,} .

\bibitem[\protect\citeauthoryear{Wald}{Wald}{1947}]{wald1947foundations}
\textsc{Wald, A.} (1947): \enquote{Foundations of a general theory of sequential decision functions,} \emph{Econometrica, Journal of the Econometric Society}, 279--313.

\bibitem[\protect\citeauthoryear{Wei}{Wei}{2021}]{wei2021persuasion}
\textsc{Wei, D.} (2021): \enquote{Persuasion under costly learning,} \emph{Journal of Mathematical Economics}, 94, 102451.

\bibitem[\protect\citeauthoryear{Zhong}{Zhong}{2022}]{zhong2022optimal}
\textsc{Zhong, W.} (2022): \enquote{Optimal dynamic information acquisition,} \emph{Econometrica}, 90, 1537--1582.

\bibitem[\protect\citeauthoryear{Zhou}{Zhou}{2024}]{zhou2023collective}
\textsc{Zhou, Y.} (2024): \enquote{Collective Sampling: An Ex Ante Perspective,} \emph{arXiv preprint arXiv:2311.05758, v2 (24 Jun 2024)}.

\end{thebibliography}

\appendix

\section{Proofs for the Baseline model}
\label{proofs}
\subsection{Skorokhod Embedding}
\label{proof-skorokhod}
I start with implementable distributions of posterior distributions in the single-agent case. Fix a prior $p_0$. Let $F_\tau$ denote the distribution over posterior beliefs induced by stopping time $\tau$.
\begin{lemma}
\label{skorokhod}
	A stopping time $\tau$ exists s.t. $F_\tau=F$ if and only if $\mathbb{E}_F[p]=p_0$.
\end{lemma}
This is related to the classical Skorokhod embedding problem \citep[see][for an excellent survey]{obloj2004skorokhod}.
    One proof of this result \citep[see Proposition 5.1,][]{obloj2004skorokhod} is given by \citet{chacon1976one} via potential theory. 
\begin{proof}
    Define $M_t:=p_t-p_0$. Hence, every Bayes plausible distribution over posterior beliefs $F$ corresponds to a centered measure $\mu$ over $[-p_0,1-p_0]$.

	The potential of a centered probability measure $\mu$ on $\mathbb{R}$ is defined as $U\mu(x)=-\int_\mathbb{R}|x-y|\mathrm{d}\mu(y)$; note that $U\mu$ is concave. For a sequence of measures $(\mu_n)$, if $U\mu_n(x)\to U\mu(x)$, then $\mu_n\Rightarrow\mu$ (weak convergence). Therefore, in order to embed $\mu$, I can instead find a sequence of measures $\mu_n$ and stopping times $\tau_n$ such that $U\mu_n\to U\mu$ and $M_{\tau_n}\sim\mu_n$. If $\tau_n$ converges to some stopping time $\tau$ a.s., then we have $M_\tau\sim\mu$. 

	
    Let $\mu_0=\delta_0$ be the Dirac measure at $0$. Choose a point between the graphs of $U\mu_0$ ($U\mu_0(x)=-|x|$) and $U\mu$, and draw a line $l_1$ through this point that stays above $U\mu$, which is possible since $U\mu$ is concave. This line cuts the potential $U\mu_0$ in two points $a_1<0<b_1$. Consider the new potential $U\mu_1$ which is $U\mu_0$ on $(-\infty,a_1]\cup[b_1,\infty)$ and linear on $[a_1,b_1]$. Choose a point between $U\mu_1$ and $U\mu$, and draw a line $l_2$ through this point that stays above $U\mu$ and cuts $U\mu_1$ in two points $a_2$ and $b_2$ such that either $a_2<a_1<b_2$ or $a_2<b_1<b_2$. Iterate this procedure. 
    We can draw lines in such a way that $U\mu_n\to U\mu$, because $U\mu$ is a concave function that can be represented as the infimum of countably many affine functions. 
 
    The last step is to find the corresponding stopping times $\tau_n$ and the limit. Consider first escape times $\tau_{a,b}=\inf\{t\geq0:M_t\notin(a,b)\}$ for $a<0<b$. 
    As $(p_t)_{t\geq0}$ is a continuous martingale and $p_\infty\in\{0,1\}$, $\tau_{a,b}\leq\infty$ a.s. and $M_{\tau_{a,b}}$ can only take values at $a$ or $b$.  Since $(M_{\tau_{a,b}\wedge t})_{t\geq0}$ is bounded, by the Martingale Convergence Theorem, $0=\mathbb{E}[M_{\tau_{a,b}}]=\mathbb{P}(M_{\tau_{a,b}}=a)\cdot a+(1-\mathbb{P}(M_{\tau_{a,b}}=a))\cdot b$, which implies $\mathbb{P}(M_{\tau_{a,b}}=a)=\frac{b}{b-a}$. Therefore, $U\mu_{\tau_{a,b}}$ is piecewise linear. Let $\tau_1=\tau_{a_1,b_1}$, then $M_{\tau_1}\sim\mu_1$. Let $\theta_t$ denote the standard shift operator and define $\tau_2=\tau_1+\tau_{a_2,b_2}\circ\theta_{\tau_1},\dots,\tau_n=\tau_{n-1}+\tau_{a_n,b_n}\circ\theta_{\tau_{n-1}},\dots$ and $\tau=\lim_{n\to\infty}\tau_n$. One can check that the limit is finite almost surely. As a result, one can embed $\mu$ via $\tau$.
\end{proof}

\begin{lemma}
\label{MPS-st}
If $\tau\geq\rho$ a.s., then $F_\tau$ is an MPS of $F_\rho$. Conversely, with $\rho$ fixed, if $F$ is an MPS of $F_\rho$, there exists $\tau\geq\rho$ a.s. such that $F_\tau=F$.
\end{lemma}
\begin{proof}
    For the first part, for $(p_t)$ is a martingale and $\mathcal{F}_{\rho}\subseteq\mathcal{F}_{\tau}$, by Jensen's inequality, $\mathbb{E}[(p-p_{\tau})^+|\mathcal{F}_\rho]\geq (p-\mathbb{E}[p_{\tau}|\mathcal{F}_\rho])^+=(p-p_\rho)^+$. Hence, $\mathbb{E}[(p-p_{\tau})^+]=\mathbb{E}[\mathbb{E}[(p-p_{\tau})^+|\mathcal{F}_\rho]]\geq\mathbb{E}[(p-p_\rho)^+]$ for $p\in[0,1]$. This is equivalent to $\int_0^p F_{\tau}(x)dx\geq\int_0^p F_{\rho}(x)dx$ for any $p\in[0,1]$, so $F_{\tau}$ is an MPS of $F_\rho$.

    For the second part, since $F$ is an MPS of $F_\rho$, $U\mu_F(x)=-\int_{\mathbb{R}}|x-y|\mathrm{d}F(y)\leq-\int_{\mathbb{R}}|x-y|\mathrm{d}F_\rho(y)=U\mu_{F_\rho}(x)$. Construct a stochastic process $(\tilde{p}_t)$ with $\tilde{p}_0\sim F_\rho$ and $\tilde{p}_t=p_t(\tilde{p}_0)$ for any realization of $\tilde{p}_0$. Then by Proposition 5.2 in \citet{obloj2004skorokhod}, there exists a stopping time $\tau'$ such that $\tilde{p}_{\tau'}\sim F$. Let $\tau=\rho+\tau'\circ\theta_\rho$. It is easy to verify that $\tau$ is also a stopping time and $p_\tau\sim F$, i.e., $F_{\tau}=F$.
\end{proof}
\begin{lemma}
\label{MPC-st}
	If $\tau\leq\rho$ a.s., $F_{\tau}$ is an MPC of $F_\rho$. 
	Conversely, if $F$ is an MPC of $H$ and $\mathbb{E}_F[p]=\mathbb{E}_H[p]=p_0$,  $\tau\leq\rho$ exist such that $F_\tau=F$ and $F_{\rho}=H$.
\end{lemma}
\begin{proof}
	The first part follows from \Cref{MPS-st}. For the second part, according to \Cref{skorokhod}, there exists a stopping time $\tau$ such that $F_\tau=F$. Let $\rho=\tau+\tau'\circ\theta_{\tau}$ where $\tau'$ is a stopping time such that $\tilde{p}_{\tau'}\sim F$ with $\tilde{p}_0\sim F_\rho$; refer to the second part in the proof of \Cref{MPS-st}.
\end{proof}


In general, with $\rho$ fixed, if $F$ is an MPC of $F_\rho$, there does not necessarily exist $\tau\leq\rho$ a.s. such that $F_\tau=F$; see an example in Supplementary \Cref{counterexample-MPC}.
However, such a $\tau$ does exist when $\rho$ is a first escape time.
\begin{lemma}
\label{MPC-markov-st}

    Suppose that $\rho=\inf\{t\geq0:p_t\notin(\underline{p},\overline{p})\}$. With $\rho$ fixed, if $F$ is an MPC of $F_\rho$, there exists $\tau\leq\rho$ a.s. such that $F_\tau=F$.
\end{lemma}
\begin{proof}
For any MPC $F$ of $F_\rho$, the Chacon and Walsh solution $\tau$ described in the proof of \Cref{skorokhod} implements $F_\tau=F$ and $\tau\leq\rho$. Recall that the sequence of $a_n$ and $b_n$ is constructed by some lines above $U\mu_F$. As $F$ is an MPC of $F_\rho$, we have $U\mu_{F_\rho}\leq U\mu_F$, so those lines are also above $U\mu_{F_\rho}$ and $a_n,b_n\in[\underline{p},\overline{p}]$. Therefore, by construction, $\tau_n\leq\rho$ for all $n$. Given that $\tau_n\to\tau$ a.s., $\tau\leq\rho$ a.s. too.
\end{proof}

\subsection{Ex Ante Reformulation}
\label{proof-ex-ante}
Here I establish the ex ante reformulation under $\mathcal{\group}$-collective stopping, which proves the results in \Cref{sect-reformulation}.

Focus on an arbitrary player $i$. 
Suppose that all other players are using pure Markov strategies $\bm{\alpha}_{-i}$ with sampling regions $(\sprg_j)_{j\ne i}$. Consider the stopping times induced by $\bm{\alpha}_{-i}$ and either $\alpha_i\equiv1$ (i.e., always stopping) or $\alpha_i\equiv0$ (i.e., never stopping) and define
\begin{align*}
\tau^{\mathcal{\group}-i}(\bm{\alpha}_{-i}):=&\tau^{\mathcal{\group}}(\alpha_i\equiv1,\bm{\alpha}_{-i})=\inf\left\{t\geq0:p_t\notin\stpbf_i^{\mathcal{\group}}(\bm{\sprg}_{-i})\right\}\\
\tau^{\mathcal{\group}\setminus\mathcal{\group}_i}(\bm{\alpha}_{-i}):=&\tau^{\mathcal{\group}}(\alpha_i\equiv0,\bm{\alpha}_{-i})=\inf\left\{t\geq0:p_t\notin\cont_i^{\mathcal{\group}}(\bm{\sprg}_{-i})\right\}
\end{align*}
where $\stpbf_i^{\mathcal{\group}}(\bm{\sprg}_{-i})$ and $\cont_i^{\mathcal{\group}}(\bm{\sprg}_{-i})$ are as defined in \Cref{eq:supp}. 

By definition, $\tau^{\mathcal{\group}-i}(\bm{\alpha}_{-i})\leq\tau^{\mathcal{\group}}(\alpha_i,\bm{\alpha}_{-i})\leq\tau^{\mathcal{\group}\setminus\mathcal{\group}_i}(\bm{\alpha}_{-i})$ for any stopping strategy $\alpha_i$ (even non-Markov), therefore by \Cref{MPS-st,MPC-st}, for a fixed prior $p_0$, 
$F_{\tau^{\mathcal{\group}}(\alpha_i,\bm{\alpha}_{-i})}$ is an MPS of $F_{\tau^{\mathcal{\group}-i}(\bm{\alpha}_{-i})}$ and an MPC of $F_{\tau^{\mathcal{\group}\backslash\mathcal{\group}_i}(\bm{\alpha}_{-i})}$. Notice that 
$\supp(F_{\tau^{\mathcal{\group}\backslash\mathcal{\group}_i}(\bm{\alpha}_{-i})})=\{\underline{p}{}_i^{\mathcal{\group}}(p_0;\bm{\sprg}_{-i}),\overline{p}{}_i^{\mathcal{\group}}(p_0;\bm{\sprg}_{-i})\}$. Hence, the support of $F_{\tau^{\mathcal{\group}}(\alpha_i,\bm{\alpha}_{-i})}$, as an MPC of $F_{\tau^{\mathcal{\group}\backslash\mathcal{\group}_i}(\bm{\alpha}_{-i})}$, must be contained in $[\underline{p}{}_i^{\mathcal{\group}}(p_0;\bm{\sprg}_{-i}),\overline{p}{}_i^{\mathcal{\group}}(p_0;\bm{\sprg}_{-i})]\subset\overline{\cont_i^{\mathcal{\group}}(\bm{\sprg}_{-i})}$; moreover, it is also contained in $[0,1]\setminus\stpbf_i^{\mathcal{\group}}(\bm{\sprg}_{-i})$ by definition of $\tau^{\mathcal{\group}}(\bm{\alpha})$. 
Hence, $F_{\tau^{\mathcal{\group}}(\alpha_i,\bm{\alpha}_{-i})}$ is an MPS of $F_{\tau^{\mathcal{\group}-i}(\bm{\alpha}_{-i})}$ and an MPC of $F_{\tau^{\mathcal{\group}\backslash\mathcal{\group}_i}(\bm{\alpha}_{-i})}$, with support in $\overline{\cont_i^{\mathcal{\group}}(\bm{\sprg}_{-i})}\setminus\stpbf_i^{\mathcal{\group}}(\bm{\sprg}_{-i})$.

\vspace{2mm}
\paragraph{$\mathcal{\group}$-collective Stopping} A distribution of posterior beliefs $F$ is a solution to Problem (\ref{G-collective}) with $(p_0,\underline{F}{}_i,\overline{F}_i,E_i)$ if
\begin{equation}
\begin{aligned}
\label{G-collective}
	F\in&\argmax_{H} \int_0^1[u_i(p)-\phi_i(p)]\mathrm{d}H(p)\\
 &\text{s.t. }H \text{ is an MPS of }\underline{F}{}_i \text{ and an MPC of }\overline{F}_i,\supp(H)\subset E_i.
\end{aligned}
\tag{S-$\mathcal{\group}$-$i$}
\end{equation}
\begin{theorem}
\label{MPE-general}
    A profile of pure Markov strategies $\bm{\alpha}$ is sequentially rational for player $i$ at $p_t=p$ under $\mathcal{\group}$-collective stopping if and only if $F_{\tau^{\mathcal{\group}}(\bm{\alpha})}$ is a binary policy solution to Problem (\ref{G-collective}) with $(p_0,\underline{F}{}_i,\overline{F}_i,E_i)$, where $p_0=p$, $\underline{F}{}_i=F_{\tau^{\mathcal{\group}-\{i\}}(\bm{\alpha}_{-i})}$, $\overline{F}_i=F_{\tau^{\mathcal{\group}\backslash\mathcal{\group}_i}(\bm{\alpha}_{-i})}$, and $E_i=\overline{\cont_i^{\mathcal{\group}}(\bm{\alpha}_{-i})}\setminus\stpbf_i^{\mathcal{\group}}(\bm{\alpha}_{-i})$.
\end{theorem}
\Cref{MPE} is a special case of this result. To prove this equivalence, we need \Cref{cost} and another lemma (\Cref{MPC-MPS}).
\begin{proof}[Proof of \Cref{cost}]
    
    By Assumption \ref{assumption}, $\mathrm{d}\langle p\rangle_s/\mathrm{d}s$ is well-defined, i.e.,
	\[
	\frac{\mathrm{d}\langle p\rangle_s}{\mathrm{d}s}=\lim_{h\to0}\frac{\langle p\rangle_{s+h}-\langle p\rangle_s}{h}=\lim_{h\to0}\frac{\langle p\rangle_{s}-\langle p\rangle_{s-h}}{h}\quad a.s.
	\]
	Notice that $\mathrm{d}\langle p\rangle_s/\mathrm{d}s$ is locally integrable given Condition (iv): For any $t'>t\geq0$, by Tonelli's theorem,
	\[
	\int_{[t,t']}\mathbb{E}\left[\left|\mathrm{d}\langle p\rangle_s/\mathrm{d}s\right|\right]\mathrm{d}s=\mathbb{E}[\langle p\rangle_{t'}-\langle p\rangle_{t}]\leq\mathbb{E}[\langle p\rangle_{t'}]<\infty.
	\]
	Therefore, I can interchange differentiation and expectation as follows:
    {\small
	\begin{align*}
		\frac{\mathrm{d}\langle p\rangle_s}{\mathrm{d}s}=\lim_{h\to0}\frac{\langle p\rangle_{s}-\langle p\rangle_{s-h}}{h}=\lim_{h\to0}\mathbb{E}\left[\frac{\langle p\rangle_{s}-\langle p\rangle_{s-h}}{h}|\mathcal{F}_s\right]=\mathbb{E}\left[\lim_{h\to0}\frac{\langle p\rangle_{s}-\langle p\rangle_{s-h}}{h}|\mathcal{F}_s\right]\\
		=\mathbb{E}\left[\lim_{h\to0}\frac{\langle p\rangle_{s+h}-\langle p\rangle_{s}}{h}|\mathcal{F}_s\right]=\lim_{h\to0}\mathbb{E}\left[\frac{\langle p\rangle_{s+h}-\langle p\rangle_{s}}{h}|\mathcal{F}_s\right]=\lim_{h\to0}\mathbb{E}\left[\frac{\langle p\rangle_{s+h}-\langle p\rangle_{s}}{h}|p_s\right].
	\end{align*}}
	The last equality is because $(p_t)$ is Markov. In sum, $\mathrm{d}\langle p\rangle_s/\mathrm{d}s$ is $\sigma(p_s)$-measurable, i.e., a function $\quadvar$ exists such that $\mathrm{d}\langle p\rangle_s/\mathrm{d}s=\quadvar(p_s)$ a.s.. By Condition (iv) in Assumption \ref{assumption}, $\quadvar(p_s)$ must be strictly positive a.s..
    
    Let $\tau$ be an arbitrary a.s. finite stopping time. Applying Ito's lemma,
	\begin{equation*}
		\phi_i(p_t)=\phi_i(p_0)+\int_0^t\phi_i'(p_s)\mathrm{d}p_s+\int_0^t\frac12\phi_i''(p_s)\mathrm{d}\langle p\rangle_s.
	\end{equation*}
	By \Cref{static cost}, $\phi_i''(p_s)=2\cost_i(p_s)/\quadvar(p_s)$.  Hence, 
	\[
	\phi_i(p_t)=\phi_i(p_0)+\int_0^t\phi_i'(p_s)\mathrm{d}p_s+\int_0^t\cost_i(p_s)\mathrm{d}s.
	\]
    Since $(p_t)$ is a continuous martingale, $\int_0^{t\wedge r_n}\phi_i'(p_s)\mathrm{d}p_s$ is also a martingale. Intuitively, by the Optional Stopping Theorem, the expectation of this integral at $\tau$ is equal to zero. Because the stochastic integral is not necessarily integrable and $\tau$ is not necessarily finite, the rigorous argument has to resort to finite approximations of $\tau$, which is the same as that in the proof of Proposition 2 in \citet{morris2019wald} and omitted for brevity.
\end{proof}
\begin{lemma}
\label{MPC-MPS}
    Fix a prior $p_0$, a player $i$, and pure Markov strategies $\bm{\alpha}_{-i}$. A strategy $\alpha_i$ exists such that $F_{\tau^{\mathcal{\group}}(\alpha_i,\bm{\alpha}_{-i})}=F$ if and only if $F$ is an MPS of $F_{\tau^{\mathcal{\group}-i}(\bm{\alpha}_{-i})}$ and an MPC of $F_{\tau^{\mathcal{\group}\backslash\mathcal{\group}_i}(\bm{\alpha}_{-i})}$ with support in $\overline{\cont_i^{\mathcal{\group}}(\bm{\sprg}_{-i})}\setminus\stpbf_i^{\mathcal{\group}}(\bm{\sprg}_{-i})$.
\end{lemma}
\begin{proof}
    Since the previous discussion already shows the ``only if'' direction, it remains to show the ``if'' direction. Recall that $F_{\tau^{\mathcal{\group}\backslash\mathcal{\group}_i}(\bm{\alpha}_{-i})}$ is a binary policy with support $\{\underline{p}{}_i^{\mathcal{\group}}(p_0;\bm{\sprg}_{-i}),\overline{p}{}_i^{\mathcal{\group}}(p_0;\bm{\sprg}_{-i})\}\subset\overline{\cont_i^{\mathcal{\group}}(\bm{\sprg}_{-i})}$. Hence, any $F$ satisfying the condition must have its support in $[\underline{p}{}_i^{\mathcal{\group}}(p_0;\bm{\sprg}_{-i}),\overline{p}{}_i^{\mathcal{\group}}(p_0;\bm{\sprg}_{-i})]\setminus\stpbf_i^{\mathcal{\group}}(\bm{\sprg}_{-i})$. Then according to \Cref{MPS-st,MPC-markov-st}, a stopping time $\tau$ exists such that $\tau^{\mathcal{\group}-i}(\bm{\alpha}_{-i})\leq\tau\leq\tau^{\mathcal{\group}\backslash\mathcal{\group}_i}(\bm{\alpha}_{-i})$ and $F_\tau=F$. In particular, $\tau$ can be $\tau^{\mathcal{\group}-i}(\bm{\alpha}_{-i})+\tau'\circ\theta_{\tau^{\mathcal{\group}-i}(\bm{\alpha}_{-i})}$ where $\tau'$ is the Chacon and Walsh solution that embeds $F$ with $p_0\sim F_{\tau^{\mathcal{\group}-i}(\bm{\alpha}_{-i})}$ (i.e., augmenting $\tau^{\mathcal{\group}-i}(\bm{\alpha}_{-i})$ with the Chacon and Walsh solution). Because the support of $F$ is contained in $[\underline{p}{}_i^{\mathcal{\group}}(p_0;\bm{\sprg}_{-i}),\overline{p}{}_i^{\mathcal{\group}}(p_0;\bm{\sprg}_{-i})]$, it holds that the Chacon and Walsh solution $\tau'\leq\tau^{\mathcal{\group}\backslash\mathcal{\group}_i}(\bm{\alpha}_{-i})-\tau^{\mathcal{\group}-i}(\bm{\alpha}_{-i})$. It completes the proof.
\end{proof}

\begin{proof}[Proof of \Cref{MPE-general}]
    Under $\mathcal{\group}$-collective stopping, given $\bm{\alpha}_{-i}$, player $i$'s expected continuation payoff of using $\alpha_i$ at time $t$ given $p_t=p$ is 
	\begin{align*}
    &\mathbb{E}\left[u_i(p_{{\tau}^{\mathcal{\group}}(\alpha_i,\bm{\alpha}_{-i})})\mathbbm{1}_{{\tau}^{\mathcal{\group}}(\alpha_i,\bm{\alpha}_{-i})<\infty}-\int_0^{{\tau}^{\mathcal{\group}}(\alpha_i,\bm{\alpha}_{-i})} \cost_i(p_s)\mathrm{d}s\bigg|p_0=p\right]\\
    &=\int_0^1[u_i(\tilde{p})-\phi_i(\tilde{p})]\mathrm{d}F_{{\tau}^{\mathcal{\group}}(\alpha_i,\bm{\alpha}_{-i})}(\tilde{p})+\phi_i(p),
	\end{align*}
	by \Cref{cost}. According to \Cref{MPC-MPS}, $\alpha_i$ is a best response of player $i$ to $\bm{\alpha}_{-i}$ at time $t$ 
    if and only if $F_{{\tau}^{\mathcal{\group}}(\alpha_i,\bm{\alpha}_{-i})}$ (which is a binary policy) 
    is her best choice among all distributions that are MPS of $F_{\tau^{\mathcal{\group}-i}(\bm{\alpha}_{-i})}$ and MPC of $F_{\tau^{\mathcal{\group}\backslash\mathcal{\group}_i}(\bm{\alpha}_{-i})}$ with support in $\overline{\cont_i^{\mathcal{\group}}(\bm{\sprg}_{-i})}\setminus\stpbf_i^{\mathcal{\group}}(\bm{\sprg}_{-i})$. Therefore, $\alpha_i$ is sequentially rational for player $i$ at $p_t=p$ under unilateral stopping if and only if $F_{{\tau}^{\mathcal{\group}}(\alpha_i,\bm{\alpha}_{-i})}$ is a binary policy solution to Problem (\ref{G-collective}) with $p_0=p$, $\underline{F}{}_i=F_{\tau^{\mathcal{\group}-\{i\}}(\bm{\alpha}_{-i})}$, $\overline{F}_i=F_{\tau^{\mathcal{\group}\backslash\mathcal{\group}_i}(\bm{\alpha}_{-i})}$, and $E_i=\overline{\cont_i^{\mathcal{\group}}(\bm{\alpha}_{-i})}\setminus\stpbf_i^{\mathcal{\group}}(\bm{\alpha}_{-i})$.
\end{proof}

\subsection{Equilibrium Analysis}

\begin{proof}[Proof of \Cref{concave,thm-g-collective}] 
Conditions (i) and (ii) are summarized as
\[
U_i\left(\underline{p}(p;\sprg),\overline{p}(p;\sprg);p\right)=V_i^{\mathcal{\group}}\left(p;\bm{\sprg}_{-i}\right),\quad\text{for all }p\in[0,1],
\]
where $U_i$ is defined in \Cref{ex-ante-payoff} with $\underline{p}(p;\sprg):=\sup\{p':p'\in[0,p]\setminus\sprg\}$ and $\overline{p}(p;\sprg):=\inf\{p':p'\in[p,1]\setminus\sprg\}$.
 
For the ``if'' part, fix an arbitrary $p\in[0,1]$. Let $\underline{p}:=\underline{p}(p;\sprg)$, $\overline{p}:=\overline{p}(p;\sprg)$, and $F_{\underline{p},\overline{p}}$ be the binary policy with support $\{\overline{p},\underline{p}\}$ and $\mathbb{E}_{F_{\underline{p},\overline{p}}}[\tilde{p}]=p$. Let $\overline{F}_i$ denote the binary policy with support $\{\underline{p}{}_i^{\mathcal{\group}}(p;\bm{\sprg}_{-i}),\overline{p}{}_i^{\mathcal{\group}}(p;\bm{\sprg}_{-i})\}$ and $\mathbb{E}_{\overline{F}_i}[\tilde{p}]=p$. By construction, $V_i^{\mathcal{\group}}(p;\bm{\sprg}_{-i})\geq\int_0^1[u_i(\tilde{p})-\phi_i(\tilde{p})]\mathrm{d}F(\tilde{p})$ for any $F$ that is an MPC of $\overline{F}_i$ with support in $\overline{\cont_i^{\mathcal{\group}}(\bm{\sprg}_{-i})}\setminus\stpbf_i^{\mathcal{\group}}(\bm{\sprg}_{-i})$. Note that $\overline{F}_i=F_{\tau^{\mathcal{\group}\backslash\mathcal{\group}_i}(\bm{\alpha}_{-i})}$, and any $F$ with $\supp(F)\subset\overline{\cont_i^{\mathcal{\group}}(\bm{\sprg}_{-i})}\setminus\stpbf_i^{\mathcal{\group}}(\bm{\sprg}_{-i})$ must be an MPS of $F_{\tau^{\mathcal{\group}-\{i\}}(\bm{\alpha}_{-i})}$. Since $\int_0^1[u_i(\tilde{p})-\phi_i(\tilde{p})]\mathrm{d}F_{\underline{p},\overline{p}}(\tilde{p})=U_i(\underline{p},\overline{p};p)=V_i^{\mathcal{\group}}(p;\bm{\sprg}_{-i})$ where the second equation is by condition (i) or (ii), $F_{\underline{p},\overline{p}}$ is optimal among all distributions that are MPS of $F_{\tau^{\mathcal{\group}-i}(\bm{\alpha}_{-i})}$ and MPC of $F_{\tau^{\mathcal{\group}\backslash\mathcal{\group}_i}(\bm{\alpha}_{-i})}$ with support in $\overline{\cont_i^{\mathcal{\group}}(\bm{\sprg}_{-i})}\setminus\stpbf_i^{\mathcal{\group}}(\bm{\sprg}_{-i})$. 
\Cref{MPE-general} then implies that the strategy profile associated with $\bm{\sprg}$, i.e., $\bm{\alpha}$ where $\alpha_i(p)=\mathbbm{1}_{p\in\sprg_i}$, is sequentially rational for player $i$ at $p_t=p$ since $F_{{\tau}^{\mathcal{\group}}(\bm{\alpha})}=F_{\underline{p},\overline{p}}$. This is true for any $p$ and any $i\in N$, hence $\bm{\sprg}$ is an equilibrium.

For the ``only if'' part, again fix an arbitrary $p\in [0,1]$ with $\underline{p}=\underline{p}(p;\sprg)$ and $\overline{p}=\overline{p}(p;\sprg)$. By definition, $U_i(\underline{p},\overline{p};p)\leq V_i^{\mathcal{\group}}(p;\bm{\sprg}_{-i})$ for all $i\in N$. Suppose that one of the inequalities is strict, say $U_1(\underline{p},\overline{p};p)<V_1^{\mathcal{\group}}(p;\bm{\sprg}_{-1})$. By the definition of concave closure, for any $\epsilon>0$, there exists a distribution $F$ over $[\underline{p}{}_i^{\mathcal{\group}}(p;\bm{\sprg}_{-i}),\overline{p}{}_i^{\mathcal{\group}}(p;\bm{\sprg}_{-i})]\setminus\stpbf_i^{\mathcal{\group}}(\bm{\sprg}_{-i})\subset\overline{\cont_i^{\mathcal{\group}}(\bm{\sprg}_{-i})}\setminus\stpbf_i^{\mathcal{\group}}(\bm{\sprg}_{-i})$ such that $\mathbb{E}_F[\tilde{p}]=p$ and $\int_{0}^1[u_1(\tilde{p})-\phi_1(\tilde{p})]\mathrm{d}F(\tilde{p})\in\big(\underline{p},\overline{p};p),V_1^{\mathcal{\group}}(p;\bm{\sprg}_{-1})\big]$. 
Hence, $F$ is strictly better than $F_{\underline{p},\overline{p}}$ for player 1. Since $F$ must be an MPS of $F_{\tau^{\mathcal{\group}-\{i\}}(\bm{\alpha}_{-i})}$ 
and an MPC of $F_{\tau^{\mathcal{\group}\backslash\mathcal{\group}_i}(\bm{\alpha}_{-i})}$, 
this contradicts to $F_{\underline{p},\overline{p}}=F_{{\tau}^{\mathcal{\group}}(\bm{\alpha})}$ where $\alpha_i(p)=\mathbbm{1}_{p\in\sprg_i}$ being a solution to Problem (\ref{G-collective}) with $p_0=p$, $\underline{F}{}_i=F_{\tau^{\mathcal{\group}-\{i\}}(\bm{\alpha}_{-i})}$, $\overline{F}_i=F_{\tau^{\mathcal{\group}\backslash\mathcal{\group}_i}(\bm{\alpha}_{-i})}$, and $E_i=\overline{\cont_i^{\mathcal{\group}}(\bm{\alpha}_{-i})}\setminus\stpbf_i^{\mathcal{\group}}(\bm{\alpha}_{-i})$ by \Cref{MPE-general}. Therefore, conditions (i) and (ii) must hold.
\end{proof}

Before proving the comparative statics, it is useful to introduce some lemmas. \Cref{reduction-uni-una} shows that we only need to consider players in $N_{uni}(\mathcal{\group})$ and $N_{una}(\mathcal{\group})$. \Cref{lattice-general} establishes the semi-lattice structure.
\begin{lemma}
\label{reduction-uni-una}
    $\sprg$ is an equilibrium under $\mathcal{\group}$-collective stopping if and only if for any $p\in[0,1]$, $U_i(\underline{p}(p;\sprg),\overline{p}(p;\sprg);p)=V_i^{IN}(p)$ for $i\in N_{uni}(\mathcal{\group})$ and $U_j(\underline{p}(p;\sprg),\overline{p}(p;\sprg);p)=V_j^{OUT}(p)$ for $j\in N_{una}(\mathcal{\group})$.
\end{lemma}
\begin{proof}[Proof of \Cref{reduction-uni-una}]
    Fix an equilibrium $(\sprg_i)_{i\in N}$ with $\sprg=\cap_{\group\in\mathcal{\group}}\cup_{j\in \group}\sprg_j$. When all players use $\widetilde{\sprg}_i=\sprg$, for any player $i\not\in N_{uni}(\mathcal{\group})\cup N_{una}(\mathcal{\group})$, $\cont_i^{\mathcal{\group}}(\bm{\widetilde\sprg}_{-i})=\stpbf_i^{\mathcal{\group}}(\bm{\widetilde\sprg}_{-i})=\sprg$. Therefore, we only need to consider the incentive of players in $N_{uni}(\mathcal{\group})\cup N_{una}(\mathcal{\group})$. The ``if'' part is thus straightforward as $\bm{\tilde{\sprg}}$ is then an equilibrium. It remains to prove the ``only if'' part.

    First, consider $i\in N_{uni}\setminus N_{una}$. Then, $\cont_i^{\mathcal{\group}}(\bm{\widetilde\sprg}_{-i})=\sprg\subset\cont_i^{\mathcal{\group}}(\bm{\sprg}_{-i})$ and $\stpbf_i^{\mathcal{\group}}(\bm{\widetilde\sprg}_{-i})=\stpbf_i^{\mathcal{\group}}(\bm{\sprg}_{-i})=\emptyset$. Hence, $U_i(\underline{p}(p;\sprg),\overline{p}(p;\sprg);p)\leq V_i^{\mathcal{\group}}(p;\bm{\widetilde\sprg}_{-i})\leq V_i^{\mathcal{\group}}(p;\bm{\sprg}_{-i})$ for any $p$. Because $\bm{\sprg}$ is an equilibrium, $U_i(\underline{p}(p;\sprg),\overline{p}(p;\sprg);p)=V_i^{\mathcal{\group}}(p;\bm{\sprg}_{-i})$, so $U_i(\underline{p}(p;\sprg),\overline{p}(p;\sprg);p)=V_i^{\mathcal{\group}}(p;\bm{\widetilde\sprg}_{-i})=V_i^{IN}(p;\sprg)$ for any $p$.

    Second, for $i\in N_{una}\setminus N_{uni}$, $\cont_i^{\mathcal{\group}}(\bm{\widetilde\sprg}_{-i})=\cont_i^{\mathcal{\group}}(\bm{\sprg}_{-i})=(0,1)$ and $\stpbf_i^{\mathcal{\group}}(\bm{\widetilde\sprg}_{-i})=\sprg\supset\stpbf_i^{\mathcal{\group}}(\bm{\sprg}_{-i})$. Hence, $U_i(\underline{p}(p;\sprg),\overline{p}(p;\sprg);p)\leq V_i^{\mathcal{\group}}(p;\bm{\widetilde\sprg}_{-i})\leq V_i^{\mathcal{\group}}(p;\bm{\sprg}_{-i})=U_i(\underline{p}(p;\sprg),\overline{p}(p;\sprg);p)$. Furthermore, $V_i^{\mathcal{\group}}(p;\bm{\widetilde\sprg}_{-i})=V_i^{OUT}(p;\sprg)$ for any $p$.

    Finally, for $i\in N_{uni}\cap N_{una}$, $\cont_i^{\mathcal{\group}}=(0,1)$ and $\stpbf_i^{\mathcal{\group}}=[0,1]$, 
    so $\sprg$ must be player $i$'s individual optimum. Hence, she has no incentive to deviate from $\bm{\widetilde\sprg}$, and obviously $U_i(\underline{p}(p;\sprg),\overline{p}(p;\sprg);p)=V_i^{IN}(p;\sprg)=V_i^{OUT}(p;\sprg)$. 
\end{proof}

\begin{lemma}
\label{lattice-general}
    If $\sprg$ and $\sprg'$ are equilibria, $\sprg\cup\sprg'$ is also an equilibrium.
\end{lemma}
\begin{proof}
    By \Cref{reduction-uni-una}, we only need consider $N_{uni}$ and $N_{una}$: Players in $N_{uni}$ can deviate by stopping earlier, while those in $N_{una}$ can deviate by stopping later. Since it is impossible to have many players in both sets and the case with $N_{uni}=N_{una}=\{i\}$ is trivial, the problem reduces to unilateral stopping among $N_{uni}$ or unanimous stopping among $N_{una}$.
    
\paragraph{Unilateral Stopping} First, consider the simple case: 
$\sprg=(\underline{p},\overline{p})$ and $\sprg'=(\underline{p}',\overline{p}{}')$. Let $\sprg'':=\sprg\cup\sprg'$. Recall that equilibrium condition (ii) in \Cref{concave} for $p\in[0,1]\setminus\sprg''$ is always satisfied under unilateral stopping given the definition of $V_i^{IN}$. It thus suffices to check condition (i) for $p\in\sprg''$.

    The case where either $\sprg''=(\underline{p},\overline{p})$ or $\sprg''=(\underline{p}',\overline{p}{}')$ is trivial. When $(\underline{p},\overline{p})\cap(\underline{p}',\overline{p}{}')=\emptyset$, condition (i) holds true since for $p\in \sprg=(\underline{p},\overline{p})$,
    \[
    V_i^{IN}\left(p;\sprg''\right)=V_i^{IN}\left(p;\sprg\right)=U_i\big(\underline{p}(p;\sprg),\overline{p}(p;\sprg);p\big)=U_i\big(\underline{p}(p;\sprg''),\overline{p}(p;\sprg'');p\big).
    \]
    Similarly for $p\in\sprg'=(\underline{p}',\overline{p}{}')$.
    
    It remains to consider the case where $(\underline{p},\overline{p})$ and $(\underline{p}',\overline{p}{}')$ are intersecting but not nested. Without loss assume $\underline{p}\leq\underline{p}'\leq\overline{p}\leq\overline{p}{}'$, so $\sprg''=(\underline{p},\overline{p}{}')$. Since $\sprg=(\underline{p},\overline{p})$ is an equilibrium, by \Cref{concave}, for any $p\in(\underline{p},\overline{p})$,
    \[
    \frac{\overline{p}-p}{\overline{p}-\underline{p}}[u_i(\underline{p})-\phi_i(\underline{p})]+\frac{p-\underline{p}}{\overline{p}-\underline{p}}[u_i(\overline{p})-\phi_i(\overline{p})]=V_i^{IN}\left(p;\sprg\right)\geq u_i(p)-\phi_i(p)
    \]
    Therefore, the line connecting points $(\underline{p},u_i(\underline{p})-\phi_i(\underline{p}))$ and $(\overline{p},u_i(\overline{p})-\phi_i(\overline{p}))$ is above the graph of $u_i-\phi_i$ within $(\underline{p},\overline{p})$.
    
    Similarly, for $\sprg'=(\underline{p}',\overline{p}{}')$, we also have: for any  $p\in(\underline{p}',\overline{p}{}'),$
    \[
    \frac{\overline{p}{}'-p}{\overline{p}{}'-\underline{p}'}[u_i(\underline{p}')-\phi_i(\underline{p}')]+\frac{p-\underline{p}'}{\overline{p}{}'-\underline{p}'}[u_i(\overline{p}{}')-\phi_i(\overline{p}{}')]=V_i^{IN}\left(p;\sprg'\right)\geq u_i(p)-\phi_i(p),
    \]
    i.e., the line connecting points $(\underline{p}',u_i(\underline{p}')-\phi_i(\underline{p}'))$ and $(\overline{p}{}',u_i(\overline{p}{}')-\phi_i(\overline{p}{}'))$ is above the graph of $u_i-\phi_i$ within $(\underline{p}',\overline{p}{}')$.
    
    By the above two inequalities and some algebra, we can derive
    \[
    \frac{\overline{p}{}'-p}{\overline{p}{}'-\underline{p}}[u_i(\underline{p})-\phi_i(\underline{p})]+\frac{p-\underline{p}}{\overline{p}{}'-\underline{p}}[u_i(\overline{p}{}')-\phi_i(\overline{p}{}')]\geq u_i(p)-\phi_i(p),\quad\forall p\in(\underline{p},\overline{p}{}').
    \]
    \[
    \Rightarrow\frac{\overline{p}{}'-p}{\overline{p}{}'-\underline{p}}[u_i(\underline{p})-\phi_i(\underline{p})]+\frac{p-\underline{p}}{\overline{p}{}'-\underline{p}}[u_i(\overline{p}{}')-\phi_i(\overline{p}{}')]=V_i^{IN}\left(p;(\underline{p},\overline{p}{}')\right),\forall p\in(\underline{p},\overline{p}{}').
    \]
    As a result, $\sprg''=(\underline{p},\overline{p}{}')$ is also an equilibrium.

    In general, $\sprg$ and $\sprg'$ 
    are open sets, which are unions of countably many (disjoint) open intervals. 
    The previous argument shows that the union of any two interval equilibria is also an equilibrium. By induction, the union of countably many interval equilibria is also an equilibrium.

\paragraph{Unanimous Stopping.}
It suffices to prove a stronger result: 
\begin{lemma}
\label{unanimous-larger}
    For $\sprg\subset\sprg'$, if $U_i(\underline{p}(p;\sprg),\overline{p}(p;\sprg);p)=V_i^{OUT}(p;\sprg)$ for all $p\in[0,1]$, then $U_i(\underline{p}(p;\sprg'),\overline{p}(p;\sprg');p)=V_i^{OUT}(p;\sprg')$ for all $p\in[0,1]$.
\end{lemma}
\begin{proof}
By definition of $V_i^{OUT}$, for any $\sprg\subset\sprg'$,
\[
V_i^{OUT}\left(p;\sprg'\right)\leq V_i^{OUT}\left(p;\sprg\right),\quad\forall p\in[0,1].
\]
For $p\in[0,1]\setminus\sprg'$, $u_i(p)-\phi_i(p)\leq V_i^{OUT}(p;\sprg')\leq V_i^{OUT}(p;\sprg_i^*)=u_i(p)-\phi_i(p)$, where the first inequality is by the definition of $V_i^{OUT}$ and the equality is by assumption and $p\in[0,1]\setminus\sprg$. Therefore, $V_i^{OUT}(p;\sprg')=u_i(p)-\phi_i(p)$.

    Then for $p\in\sprg'$, let $\underline{p}:=\underline{p}(p;\sprg')$ and $\overline{p}:=\overline{p}(p;\sprg')$; note that $\underline{p},\overline{p}\in[0,1]\setminus\sprg'\subset[0,1]\setminus\sprg$. Since $V_i^{OUT}(\cdot;\sprg)$ is concave and $V_i^{OUT}(\tilde{p};\sprg)=u_i(\tilde{p})-\phi_i(\tilde{p})$ at $\tilde{p}=\underline{p},\overline{p}$ by assumption, it holds that for any $\tilde{p}\in[0,1]\setminus\sprg'$,
    \[
    \frac{\overline{p}-\tilde{p}}{\overline{p}-\underline{p}}[u_i(\underline{p})-\phi_i(\underline{p})]+\frac{\tilde{p}-\underline{p}}{\overline{p}-\underline{p}}[u_i(\overline{p})-\phi_i(\overline{p})]\geq V_i^{OUT}(\tilde{p};\sprg)\geq V_i^{OUT}(\tilde{p};\sprg').
    \]
    As a result, by the definition of $V_i^{OUT}$, it must be that within $(\underline{p},\overline{p})$, $V_i^{OUT}(\cdot;\sprg')$ coincides with the line connecting $(\underline{p},u_i(\underline{p})-\phi_i(\underline{p}))$ and $(\overline{p},u_i(\overline{p})-\phi_i(\overline{p}))$, therefore $V_i^{OUT}(p;\sprg')=\frac{\overline{p}-p}{\overline{p}-\underline{p}}[u_i(\underline{p})-\phi_i(\underline{p})]+\frac{p-\underline{p}}{\overline{p}-\underline{p}}[u_i(\overline{p})-\phi_i(\overline{p})]$.

    In sum, $U_i(\underline{p}(p;\sprg'),\overline{p}(p;\sprg');p)=V_i^{OUT}(p;\sprg')$ for all $p\in[0,1]$.
\end{proof}
As an implication of \Cref{unanimous-larger}, $\sprg\cup\sprg'\supset\sprg$ is also an equilibrium.
\end{proof}
\begin{proof}[Proof of \Cref{pareto,pareto-general}]
    Again given \Cref{reduction-uni-una}, the case with $N_{uni}=N_{una}=\{i\}$ for $i\in N$ is trivial, and the problem reduces to either unilateral or unanimous stopping.

    \paragraph{Unilateral Stopping} Fix a Pareto weight $\bm{\lambda}\in\mathbb{R}^N_+\setminus\{\bm{0}\}$. Construct two virtual players with 
    the same preference $W(\cdot;\bm{\lambda})=\sum_{i\in N}\lambda_i(u_i-c_i)$. Both $\sprg$ and the $\bm{\lambda}$-efficient sampling region $\sprg_{\bm{\lambda}}^*$ are equilibria under unilateral stopping between these two virtual players. Therefore, by \Cref{lattice-general}, $\sprg\cup \sprg_{\bm{\lambda}}^*$ is also an equilibrium, implying that the virtual player with preference $W(\cdot;\bm{\lambda})$ prefers $\sprg\cup\sprg_{\bm{\lambda}}^*$ to $\sprg_{\bm{\lambda}}^*$. Because $\sprg_{\bm{\lambda}}^*$ is uniquely $\bm{\lambda}$-efficient, it must be $\sprg\cup \sprg_{\bm{\lambda}}^*=\sprg_{\bm{\lambda}}^*$ and thus $\sprg\subset\sprg_{\bm{\lambda}}^*$.
    

    \paragraph{Unanimous Stopping.} Suppose $\sprg\subset \sprg_{\bm{\lambda}}^*$, then $V_{i}^{OUT}(\cdot;\sprg)\geq V_{i}^{OUT}(\cdot;\sprg_{\bm{\lambda}}^*)$. Since $\sprg$ is an equilibrium, $\forall p\in[0,1]$,
    \[
    \sum_{i\in N}\lambda_iU_i\big(\underline{p}(p;\sprg),\overline{p}(p;\sprg);p\big)=\sum_{i\in N}\lambda_iV_{i}^{OUT}(p;\sprg)\geq\sum_{i\in N}\lambda_iV_{i}^{OUT}(p;\sprg_{\bm{\lambda}}^*)=\hat{W}(p;\bm{\lambda}).
    \]
    Hence, $\sprg$ is also $\bm{\lambda}$-efficient. By uniqueness, it must be $\sprg=\sprg_{\bm{\lambda}}^*$.
\end{proof}
\begin{proof}[Proof of \Cref{control-sharing}]
    Suppose that $\sprg\supsetneq\sprg_i^*\cup\sprg_{-i}^*$. By \Cref{unanimous-larger}, $\sprg_i^*\cup\sprg_{-i}^*\supset\sprg_i^*$ is an equilibrium. A contradiction to $\sprg$ being minimal.
\end{proof}
\begin{proof}[Proof of \Cref{misalign}] 
    Let $T^b$ denote the set of equilibria under misalignment $b$. 
    Consider any $\sprg\in T^b$. Fix a prior $p$ with $\underline{p}:=\underline{p}(p;\sprg)$ and $\overline{p}:=\overline{p}(p;\sprg)$. Let $F_{\underline{p},\overline{p}}$ denote the corresponding binary policy. Since $c_1(\cdot)=c_2(\cdot)$, let $\phi(\cdot):=\phi_i(\cdot)$. By \Cref{MPE}, $\int_0^1[u_i(\tilde{p};b)-\phi(\tilde{p})]\mathrm{d}F_{\underline{p},\overline{p}}(\tilde{p})\geq\int_0^1[u_i(\tilde{p};b)-\phi(\tilde{p})]\mathrm{d}F(\tilde{p})$ for any $F$ that is an MPC (resp., MPS with support in $[0,1]\setminus \sprg$) of $F_{\underline{p},\overline{p}}$. Since $u_1=f+bg$ and $u_2=f-bg$, for any $F$ that is an MPC (resp., MPS with support in $[0,1]\setminus \sprg$) of $F_{\underline{p},\overline{p}}$ and $b'\in[0,b)$,
    {\small
    \[
    \int_0^1\left[f(\tilde{p})-\phi(\tilde{p})\right]\mathrm{d}(F_{\underline{p},\overline{p}}-F)(\tilde{p})\geq b\left|\int_0^1g(\tilde{p})\mathrm{d}(F_{\underline{p},\overline{p}}-F)(\tilde{p})\right|\geq b'\left|\int_0^1g(\tilde{p})\mathrm{d}(F_{\underline{p},\overline{p}}-F)(\tilde{p})\right|.
    \]}
    \noindent Hence, $\int_0^1[u_i(\tilde{p};b')-\phi(\tilde{p})]\mathrm{d}F_{\underline{p},\overline{p}}(\tilde{p})\geq\int_0^1[u_i(\tilde{p};b')-\phi(\tilde{p})]\mathrm{d}F(\tilde{p})$ for any MPC (resp., MPS with support in $[0,1]\setminus \sprg$) $F$ of $F_{\underline{p},\overline{p}}$. This is true for any belief $p$. As a result, $\sprg\in T^{b'}$. Hence, $T^{b}\subset T^{b'}$ for any $b>b'\geq0$.
\end{proof}

\begin{proof}[Proof of \Cref{misalign-max,misalign-general}] 
    Follows from \Cref{lattice,misalign,reduction-uni-una}.
\end{proof}

\begin{proof}[Proof of \Cref{cs-general}]
    Note that with $\mathcal{\group}\subset\mathcal{\group}'$, $N_{uni}(\mathcal{\group})\subset N_{uni}(\mathcal{\group}')$ and $N_{una}(\mathcal{\group})\supset N_{una}(\mathcal{\group}')$. The result is then due to the following lemma.
\end{proof}
\begin{lemma}
    If $N_{uni}(\mathcal{\group})\subset N_{uni}(\mathcal{\group}')$ and $N_{una}(\mathcal{\group})\supset N_{una}(\mathcal{\group}')$, then the maximum equilibrium under $\mathcal{\group}'$ cannot be strictly larger than the maximum equilibrium under $\mathcal{\group}$. The same applies to minimal equilibria.
\end{lemma}
\begin{proof}
    \emph{For maximum equilibria.} Suppose not, then $\sprg\subsetneq\sprg'$ where $\sprg$ and $\sprg'$ are maximum equilibria under $\mathcal{\group}$ and $\mathcal{\group}'$, respectively. I want to show $\sprg'$ is also an equilibrium under $\mathcal{\group}$.  It suffices to consider players in $N_{uni}(\mathcal{\group})$ and $N_{una}(\mathcal{\group})$. On the one hand, since $\sprg'$ is an equilibrium under $\mathcal{\group}'$, players in $N_{uni}(\mathcal{\group})\subset N_{uni}(\mathcal{\group}')$ and $N_{una}(\mathcal{\group}')$ have no incentive to deviate. On the other hand, given that $\sprg'\supset\sprg$ and $\sprg$ is an equilibrium under $\mathcal{\group}$, players in $N_{una}(\mathcal{\group})$ have no incentive to prolong learning beyond $\sprg$, thus by \Cref{unanimous-larger} they also have no incentive to deviate from $\sprg'$. As a result, $\sprg'$ is also an equilibrium under $\mathcal{\group}$, contradicting to $\sprg$ being the maximum equilibrium.

    \paragraph{For minimal equilibria} The proof is similar and omitted for brevity.
\end{proof}

\newpage
\clearpage
\pagenumbering{arabic}
\renewcommand*{\thepage}{B-\arabic{page}}
\section{Supplementary Appendix 
}
\label{supp-app}
\Cref{counterexample-MPC} provides a counterexample for \Cref{MPC} with non-Markov strategies; \Cref{sect-app-examples} provides additional examples for equilibrium analysis under unanimous stopping; and \Cref{proofs-app} proves the results in the committee search application. \Cref{Poisson} adapts the baseline model to conclusive Poisson learning. 
Finally, \Cref{GP12} discusses how the competition in persuasion application connects to \citeauthor{gul2012war}'s (\citeyear{gul2012war}) model of war of information.

\subsection{An Example of Infeasible MPC with Non-markov Strategies}
\label{counterexample-MPC}
The following example shows not every MPC is feasible under unilateral stopping when players use non-Markov strategies.
\begin{example}
Suppose $p_0=1/2$ and consider $\rho=\inf\{t\geq\tau_0: p_t\in(1/2,1)^c\}$ where $\tau_0=\inf\{t\geq0:p_t\in(0,2/3)^c\}$. Then the induced distribution of $p_\rho$ is $\mu_{\rho}=\frac14\delta_{0}+\frac12\delta_{1/2}+\frac14\delta_{1}$ where $\delta_p$ is the Dirac measure at $p$. Notice that $\mu=\frac12\delta_{1/4}+\frac12\delta_{3/4}$ is an MPC of $\mu_\rho$. However, if the sample path is such that after passing $2/3$, $p_t$ hits $1/2$ before hitting $3/4$, the probability of which conditional on $p_s=2/3$ is $1/3$, the belief process must be trapped at $1/2$ as required by $\rho$, instead of arriving at $3/4$. Hence, over these sample paths, any stopping time yielding $\mu$ must be greater than $\rho$, making it infeasible to implement $\mu$ under unilateral stopping. 
\end{example}

\subsection{Additional Examples for Unanimous Stopping}
\label{sect-app-examples}
In this section, I present two additional examples under unanimous stopping where non-interval equilibria exist and some of them are not comparable to the Pareto efficient sampling regions.

The first example uses the same payoff structure $u_1-\phi_1$ I used for player 1 when illustrating \Cref{concave}. One can find it in \Cref{fig-unanimous}. Consider unanimous stopping between two players with the same preference $u_1-\phi_1$. As demonstrated by \Cref{fig-twointerval}, $\sprg=(b'',\hat{b})\cup(\hat{b},B'')$ can be an equilibrium which consists of two intervals: players only sample within $(b'',\hat{b})$ and $(\hat{b},B'')$; they stop at $\hat{b}$ or whenever the belief is below $b''$ or above $B''$. Importantly, this equilibrium $\sprg$ is not comparable to the Pareto efficient sampling region $\sprg^*=(b',B')$ (recall that $(b',B')$ is player 1's favorite) since $\hat{b}\in\sprg^*$ but $\hat{b}\notin\sprg$. Notice that $\sprg$ in this example is not even minimal, but still not comparable to the Pareto optimum. For the semi-lattice structure, given that both $\sprg$ and $\sprg^*$ are equilibria, one can easily verify that $\sprg\cup\sprg^*=(b'',B'')$ is also an equilibrium but $\sprg\cap\sprg^*=(b',\hat{b})\cup(\hat{b},B')$ is not. 

\begin{figure}[!t]
    \centering
    \includegraphics[scale=0.35]{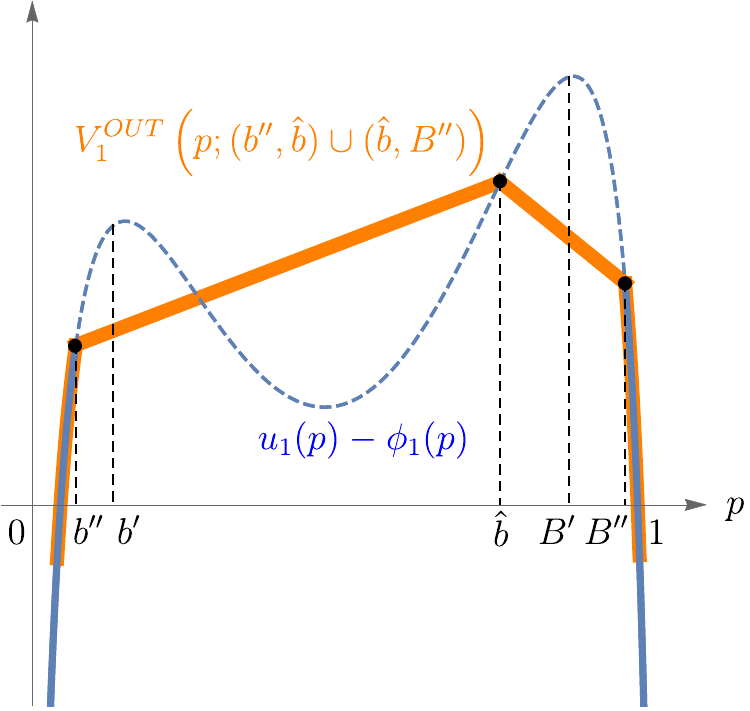}
    \caption{An equilibrium with $\sprg=(b'',\hat{b})\cup(\hat{b},B'')$.}
    \label{fig-twointerval}
\end{figure}

The second example presents an opposite situation where the Pareto efficient sampling region consists of two intervals but there is an interval equilibrium that is not comparable to the optimum. Two players have the same preference given by $u-\phi$ in \Cref{fig-unnested}. On the one hand, \Cref{fig-unnested-eqm} shows an interval equilibrium with $\sprg=(\underline{p},\overline{p})$ and it is minimal; on the other hand, \Cref{fig-unnested-opt} shows the Pareto optimum $\sprg^*=(\underline{p}_1^*,\overline{p}_1^*)\cup(\underline{p}_2^*,\overline{p}_2^*)$ derived from concavification. Notice that $\sprg$ is neither nesting $\sprg^*$ nor nested by $\sprg^*$. Since $\sprg^*$ is also an equilibrium in this case, $\sprg\cup\sprg^*=(\underline{p}_1^*,\overline{p}_2^*)$ is an equilibrium, but again $\sprg\cap\sprg^*=(\underline{p},\overline{p}_1^*)\cup(\underline{p}_2^*,\overline{p})$ is not.

\begin{figure}[t]
	\centering
	\begin{subcaptionblock}{0.48\textwidth}
		\centering
		\includegraphics[scale=0.35]{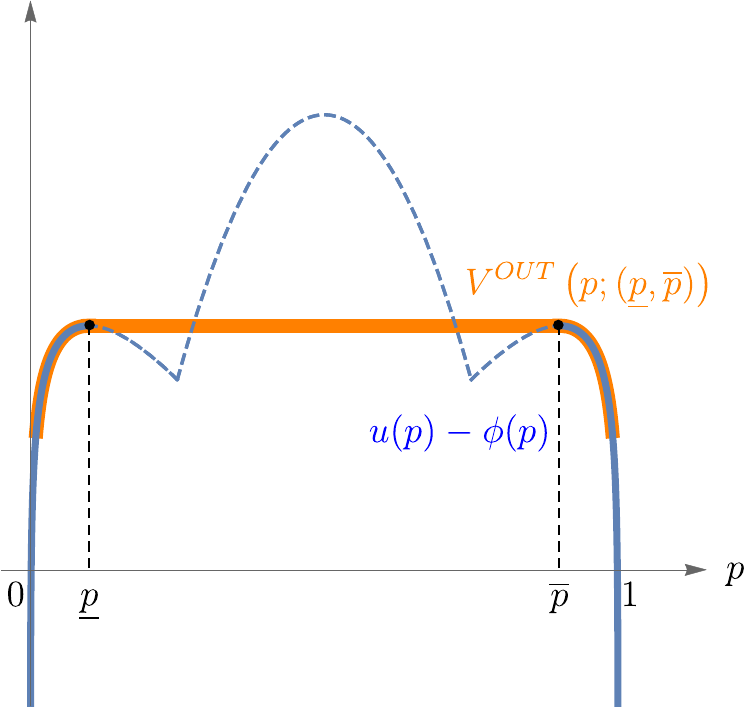}
            \caption{An equilibrium with $\sprg=(\underline{p},\overline{p})$.}
        \label{fig-unnested-eqm}
	\end{subcaptionblock}
        \begin{subcaptionblock}{0.48\textwidth}
		\centering
		\includegraphics[scale=0.35]{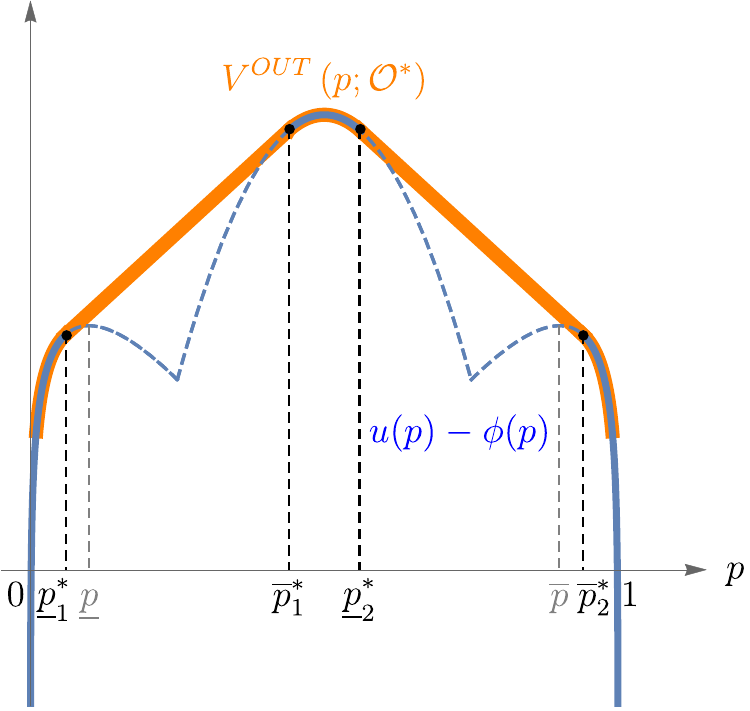}
            \caption{The Pareto optimum $\sprg^*=(\underline{p}_1^*,\overline{p}_1^*)\cup(\underline{p}_2^*,\overline{p}_2^*)$.}
        \label{fig-unnested-opt}
	\end{subcaptionblock}
 	\caption{Another Example for Unanimous Stopping.}
    \label{fig-unnested}
\end{figure}

\subsection{Proofs for the Committee Search Application}
\label{proofs-app}

\begin{proof}[Proof of \Cref{prop:strong}]
    \emph{``strong equilibria''$\Rightarrow$ ``fixed points''.} First, suppose that $\sprg$ is an equilibrium that contains an interval $(\underline{p},\overline{p})$ that is either below or above $\piv$. However, since players' indirect payoffs $u_i-\phi_i$ are all piecewise strictly concave given the cutoff $\piv$, all players can be better off by collectively deviating to $\sprg\setminus(\underline{p},\overline{p})$ for all priors, and strictly better for $p\in(\underline{p},\overline{p})$. This thus implies that $\sprg$ cannot be a strong equilibrium. Hence, all non-interval sampling regions or intervals that do not cross $\piv$ cannot be strong equilibria.
    
    Next, I want to show that interval equilibria without the fixed-point property cannot be strong. It is useful to introduce the following lemma about $U_i$:
        \begin{lemma}
    \label{lm:BR-extension}
        For any $\underline{p}\leq\piv\leq\overline{p}$ and $\underline{p}{}'\leq\piv\leq\overline{p}{}'$, if 
        \begin{itemize}
            \item[(i)] $U_i(\underline{p}{},\overline{p};\piv)\geq U_i(\underline{p}{},\overline{p}{}'';\piv)$ for any $\overline{p}{}''\in[\min\{\overline{p},\overline{p}{}'\},\max\{\overline{p},\overline{p}{}'\}]$, and
            \item[(ii)] $U_i(\underline{p}{},\overline{p};\piv)\geq U_i(\underline{p}{}'',\overline{p}{};\piv)$ for any $\underline{p}{}''\in[\min\{\underline{p},\underline{p}{}'\},\max\{\underline{p},\underline{p}{}'\}]$,
        \end{itemize}
        then $U_i(\underline{p}{},\overline{p};p)\geq U_i(\underline{p}{}',\overline{p}{}';p)$ for all $p\in[\min\{\underline{p},\underline{p}{}'\},\max\{\overline{p},\overline{p}{}'\}]$.
    \end{lemma}
    Intuitively, if player $i$ prefers $(\underline{p},\overline{p})$ to both $(\underline{p}{}',\overline{p})$ and $(\underline{p},\overline{p}{}')$, then she must also prefer $(\underline{p},\overline{p})$ to both $(\underline{p}{}',\overline{p}{}')$. This lemma thus suggests that it is without loss to only consider one-sided best responses (with $\piv$ as the benchmark).
    
    Suppose that $(\underline{p},\overline{p})$ is not a fixed point in $\underline{B}{}_{\lbp(\mathcal{\group})}$ and $\overline{B}_{\ubp(\mathcal{\group})}$. Without loss say $\underline{p}\ne\underline{B}{}_{\lbp(\mathcal{\group})}(\overline{p})=:\underline{p}'$. The case in which $\overline{p}\ne\overline{B}_{\ubp(\mathcal{\group})}(\underline{p})$ is symmetric.
    \begin{itemize}
        \item When $\underline{p}<\underline{p}'$:  Consider some $\group'\in\arg\min_{\group\in\mathcal{\group}}\max_{i\in\group}i$. By definition of $\lbp(\mathcal{\group})$, $\lbp(\mathcal{\group})\in\group'\subset\{1,\dots,\lbp(\mathcal{\group})\}$. Think about their deviations to $\underline{p}{}_j'=\underline{p}'$ for $j\in\group'$ while the upper bounds and all other players' strategies remain the same. Since $\underline{p}=\max_{\group\in\mathcal{\group}}\min_{i\in\group}\underline{p}{}_i<\underline{p}'$, after this deviation, we should have $\max_{\group\in\mathcal{\group}}\min_{i\in\group}\underline{p}{}_i'=\underline{p}'$ and still $\min_{\group\in\mathcal{\group}}\max_{i\in\group}\overline{p}{}_i'=\overline{p}$, thus the induced collective sampling region is $(\underline{p}',\overline{p})\subset(\underline{p},\overline{p})$. It suffices to check the payoffs for players within $\group'$ for $p\in(\underline{p},\overline{p})$. By definition of $\underline{p}'=\underline{B}{}_{\lbp(\mathcal{\group})}(\overline{p})$, $U_{\lbp(\mathcal{\group})}(\underline{p}',\overline{p};p)> U_{\lbp(\mathcal{\group})}(\underline{p},\overline{p};p)$ holds for $p=\piv$ and extends to all $p\in(\underline{p},\overline{p})$ because the left-hand side is concave in $p$ while the right-hand side is linear in $p$ and they coincide at $p=\underline{p}$ and at $p=\overline{p}$ (or simply by \Cref{lm:BR-extension}). As for other players $j$ in $\group'$, we have $j\leq\lbp(\mathcal{\group})$. Notice that $U_j(\underline{p}'',\overline{p};p)$ is single-peaked in $\underline{p}''$ and $\underline{p}<\underline{p}'=\underline{B}{}_{\lbp(\mathcal{\group})}(\overline{p})\leq\underline{B}{}_{j}(\overline{p})$, therefore $U_j(\underline{p}',\overline{p};p)> U_j(\underline{p},\overline{p};p)$ at $p=\piv$, which extends to all $p\in(\underline{p},\overline{p})$ for the same reason as before. In sum, all players in $\group'$ benefit from the above deviation regardless of the belief, and strictly so when $p\in(\underline{p},\overline{p})$. As a result, $(\underline{p},\overline{p})$ is not strong.

    \item When $\underline{p}>\underline{p}{}'$: Consider players in $\group'=\{j:j=\max_{i\in\group}i,\group\in\mathcal{\group}\}$. By definition, $\group'\subset\{\lbp(\group),\dots, n\}$ for any $\group\in\mathcal{\group}$. Think about their deviations to $\underline{p}{}_j'=\underline{p}{}'$ for all $j\in\group'$ while the upper bounds and all other players' strategies remain the same. This deviation will lead to $(\underline{p}{}',\overline{p})$ as the collective sampling region and $(\underline{p}{}',\overline{p})\supset(\underline{p},\overline{p})$. It suffices to check players' incentives for $p\in(\underline{p}{}',\overline{p})$. By definition and by the same argument as before, $U_{\lbp(\mathcal{\group})}(\underline{p}',\overline{p};p)> U_{\lbp(\mathcal{\group})}(\underline{p},\overline{p};p)$ holds for all $p\in(\underline{p}{}',\overline{p})$. Now because $j\geq\lbp(\mathcal{\group})$ for all $j\in\group'$ and thus $\underline{B}{}_{j}(\overline{p})\leq\underline{B}{}_{\lbp(\mathcal{\group})}(\overline{p})=\underline{p}{}'<\underline{p}$, together with the fact that $U_j$ is single-peaked in $\underline{p}''$, we should also have $U_{j}(\underline{p}',\overline{p};p)> U_{j}(\underline{p},\overline{p};p)$ holds for all $p\in(\underline{p}{}',\overline{p})$. Hence, $(\underline{p},\overline{p})$ is not strong.

    \end{itemize}
    \noindent\emph{``fixed points''$\Rightarrow$``strong equilibria''.} According to the proof of \Cref{cs-committee}, any fixed point in $\underline{B}{}_{\lbp(\mathcal{\group})}$ and $\overline{B}_{\ubp(\mathcal{\group})}$ is minimal when $\ubp(\mathcal{\group})<\lbp(\mathcal{\group})$ and there is a maximum one when $\ubp(\mathcal{\group})\geq\lbp(\mathcal{\group})$. Suppose that $(\underline{p},\overline{p})$ is a fixed point in $\underline{B}{}_{\lbp(\mathcal{\group})}$ and $\overline{B}_{\ubp(\mathcal{\group})}$, and the maximum one when $\lbp(\mathcal{\group})\leq\ubp(\mathcal{\group})$. 
    
    First, $(\underline{p},\overline{p})$ is indeed an equilibrium. Suppose that players all use the same strategy $\sprg_i=(\underline{p},\overline{p})$, which is trivially an equilibrium when $N_{uni}(\mathcal{\group})=N_{una}(\mathcal{\group})=\emptyset$. According to \Cref{existence-OSBR} established later, when $N_{uni}(\mathcal{\group})\ne\emptyset$ and thus $\lbp\leq\ubp$, $(\underline{p},\overline{p})$ is the maximum equilibrium under unilateral stopping between $\lbp$ and $\ubp$. Since $\lbp$ and $\ubp$ have no incentive to deviate to $(\underline{p}{}',\overline{p}{}')\subset(\underline{p},\overline{p})$ and by definition, $\lbp(\mathcal{\group})\leq i\leq\ubp(\mathcal{\group})$ for any $i\in N_{uni}(\mathcal{\group})$, any such player $i$ also has no incentive to deviate to $(\underline{p}{}',\overline{p}{}')$ given the ordered preferences.\footnote{In detail, for any $\underline{p}{}''\in[\underline{p},\underline{p}{}']$ and any $\overline{p}{}''\in[\overline{p}{}',\overline{p}]$, $U_\lbp(\underline{p},\overline{p};\piv)\geq U_\lbp(\underline{p}{}',\overline{p};\piv)$ implies $U_i(\underline{p},\overline{p};\piv)\geq U_i(\underline{p}{}',\overline{p};\piv)$, and $U_\ubp(\underline{p},\overline{p};\piv)\geq U_\ubp(\underline{p}{},\overline{p}{}';\piv)$ implies $U_i(\underline{p},\overline{p};\piv)\geq U_i(\underline{p}{},\overline{p}{}';\piv)$. So by \Cref{lm:BR-extension}, $U_i(\underline{p},\overline{p};p)\geq U_i(\underline{p}{}',\overline{p}{}';p)$ for all $p$.} The case with $N_{una}(\mathcal{\group})\ne\emptyset$ is symmetric, where instead $\lbp(\mathcal{\group})\geq\ubp(\mathcal{\group})$ and $\lbp(\mathcal{\group})\geq j\geq\ubp(\mathcal{\group})$ for all $j\in N_{una}(\mathcal{\group})$. In conclusion, $(\underline{p},\overline{p})$ is indeed an equilibrium under $\mathcal{\group}$-collective stopping.

    It remains to check that $(\underline{p},\overline{p})$ is strong. Consider an arbitrary sampling interval $(\underline{p}{}',\overline{p}{}')$ as the target of a group deviation.

When $\underline{p}{}'>\underline{p}$: To implement such a deviation, the group must include some player $j\geq\lbp$ (since $\max_{i\in\group}\group\geq\lbp$ for any $\group\in\mathcal{\group}$). Suppose that the deviation to $(\underline{p}{}',\overline{p}{}')$ is beneficial for player $j$, that is, $U_j(\underline{p}{}',\overline{p}{}';p)\geq U_j(\underline{p},\overline{p};p)$ for all $p$. Since both sides are linear in $p$ for $p\in(\underline{p}{}',\overline{p}{}')$, it implies that the line connecting $(\underline{p}{}',(u_j-\phi)(\underline{p}{}'))$ and $(\overline{p}{}',(u_j-\phi)(\overline{p}{}'))$ is above the one connecting $(\underline{p}{},(u_j-\phi)(\underline{p}{}))$ and $(\overline{p}{},(u_j-\phi)(\overline{p}{}))$. Hence, the line connecting $(\underline{p}{}',(u_j-\phi)(\underline{p}{}'))$ and $(\overline{p}{},(u_j-\phi)(\overline{p}{}))$ is above the one connecting $(\underline{p}{},(u_j-\phi)(\underline{p}{}))$ and $(\overline{p}{},(u_j-\phi)(\overline{p}{}))$; in other words, $U_j(\underline{p}{}',\overline{p}{};p)\geq U_j(\underline{p},\overline{p};p)$ for all $p\in[\underline{p}{}',\overline{p}{}]$. Since $U_j(\cdot,\overline{p};\piv)$ is single-peaked, it must be $\underline{B}{}_j(\overline{p})>\underline{p}$. But $j\geq\lbp$, so $\underline{B}{}_\lbp(\overline{p})\geq\underline{B}{}_j(\overline{p})>\underline{p}$, leading to a contradiction to $\underline{p}=\underline{B}{}_\lbp(\overline{p})$.
        
When $\overline{p}{}'<\overline{p}$: The deviating group must then include some player $k\leq\ubp$. If player $k$ benefits from such a deviation, it will imply $\overline{B}{}_\ubp(\underline{p})\leq\overline{B}{}_k(\underline{p})<\overline{p}$ and lead to a contradiction to $\overline{p}=\underline{B}{}_\ubp(\overline{p})$. The analysis is analogous to the previous one and omitted for brevity.
 
 When $(\underline{p}{}',\overline{p}{}')\supset(\underline{p},\overline{p})$: Without loss $\underline{p}{}'<\underline{p}$. To implement such a deviation, the deviating group must include some $j\leq\lbp$, since it has to gather at least one player from the group $\group'$ such that $\lbp=\max_{i\in\group'}i$. If $\overline{p}{}'=\overline{p}$, then we must have $\underline{B}{}_j(\overline{p})\geq\underline{B}{}_\lbp(\overline{p})=\underline{p}$, which implies that player $j$ cannot benefit from a deviation to $(\underline{p}{}',\overline{p})$ with $\underline{p}{}'<\underline{p}$. 
        
        If $\overline{p}{}'>\overline{p}$, the deviation must also involve some player $k\geq\ubp$. First, consider the case when $\lbp>\ubp$, $(\underline{p},\overline{p})$ is a minimal equilibrium under unanimous stopping between players $\lbp$ and $\ubp$. If $\lbp\geq j\geq\ubp$, since both $\lbp$ and $\ubp$ have no incentive to deviate to $(\underline{p}{}',\overline{p}{}')\supset(\underline{p},\overline{p})$, neither does player $j$ due to the ordered preferences. The same holds for player $k$ if $\lbp\geq k\geq\ubp$. 
        
        As a result, the only remaining possibilities are $j<\ubp<\lbp<k$ when $\lbp>\ubp$ and $j\leq\lbp\leq\ubp\leq k$ when $\lbp\leq\ubp$. Suppose that both players $j$ and $k$ can benefit from the deviation to $(\underline{p}{}',\overline{p}{}')\supset(\underline{p},\overline{p})$ at all beliefs. Then $(\underline{p}{}',\overline{p}{}')$ must be an equilibrium under unilateral stopping between players $j$ and $k$. Let $(\underline{p}{}_{\max},\overline{p}{}_{\max})$ denote the maximum equilibrium under unilateral stopping between them, hence $(\underline{p}{}_{\max},\overline{p}{}_{\max})\supset(\underline{p}{}',\overline{p}{}')\supsetneq(\underline{p},\overline{p})$. Since $j\leq k$, by \Cref{existence-OSBR}, $(\underline{p}{}_{\max},\overline{p}{}_{\max})$ is the maximum fixed point in $\underline{B}{}_j$ and $\overline{B}{}_k$. Recall that $(\underline{p},\overline{p})$ is a fixed point in $\underline{B}{}_\lbp$ and $\overline{B}{}_\ubp$, and the maximum one when $\lbp\leq\ubp$. However, given that $\underline{B}{}_\lbp\leq\underline{B}{}_j$ and $\overline{B}{}_k\leq\overline{B}{}_\ubp$, the proof of \Cref{cs-committee} shows that $(\underline{p}{}_{\max},\overline{p}{}_{\max})$ cannot be strictly larger than $(\underline{p},\overline{p})$, leading to a contradiction.
    In conclusion, there cannot be any feasible and profitable group deviation from $(\underline{p},\overline{p})$. Therefore, $(\underline{p},\overline{p})$ is strong. It completes the proof.
\end{proof}
\begin{proof}[Proof of \Cref{lm:BR-extension}]
        First, $U_i(\underline{p}{},\overline{p};\piv)\geq U_i(\underline{p}{},\overline{p}{}';\piv)$ implies
    \begin{align*}
        \frac{\overline{p}-\piv}{\overline{p}-\underline{p}{}}&[u_i(\underline{p}{})-\phi_i(\underline{p}{})]+\frac{\piv-\underline{p}{}}{\overline{p}-\underline{p}{}}[u_i(\overline{p})-\phi_i(\overline{p})]
        \\
        &\geq 
        \frac{\overline{p}{}'-\piv}{\overline{p}{}'-\underline{p}{}}[u_i(\underline{p}{})-\phi_i(\underline{p}{})]+\frac{\piv-\underline{p}{}}{\overline{p}{}'-\underline{p}{}}[u_i(\overline{p}{}')-\phi_i(\overline{p}{}')],
    \end{align*}
    \begin{equation}
    \label{ineq-1}
    \begin{aligned}
        \Rightarrow\ &U_i\big(\underline{p},\overline{p};p\big)=\frac{\overline{p}-p}{\overline{p}-\underline{p}{}}[u_i(\underline{p}{})-\phi_i(\underline{p}{})]+\frac{p-\underline{p}{}}{\overline{p}-\underline{p}{}}[u_i(\overline{p})-\phi_i(\overline{p})]\\
        \geq &\frac{\overline{p}{}'-p}{\overline{p}{}'-\underline{p}{}}[u_i(\underline{p}{})-\phi_i(\underline{p}{})]+\frac{p-\underline{p}{}}{\overline{p}{}'-\underline{p}{}}[u_i(\overline{p}{}')-\phi_i(\overline{p}{}')],\text{}\forall p\in[\underline{p}{},\min\{\overline{p},\overline{p}{}'\}].
    \end{aligned}
    \end{equation}
    where the implication from $\piv$ to general $p$ is due to the linearity of both sides in $p$ and the fact that the two sides coincide at $p=\underline{p}{}$. In other words, the line connecting points $(\underline{p}{},(u_i-\phi_i)(\underline{p}{}))$ and $(\overline{p},(u_i-\phi_i)(\overline{p}))$ is above the line connecting $(\underline{p}{},(u_i-\phi_i)(\underline{p}{}))$ and $(\overline{p}{}',(u_i-\phi_i)(\overline{p}{}'))$ over $[\underline{p}{},\min\{\overline{p},\overline{p}{}'\}]$. When $\overline{p}\leq\overline{p}{}'$, this further implies that the line connecting $(\overline{p},(u_i-\phi_i)(\overline{p}))$ and $(\overline{p}{}',(u_i-\phi_i)(\overline{p}{}'))$ should also be above the line connecting $(\underline{p}{},(u_i-\phi_i)(\underline{p}{}))$ and $(\overline{p}{}',(u_i-\phi_i)(\overline{p}{}'))$ over $[\min\{\overline{p},\overline{p}{}'\},\max\{\overline{p},\overline{p}{}'\}]$. 
    Then by the concavity of $u_i-\phi_i$ above $\piv$, for any $ p\in[\overline{p},\overline{p}{}']$,
    \begin{equation}
    \label{ineq-3}
    U_i\big(\underline{p},\overline{p};p\big)=u_i(p)-\phi_i(p)\geq\frac{\overline{p}{}'-p}{\overline{p}{}'-\underline{p}{}}[u_i(\underline{p}{})-\phi_i(\underline{p}{})]+\frac{p-\underline{p}{}}{\overline{p}{}'-\underline{p}{}}[u_i(\overline{p}{}')-\phi_i(\overline{p}{}')].
    \end{equation}
    Instead when $\overline{p}>\overline{p}{}'$, by replacing $\overline{p}{}'$ with $\overline{p}{}''$ in Inequality (\ref{ineq-1}) (since the analysis holds for any $\overline{p}{}''\in[\overline{p}{}',\overline{p}]$) and evaluating it at $\overline{p}''=p$, we can get
    \[
    U_i\big(\underline{p},\overline{p};p\big)\geq u_i(p)-\phi_i(p)=U_i\big(\underline{p}{}',\overline{p}{}';p\big),\text{ }\forall p\in[\overline{p}{}',\overline{p}{}].
    \]
    
    Similarly, $U_i(\underline{p}{},\overline{p};\piv)\geq U_i(\underline{p}{}',\overline{p};\piv)$ gives us
    \begin{align*}
        \frac{\overline{p}-\piv}{\overline{p}-\underline{p}{}}&[u_i(\underline{p}{})-\phi_i(\underline{p}{})]+\frac{\piv-\underline{p}{}}{\overline{p}-\underline{p}{}}[u_i(\overline{p})-\phi_i(\overline{p})]
        \\
        &\geq 
        \frac{\overline{p}-\piv}{\overline{p}-\underline{p}'}[u_i(\underline{p}')-\phi_i(\underline{p}')]+\frac{\piv-\underline{p}'}{\overline{p}-\underline{p}'}[u_i(\overline{p})-\phi_i(\overline{p})].
    \end{align*}
    And by similar arguments,
    \begin{equation}
    \label{ineq-2}
    \begin{aligned}
        &U_i\big(\underline{p},\overline{p};p\big)=\frac{\overline{p}-p}{\overline{p}-\underline{p}{}}[u_i(\underline{p}{})-\phi_i(\underline{p}{})]+\frac{p-\underline{p}{}}{\overline{p}-\underline{p}{}}[u_i(\overline{p})-\phi_i(\overline{p})]\\
        \geq&\frac{\overline{p}-p}{\overline{p}-\underline{p}'}[u_i(\underline{p}')-\phi_i(\underline{p}')]+\frac{p-\underline{p}'}{\overline{p}-\underline{p}'}[u_i(\overline{p})-\phi_i(\overline{p})],\text{}\forall p\in[\max\{\underline{p},\underline{p}'\},\overline{p}],
    \end{aligned}
    \end{equation}
    i.e., the line connecting $(\underline{p}{},(u_i-\phi_i)(\underline{p}{}))$ and $(\overline{p},(u_i-\phi_i)(\overline{p}))$ is above the one connecting $(\underline{p}',(u_i-\phi_i)(\underline{p}'))$ and $(\overline{p},(u_i-\phi_i)(\overline{p}))$ over $[\max\{\underline{p},\underline{p}'\},\overline{p}]$. 
    And similar, when $\underline{p}{}'\leq\underline{p}$, for any $p\in[\underline{p}{}',\underline{p}]$,
    \begin{equation}
    \label{ineq-4}
    U_i\big(\underline{p},\overline{p};p\big)=u_i(p)-\phi_i(p)\geq\frac{\overline{p}{}-p}{\overline{p}{}-\underline{p}{}'}[u_i(\underline{p}{}')-\phi_i(\underline{p}{}')]+\frac{p-\underline{p}{}'}{\overline{p}{}-\underline{p}{}'}[u_i(\overline{p}{})-\phi_i(\overline{p}{})].
    \end{equation}
    Instead when $\underline{p}<\underline{p}{}'$, 
    \[
    U_i\big(\underline{p}{},\overline{p};p\big)\geq u_i(p)-\phi_i(p)=U_i\big(\underline{p}',\overline{p}{}';p\big),\text{ }\forall p\in[\underline{p}{},\underline{p}'].
    \]
    Inequalities (\ref{ineq-1}), (\ref{ineq-3}), (\ref{ineq-2}), and (\ref{ineq-4}) (involving $U_i(\underline{p},\overline{p};p)$ and two lines) imply that for any $p\in[\max\{\underline{p},\underline{p}{}'\},\min\{\overline{p},\overline{p}{}'\}]$ and also $p\in[\underline{p}{}',\underline{p}]$ when $\underline{p}{}'\leq\underline{p}$ and $p\in[\overline{p},\overline{p}{}']$ when $\overline{p}\leq\overline{p}{}'$,
    \begin{align*}
        U_i\big(\underline{p}{},\overline{p};p\big)\geq\frac{\overline{p}{}'-p}{\overline{p}{}'-\underline{p}'}[u_i(\underline{p}')-\phi_i(\underline{p}')]+\frac{p-\underline{p}'}{\overline{p}{}'-\underline{p}'}[u_i(\overline{p}{}')-\phi_i(\overline{p}{}')]=U_i\big(\underline{p}',\overline{p}{}';p\big).
    \end{align*}
    The graphical intuition is as follows: Since the graph of $U_i(\underline{p},\overline{p};p)$ is above the line connecting $(\underline{p}{},(u_i-\phi_i)(\underline{p}{}))$ and $(\overline{p}{}',(u_i-\phi_i)(\overline{p}{}'))$ over $[\underline{p}{},\overline{p}{}']$ and above the one connecting $(\underline{p}',(u_i-\phi_i)(\underline{p}'))$ and $(\overline{p},(u_i-\phi_i)(\overline{p}))$ over $[\underline{p}',\overline{p}]$, it must also be above the line connecting $(\underline{p}',(u_i-\phi_i)(\underline{p}'))$ and $(\overline{p}{}',(u_i-\phi_i)(\overline{p}{}'))$ over $[\underline{p}',\overline{p}{}']$, as indicated by the last inequality above. 
    
    In sum, $U_i(\underline{p}{},\overline{p};p)\geq U_i(\underline{p}',\overline{p}{}';p)$ for all $p\in[\min\{\underline{p}{},\underline{p}{}'\},\max\{\overline{p},\overline{p}{}'\}]$.
\end{proof}

\begin{proof}[Proof of \Cref{cs-committee}]
    According to the discussion in the main text, it remains to show that for $k\leq i$ and $j\leq l$, any fixed point (or the maximum one) $(\underline{p}{}',\overline{p}{}')$ between $\underline{B}{}_{k}$ and $\overline{B}{}_l$ cannot be strictly larger than any fixed point (or the maximum one) $(\underline{p}{},\overline{p}{})$ between $\underline{B}{}_{i}$ and $\overline{B}{}_j$.
    
    By definition of one-sided best responses, we should have the following first-order condition for $\overline{B}_i(\underline{p})$:
    \[
    \overline{B}_i(\underline{p})v_i-\phi_i(\overline{B}_i(\underline{p}))-(1-\underline{p})+\phi_i(\underline{p})=[\overline{B}_i(\underline{p})-\underline{p}][v_i-\phi_i'(\overline{B}_i(\underline{p}))]
    \]
    Therefore, it holds that
    \[
    \overline{B}{}_i'(\underline{p})=\frac{v_i-\phi_i'(\overline{B}_i(\underline{p}))+1+\phi_i'(\underline{p})}{-[\overline{B}_i(\underline{p})-\underline{p}]\phi_i''(\overline{B}_i(\underline{p}))}.
    \]
    Since $\phi_i''>0$ and $\overline{B}_i\geq\underline{p}$, $\overline{B}{}_i$ must be single-dipped in $\underline{p}$. Symmetrically, $\underline{B}{}_i$ is single-peaked. In particular, 
    \[
    \underline{B}{}_i'(\overline{p})=\frac{v_i-\phi_i'(\overline{p})+1+\phi_i'(\underline{B}{}_i(\overline{p}))}{-[\overline{p}-\underline{B}{}_i(\overline{p})]\phi_i''(\underline{B}{}_i(\overline{p}))}.
    \]
    In general, the graphs of $\overline{B}_i$ and $\underline{B}{}_i$ should intersect only once, exactly at the peak of $\overline{B}_i$ which is also the bottom of $\underline{B}{}_i$. Formally, if being interior, $\underline{p}=\underline{B}{}_i(\overline{p})$ and $\overline{p}=\overline{B}_i(\underline{p})$ imply that $v_i-\phi_i'(\overline{p})=-1-\phi_i'(\underline{p})$ and thus $\overline{B}{}_i'(\underline{p})=\underline{B}{}'_i(\overline{p})=0$. Given that $\overline{B}_i$ is single-dipped and $\underline{B}{}_i$ is single-peaked, the interior intersection point is unique. As for the corner intersection point, it can only be $(\piv,\piv)$ if any. However, since $U_i$ is discontinuous at $\piv$ for all $i\ne\texttt{piv}$, it is easy to verify that $\overline{B}_i(\piv)>\piv$ for $i<\texttt{piv}$ and $\underline{B}{}_i(\piv)<\piv$ for $i>\texttt{piv}$. Hence, $(\piv,\piv)$ is a fixed point of $\overline{B}_i$ and $\underline{B}{}_i$ if and only if $i=\texttt{piv}$.
    
    Now consider the fixed points of $\overline{B}_j$ and $\underline{B}{}_i$ where $j<i$. Since $\overline{B}_j\geq\overline{B}_i$, any fixed point $(\underline{p},\overline{p})$ must be interior\footnote{Given $j<i$, if $\underline{B}{}_i(\piv)=\piv$ and thus $i\leq\texttt{piv}$, then $j<\texttt{piv}$ so $\overline{B}{}_i(\piv)>\piv$. Similarly, if $\overline{B}{}_i(\piv)=\piv$ and thus $j\geq\texttt{piv}$, then $i>\texttt{piv}$ so $\underline{B}{}_i(\piv)>\piv$. In sum, $(\piv,\piv)$ is never a fixed point.} and lie on the decreasing part of $\underline{B}{}_i$; symmetrically, they are on the increasing part of $\overline{B}_j$ at the same time. Hence, 
    \[
    \left\{(\underline{p},\overline{p}):\overline{p}=\overline{B}_j(\underline{p}),\underline{p}=\underline{B}{}_i(\overline{p})\right\}\subset\left\{(\underline{p},\overline{p}):\not\exists(\underline{p}{}',\overline{p}{}')\supsetneq(\underline{p},\overline{p})\text{ s.t. }\overline{p}{}'\leq\overline{B}_j(\underline{p}),\underline{p}{}'\geq\underline{B}_i(\overline{p})\right\}.
    \]
    For any $k\leq i$ and $l\geq j$, since $\underline{B}_k\geq\underline{B}_i$ and $\overline{B}_l\leq\overline{B}_j$, we have
    \[
    \left\{(\underline{p},\overline{p}):\overline{p}=\overline{B}_l(\underline{p}),\underline{p}=\underline{B}{}_k(\overline{p})\right\}\subset\left\{(\underline{p},\overline{p}):\overline{p}\leq\overline{B}_j(\underline{p}),\underline{p}\geq\underline{B}_i(\overline{p})\right\}.
    \]
    As a result, there does not exist a fixed point of $\overline{B}_l$ and $\underline{B}_k$ that is strictly larger than any fixed point of $\overline{B}_j$ and $\underline{B}{}_i$.

    Instead, for $j\geq i$, the fixed points of $\overline{B}_j$ and $\underline{B}{}_i$ must lie on the increasing part of $\underline{B}{}_i$ and the decreasing part of $\overline{B}_j$. In this case, the fixed points are nested, so a maximum one exists. Importantly, the maximum fixed point lies on the frontier of the graph of $\overline{B}_j$ and $\underline{B}{}_i$, as for all fixed points in the previous case. Then by similar arguments as before, as $\underline{B}_k\geq\underline{B}_i$ and $\overline{B}_l\leq\overline{B}_j$, there cannot be a fixed point of $\overline{B}_l$ and $\underline{B}_k$ that is strictly larger than the maximum one of $\overline{B}_j$ and $\underline{B}_i$.
\end{proof}

\paragraph{Equilibrium Existence}
\label{sect-cs-existence}
Here I show the existence of strong equilibrium, which according to \Cref{prop:strong}, is equivalent to the existence of a fixed point (or the maximum one) in best responses. It turns out that we have the following characterization of the fixed points based on the maximum or minimal equilibria under unilateral or unanimous stopping between two players, which always exist as we know from the two-player case. 
\begin{lemma}
    \label{existence-OSBR}
    For any two players $\ubp$ and $\lbp$ such that $\ubp\geq\lbp$, $(\underline{p},\overline{p})$ is the maximum fixed point in $\underline{B}{}_{\lbp}$ and $\overline{B}{}_{\ubp}$ if and only if it is the maximum interval equilibrium under unilateral stopping between these two players. For $\ubp<\lbp$, $(\underline{p},\overline{p})$ is a fixed point in $\underline{B}{}_{\lbp}$ and $\overline{B}{}_{\ubp}$ if and only if it is a minimal interval equilibrium under unanimous stopping between these two players.
\end{lemma}
\vspace{-4mm}
\begin{proof}
    Suppose that $\ubp\geq\lbp$ and $(\underline{p},\overline{p})$ is a fixed point in $\underline{B}{}_{\lbp}$ and $\overline{B}{}_{\ubp}$. Then $U_\lbp(\underline{p}{},\overline{p};\piv)\geq U_\lbp(\underline{p}{}',\overline{p};\piv)$ for any $\underline{p}{}'< \piv$ and $U_\ubp(\underline{p}{},\overline{p};\piv)\geq U_\ubp(\underline{p}{},\overline{p}{}';\piv)$ for any $\overline{p}{}'\geq\piv$. By single-peakedness and $\ubp\geq\lbp$, we have $U_\lbp(\underline{p}{},\overline{p};\piv)\geq U_\lbp(\underline{p},\overline{p}{}';\piv)$ for any $\overline{p}{}'\leq\overline{p}$, and $U_\ubp(\underline{p}{},\overline{p};\piv)\geq U_\ubp(\underline{p}{}',\overline{p};\piv)$ for any $\underline{p}{}'\geq\underline{p}$. Hence, by \Cref{lm:BR-extension}, for both $i=\lbp,\ubp$, $U_i(\underline{p}{},\overline{p};p)\geq U_i(\underline{p}{}',\overline{p}{}';p)$ for $p\in[\underline{p},\overline{p}]$ and $(\underline{p}{}',\overline{p}{}')\subset(\underline{p},\overline{p})$. By \Cref{concave}, $(\underline{p},\overline{p})$ is an equilibrium under unilateral stopping between $\lbp$ and $\ubp$. 
    

    Given the above result, it suffices to show that the maximum interval equilibrium under unilateral stopping is also a fixed point. Then, the maximum fixed point must coincide with the maximum equilibrium; otherwise, either the fixed point is not maximal or the equilibrium is not maximal. 
    
    Let $(\underline{p}{}_{\max},\overline{p}{}_{\max})$ denote the maximum interval equilibrium under unilateral stopping, which must contain $\piv$. Suppose that $(\underline{p}{}_{\max},\overline{p}{}_{\max})$ is not a fixed point in $\underline{B}{}_{\lbp}$ and $\overline{B}{}_{\ubp}$, so either $\underline{p}{}_{\max}\ne\underline{B}{}_\lbp(\overline{p}_{\max})=:\underline{p}$ or $\overline{p}{}_{\max}\ne\overline{B}{}_{\ubp}(\underline{p}{}_{\max})=:\overline{p}$. Suppose that $\underline{p}{}_{\max}\ne\underline{p}$, then it must be $\underline{p}{}_{\max}>\underline{p}$, since otherwise player $\lbp$ has strict incentives to deviate to $(\underline{p},\overline{p}{}_{\max})$ from $(\underline{p}{}_{\max},\overline{p}{}_{\max})$, a contradiction. Consider the sampling region $(\underline{p},\overline{p}{}_{\max})$. By definition of $\underline{p}=\underline{B}{}_\lbp(\overline{p}_{\max})$, $U_{\lbp}(\underline{p},\overline{p}{}_{\max};\piv)\geq U_{\lbp}(p,\overline{p}{}_{\max};\piv)$ for all $p\in[\underline{p},\underline{p}{}_{\max}]$. Hence, for any $p\in[\underline{p}{}_{\max},\overline{p}{}_{\max}]$,
    \begin{align*}
        &U_\lbp\left(\underline{p},\overline{p}{}_{\max};\piv\right)\geq U_\lbp\left(\underline{p}{}_{\max},\overline{p}{}_{\max};\piv\right)=V_\lbp^{IN}\left(\piv;(\underline{p}{}_{\max},\overline{p}{}_{\max})\right)\\
        \Rightarrow &U_\lbp\left(\underline{p},\overline{p}{}_{\max};p\right)\geq U_\lbp\left(\underline{p}{}_{\max},\overline{p}{}_{\max};p\right)=V_\lbp^{IN}\left(p;(\underline{p}{}_{\max},\overline{p}{}_{\max})\right)\geq u_\lbp(p)-\phi_\lbp(p),
    \end{align*}
    where the implication from $\piv$ to general $p$ is due to the linearity of $U_\lbp$ in $p$ and the fact that the two sides coincide at $p=\overline{p}{}_{\max}$. Similarly,    
    \begin{align*}
        &U_\lbp\left(\underline{p},\overline{p}{}_{\max};\piv\right)\geq U_\lbp\left(p,\overline{p}{}_{\max};p\right),\quad\forall p\in[\underline{p},\underline{p}{}_{\max}]\\
        \implies &U_\lbp\left(\underline{p},\overline{p}{}_{\max};\tilde{p}\right)\geq U_\lbp\left(p,\overline{p}{}_{\max};\tilde{p}\right),\quad\forall \tilde{p}\in[p,\underline{p}{}_{\max}],\forall p\in[\underline{p},\underline{p}{}_{\max}]\\
        \implies &U_\lbp\left(\underline{p},\overline{p}{}_{\max};p\right)\geq u_\lbp(p)-\phi_\lbp(p),\quad\forall p\in[\underline{p},\underline{p}{}_{\max}].
    \end{align*}
    Therefore, $U_\lbp(\underline{p},\overline{p}{}_{\max};p)\geq u_\lbp(p)-\phi_\lbp(p)$ for all $p\in[\underline{p},\overline{p}{}_{\max}]$, which implies $U_\lbp(\underline{p},\overline{p}{}_{\max};p)=V_\lbp^{IN}(p;(\underline{p},\overline{p}{}_{\max}))$ for all $p\in[\underline{p},\overline{p}{}_{\max}]$. Moreover, since $\underline{B}{}_{\ubp}\leq\underline{B}{}_\lbp$ and $U_{\ubp}(p,\overline{p};\piv)$ is single-peaked in $p$ over $[\underline{p},\underline{p}{}_{\max}]$, we have $U_{\ubp}(\underline{p},\overline{p}{}_{\max};\piv)\geq U_{\ubp}(p,\overline{p}{}_{\max};\piv)$ for all $p\in[\underline{p},\underline{p}{}_{\max}]$. By the same argument, it can be shown that $U_{\ubp}(\underline{p},\overline{p}{}_{\max};p)=V_{\ubp}^{IN}(p;(\underline{p},\overline{p}{}_{\max}))$ for all $p\in[\underline{p},\overline{p}{}_{\max}]$. Hence by \Cref{concave}, $(\underline{p},\overline{p}{}_{\max})$ is also an equilibrium under unilateral stopping, which contradicts to the definition of $(\underline{p}{}_{\max},\overline{p}{}_{\max})$. Hence, we must have $\underline{p}{}_{\max}=\underline{p}=\underline{B}{}_\lbp(\overline{p}_{\max})$. Symmetrically, one can also prove $\overline{p}{}_{\max}=\overline{p}=\overline{B}{}_{\ubp}(\underline{p}{}_{\max})$. In conclusion,  $(\underline{p}{}_{\max},\overline{p}{}_{\max})$ is a fixed point.

The proof for the case with $\ubp<\lbp$ and minimal equilibria is similar and omitted for brevity.
\end{proof}

\vspace{-4mm}

\subsection{Poisson Learning}
\label{Poisson}
In this section, I consider Poisson learning instead of incremental learning assumed in the baseline model and show that the analysis can be adapted. For simplicity, I focus on unilateral and unanimous stopping with two players. Suppose that information arrives in the form of Poisson signals. The arrival rate of the signal is $(a_{1,t}\wedge a_{2,t})\lambda\mathbbm{1}_{\theta=1}$ under unilateral stopping, and $(a_{1,t}\vee a_{2,t})\lambda\mathbbm{1}_{\theta=1}$ under unanimous stopping. Therefore, this Poisson signal is conclusive for $\theta=1$. Upon arrival, the belief jumps to $p_t=1$; upon no breakthrough, the belief dynamic is
\[
\mathrm{d}p_t=-\lambda p_t(1-p_t)\mathrm{d}t.
\]

Analogous to \Cref{cost} in the baseline model, it can be shown that
\[
\mathbb{E}\left[\int_0^\tau \cost_i(p_s)\mathrm{d}s\right]
=\int_0^{p_0}[\varphi_i(p_0)-\varphi_i(q)]\mathrm{d}F_\tau(q),
\]
\[
\text{where}\quad \varphi_i(q)=\int_z^q\frac{\cost_i(y)}{\lambda y(1-y)}\mathrm{d}y\quad\text{for some }z\in(0,1).
\]
Therefore, the problem under unilateral stopping can be reformulated as
\begin{align*}
	&\max_{H} \int_0^1\left[u_i(p)-\mathbbm{1}_{p\leq p_0}\left(\varphi_i(p_0)-\varphi_i(p)\right)\right]\mathrm{d}H(p)\\
	&\text{s.t. }H \text{ is an MPC of }F,\text{ } H \text{ is constant on }[p_0,1) \text{ and } \mathbb{E}_H[p]=p_0.
\end{align*}
And the problem under unanimous stopping is given by
\begin{align*}
&\max_{H} \int_0^1\left[u_i(p)-\mathbbm{1}_{p\leq p_0}\left(\varphi_i(p_0)-\varphi_i(p)\right)\right]\mathrm{d}H(p)\\
&\text{s.t. }H \text{ is an MPS of }F,\text{ } H \text{ is constant on }[p_0,1)\cup\sprg \text{ and } \mathbb{E}_H[p]=p_0.
\end{align*}

These problems are different from those in the baseline model, but we are still able to use the concavification method to characterize their binary policy solutions (as in \Cref{concave}; here one support must be 1).

\subsection{\citeauthor{gul2012war}'s (\citeyear{gul2012war}) War of Information}
\label{GP12}
This section discusses how my model can also be applied to \citeauthor{gul2012war}'s (\citeyear{gul2012war}) setting of competition in persuasion.

In their setup, there are two parties, $n=2$, with indirect payoffs $u_1(p)=\mathbbm{1}_{p>1/2}$ and $u_2(p)=1-u_1(p)$. These indirect payoffs come from a voter trying to pick the better party.\footnote{In details, $\theta=1$ (resp., $\theta=0$) refers to Party 1's being better (resp., worse) than Party 2. A voter chooses Party 1 when her belief about $\theta=1$ is above 1/2, and Party 2 otherwise. Party $i$ obtains a payoff of 1 if it is chosen and 0 otherwise.} Two parties provide costly information to persuade the voter according to the belief process $(p_t)_{t\geq0}$. For simplicity, let $p_0=1/2$.

Slightly different from my baseline model with unanimous stopping, \citet{gul2012war} assume that Party 1 (resp., Party 2) can only provide information with constant flow cost $\cost_1$ (resp., $\cost_2$) when it falls behind, i.e., when $p_t\leq1/2$ (resp., $p_t>1/2$). This difference calls for two modifications. First, since parties do not provide information when they are leading, they do not bear sampling costs in that case. Therefore, $\cost_1(p)=\cost_1\cdot\mathbbm{1}_{p\leq1/2}$ and $\cost_2(p)=\cost_2\cdot\mathbbm{1}_{p>1/2}$, and accordingly, $\phi_1(p)=0$ for $p>1/2$ and $\phi_2(p)=0$ for $p\leq1/2$. 

Second, the collective stopping rule changes: When $p_t\leq1/2$, the persuasion process ends if and only if Party 1 stops; when $p_t>1/2$, the game ends if and only if Party 2 stops. Interestingly, given the above payoff structure, the leading party has no incentive to continue providing information even if it can, so this new stopping rule is in fact strategically equivalent to unanimous stopping.\footnote{As will become clear below, the equilibrium under this new stopping rule is also an equilibrium under unanimous stopping with ``one-sided best responses.'' 
}

\begin{figure}[!t]
    \centering
    \includegraphics[scale=0.35]{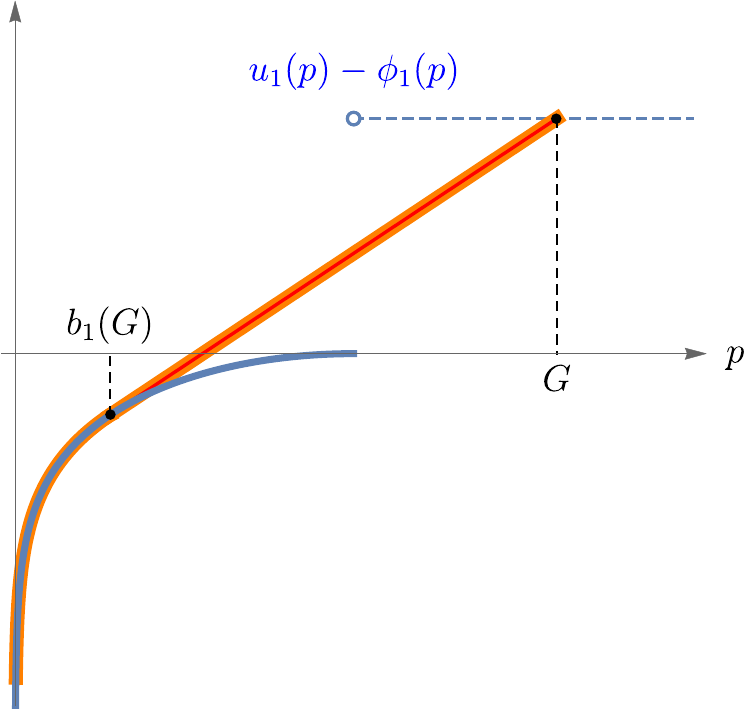}
    \caption{Party 1's Problem: Given $G$, choosing $g=b_1(G)$.}
    \label{GP}
\end{figure}

Under this new stopping rule, Party 1 (resp., Party 2) decides when to stop providing information when $p_t\leq1/2$ (resp., $p_t>1/2$). Focusing on PSMPE, Party 1's (resp., Party 2's) strategy can be represented by a lower bound $g\leq1/2$ (resp., an upper bound $G>1/2$) at which it stops. 
According to the equilibrium characterization by concavification, given $G$, Party 1 chooses $g\leq1/2$ to maximize $\frac{G-1/2}{G-g}\cdot[0-\phi_1(g)]+\frac{1/2-g}{G-g}\cdot1$, where the best response is denoted by $b_1(G)$; given $g$, Party 2 chooses $G>1/2$ to maximize $\frac{G-1/2}{G-g}\cdot1+\frac{1/2-g}{G-g}\cdot[0-\phi_2(G)]$, where the best response is denoted by $B_2(g)$. See \Cref{GP} for a graphical illustration. When $\cost_i>0$, there exists a unique solution $(g^*,G^*)$ such that $g^*=b_1(G^*)$ and $G^*=B_2(g^*)$. 

The solution $(g^*,G^*)$ corresponds to the unique equilibrium in the baseline model of \citet{gul2012war}. Parties' actions are strategic substitutes: both $b_1$ and $B_2$ are increasing functions. As a result, decreasing one party's cost makes that party provide more information and its opponent provide less.

\end{document}